\documentclass[10pt]{amsart}

\usepackage{enumerate}
\usepackage{amssymb}

\newtheorem{theorem}{Theorem}[section]
\newtheorem{lemma}[theorem]{Lemma}
\newtheorem{corollary}[theorem]{Corollary}
\newtheorem{proposition}[theorem]{Proposition}
\newtheorem{claim}[theorem]{Claim}

\newtheorem*{theorema}{Theorem A}
\newtheorem*{theoremb}{Theorem B}

\theoremstyle{definition}
\newtheorem{definition}[theorem]{Definition}
\newtheorem{example}[theorem]{Example}

\theoremstyle{remark}
\newtheorem{remark}[theorem]{Remark}

\numberwithin{subsection}{section}
\numberwithin{equation}{section}

\newcommand\cA{{\mathcal{A}}}
\newcommand\cC{{\mathcal{C}}}
\newcommand\cM{{\mathcal{M}}}
\newcommand\cP{{\mathcal{P}}}
\newcommand\bR{{\mathbb{R}}}

\newcommand\oH{{\overline{H_2}}}
\newcommand\oHp{{\overline{H_2^+}}}
\newcommand\oi{{\widehat{\infty}}}
\newcommand\oR{{\overline{\mathbb{R}}}}

\newcommand\pp{{\Pi_{P,Q}^+}}
\newcommand\ppp{{\Pi_{P_1,P_2}^+}}

\makeindex

\begin{document}


\title{Optimal version of the fundamental theorem of chronogeometry}


\author{Michiya Mori}
\address{Graduate School of Mathematical Sciences, The University of Tokyo, 3-8-1 Komaba, Meguro-ku, Tokyo, 153-8914, Japan; Center for Interdisciplinary Theoretical and Mathematical Sciences (iTHEMS), RIKEN, 2-1 Hirosawa, Wako, Saitama, 351-0198, Japan}
\curraddr{}
\email{mmori@ms.u-tokyo.ac.jp}
\thanks{The first author was supported by JSPS KAKENHI Grant Number 22K13934.}

\author{Peter \v{S}emrl}
\address{Institute of Mathematics, Physics, and Mechanics, Jadranska 19, SI-1000 Ljubljana, Slovenia; Faculty of Mathematics and Physics, University of Ljubljana, Jadranska 19, SI-1000 Ljubljana, Slovenia}
\curraddr{}
\email{peter.semrl@fmf.uni-lj.si}
\thanks{The second author was supported by grants J1-2454 and P1-0288 from ARIS, Slovenia.}

\date{}

\subjclass[2020]{Primary 83A05, Secondary 47B49, 51B20, 53C50.}

\keywords{fundamental theorem of chronogeometry; special relativity; coherency preserving mapping}


\begin{abstract}
We study lightlikeness preserving mappings from the $4$-dimensional Minkowski spacetime $\mathcal{M}_4$ to itself under no additional regularity assumptions like continuity, surjectivity, or injectivity. 
We prove that such a mapping $\phi$ satisfies one of the following three conditions.
\begin{enumerate}
\item The mapping $\phi$ can be written as a composition of a Lorentz transformation, a multiplication by a positive scalar, and a translation. 
\item There is an event $r\in \mathcal{M}_4$ such that $\phi(\mathcal{M}_4\setminus\{r\})$ is contained in one light cone. 
\item There is a lightlike line $\ell$ such that $\phi(\mathcal{M}_4\setminus \ell)$ is contained in another lightlike line. 
\end{enumerate}
Here, a line that is contained in some light cone in $\mathcal{M}_4$ is called a lightlike line.
We also give several similar results on mappings defined on a certain subset of $\mathcal{M}_4$ or the compactification of $\mathcal{M}_4$. 
\end{abstract}

\maketitle

\tableofcontents




\part{Introduction and preliminaries}

\section{Introduction}

\subsection{Concise description of our result}
In this paper, we study the standard Minkowski spacetime \index{Minkowski spacetime} $\mathcal{M}_4$ \index{$M_4$@$\mathcal{M}_4$} of dimension $3+1$.
In the mathematical foundations of special relativity, we adopt the harmless normalization that the speed of light equals $1$.
Recall that two spacetime events $r_1 = (x_1 , y_1 , z_1 , t_1), r_2 = (x_2, y_2, z_2 , t_2 ) \in \mathcal{M}_4$ are \emph{lightlike} \index{lightlike} if
\[
(x_2 - x_1)^2 + (y_2 - y_1)^2 + (z_2 - z_1)^2 = (t_2 - t_1)^2 .
\]
Hence, $r_1$ and $r_2$ are lightlike if the light signal can pass between $r_1$ and $r_2$.

The Lorentz--Minkowski indefinite inner product \index{Lorentz--Minkowski indefinite inner product} on $\mathcal{M}_4$ is defined by \index{$\langle\cdot, \cdot\rangle$}
\[
\langle r_1 , r_2 \rangle = -x_1 x_2 - y_1 y_2  -z_1 z_2  + t_1 t_2 
\]
for a pair of spacetime events $r_1 = (x_1 , y_1 , z_1 , t_1)$, $r_2 = (x_2, y_2, z_2 , t_2 ) \in \mathcal{M}_4$.
Thus, $r_1 , r_2$ are lightlike if and only if 
\[
\langle r_1 - r_2 , r_1 - r_2 \rangle = 0  .
\]
For a given spacetime event $r$, the set of all spacetime events $s$ satisfying $\langle s - r , s - r \rangle = 0$, that is, the set of all spacetime events $s$ such that $s$ and $r$ are lightlike,
is called the \emph{light cone} \index{light cone} with vertex $r$.

Recall that a \emph{Lorentz matrix}\index{Lorentz matrix} is a $4\times 4$ real matrix $Q$ satisfying $Q^t MQ = M$, where
\begin{equation}\label{Matrix}
M = \left[ \begin{matrix} -1 & 0 & 0 & 0 \cr
0 & -1 & 0 & 0 \cr
0 & 0 & -1 & 0 \cr
0 & 0 & 0 & 1 \cr\end{matrix} \right] \index{$M$}
\end{equation}
and $Q^t$ denotes the transpose of $Q$.
A mapping on $\cM_4$ of the form $r\mapsto Qr$ for some Lorentz matrix $Q$ is called a \emph{Lorentz transformation}\index{Lorentz transformation}.  
A Lorentz transformation is characterized as a linear mapping $\phi\colon \cM_4\to \cM_4$ satisfying 
\[
\langle \phi(r_1), \phi(r_2) \rangle = \langle r_1 , r_2 \rangle
\] 
for every pair $r_1,r_2\in \cM_4$.
A mapping on $\cM_4$ of the form $r\mapsto Qr+a$ for some Lorentz matrix $Q$ and a spacetime event $a$ is called a \emph{Poincar\'e transformation}\index{Poincar\'e transformation}.

A map $\phi \colon \mathcal{M}_4 \to \mathcal{M}_4$ is said to \emph{preserve lightlikeness in both directions} \index{preserve lightlikeness in both directions} if it satisfies
\begin{equation}\label{ggg}
\langle r_1 - r_2 , r_1 - r_2 \rangle = 0 \iff \langle \phi (r_1 ) - \phi ( r_2 ) , \phi ( r_1 ) - \phi ( r_2 ) \rangle = 0
\end{equation}
for every pair of spacetime events $r_1, r_2 \in \mathcal{M}_4$.
Observe that a bijection $\phi \colon \mathcal{M}_4 \to \mathcal{M}_4$ preserves lightlikeness in both directions if and only if the light cone with vertex $r$ is mapped by $\phi$ onto the light cone with vertex $\phi(r)$ for every $r\in \cM_4$.
The following theorem is known as the \emph{fundamental theorem of chronogeometry}.
\begin{theorem}\index{fundamental theorem of chronogeometry}\label{ftc}
Every bijective map $\phi \colon \mathcal{M}_4 \to \mathcal{M}_4$ satisfying \eqref{ggg} for every pair of spacetime events $r_1, r_2 \in \mathcal{M}_4$, is of the form
\begin{equation}\label{fiesa}
\phi (r) =   c\,Qr + a,\ \ r\in \cM_4,
\end{equation}
for some positive real number $c$, some Lorentz matrix $Q$, and some spacetime event $a$.
\end{theorem}
Note that linearity or continuity is not assumed in this theorem.
Conversely, it is easily seen that a mapping of the form \eqref{fiesa} preserves lightlikeness in both directions.
According to \cite{Al3}, the fundamental theorem of chronogeometry was first given by Alexandrov in 1949 \cite{Al1} (see also \cite{Al2}, \cite{Ben}, \cite{Zee}).
For a physical interpretation of this theorem, we refer to \cite[p.691]{Pfe}. 

The fundamental theorem of chronogeometry has been improved in many ways. The same conclusion as in Theorem \ref{ftc} has been obtained under some weaker assumptions, see \cite{BoH}, \cite{ChP}, \cite{Les}, \cite{LeM}, \cite{Pfe}, \cite{PoR}, \cite{Sch}, and \cite{Zee}. Our aim is to optimize the fundamental theorem of chronogeometry.
We will assume no regularity conditions like injectivity or surjectivity. 
The assumption of preserving lightlikeness in both directions will be replaced by the following weaker condition.
A map $\phi \colon \mathcal{M}_4 \to \mathcal{M}_4$ is said to \emph{preserve lightlikeness} \index{preserve lightlikeness} (or more precisely, it preserves lightlikeness in one direction) if 
for every pair of spacetime events $r_1, r_2 \in \mathcal{M}_4$ we have
\[
\langle r_1 - r_2 , r_1 - r_2 \rangle = 0 \Rightarrow \langle \phi (r_1 ) - \phi ( r_2 ) , \phi ( r_1 ) - \phi ( r_2 ) \rangle = 0.
\]
In other words, a mapping $\phi \colon \mathcal{M}_4 \to \mathcal{M}_4$ preserves lightlikeness when the light cone with vertex $r$ is mapped by $\phi$ to a (possibly proper) subset of the light cone with vertex $\phi(r)$ for every $r\in \cM_4$.
We study a general mapping $\phi \colon \mathcal{M}_4 \to \mathcal{M}_4$ that preserves lightlikeness in one direction only.
Under such a weak assumption, not all lightlikeness preserving maps are of the form \eqref{fiesa}. 

Now let us state one consequence of our main results. 
A \emph{lightlike line} \index{lightlike line} is a subset of $\mathcal{M}_4$ of the form 
\[
\{(x_0+tx, y_0+ty, z_0+tz, t)\,:\, t\in \mathbb{R}\}
\]
for some $(x_0, y_0, z_0), (x,y,z) \in \mathbb{R}^3$ with $x^2+y^2+z^2=1$. 
Observe that this is a subset of the light cone with vertex $(x_0,y_0, z_0, 0)\in \cM_4$.
\begin{theorem}\label{MM}
Let $\phi\colon \mathcal{M}_4\to \mathcal{M}_4$ preserve lightlikeness in one direction. 
Then one of the following holds.
\begin{enumerate}
\item The mapping is of the form \eqref{fiesa} for some positive real number $c$, some Lorentz matrix $Q$, and some spacetime event $a$. 
\item There are events $r, r'\in \mathcal{M}_4$ such that $\phi(\mathcal{M}_4\setminus\{r\})$ is contained in the light cone with vertex $r'$. 
\item There are lightlike lines $\ell, \ell'$ in $\mathcal{M}_4$ such that $\phi(\mathcal{M}_4\setminus \ell)$ is contained in $\ell'$. 
\end{enumerate}
\end{theorem}
In fact, we will give a result that applies to a mapping defined on a more general subset of $\mathcal{M}_4$.

\subsection{The Minkowski spacetime versus the space of hermitian matrices}\label{vs}
There is another source of motivation from matrix theory to think of lightlikeness preserving mappings.
With $H_2$ \index{$H_2$} we denote the set of all $2 \times 2$ complex hermitian matrices.
To each spacetime event $r = (x,y,z,t) \in \mathcal{M}_4$, we associate a $2\times 2$ hermitian matrix
\begin{equation}\label{zacet}\index{$xi$@$\xi$}
\xi(r)=A= \left[ \begin{matrix}  t-z & x+iy \cr x-iy & t+z \cr\end{matrix} \right]\in H_2. 
\end{equation}
A straightforward computation shows that the spacetime
events $r_1 , r_2$ are lightlike if and only if the associated matrices $A_1$ and $A_2$ satisfy
\[
\det (A_2 - A_1) = 0.
\]
For a $2\times 2$ matrix, its determinant is zero if and only if it is either the zero matrix or a matrix of rank one. Recall that two matrices
$A_1 , A_2$ are said to be \emph{coherent}\index{coherent}, $A_1 \sim A_2$, \index{$\sim$1@$\sim$ (on $H_2$)} if
\[
\mathrm{rank}\, (A_2 - A_1 ) \le 1.
\]
Hence, two spacetime
events $r_1 , r_2$ are lightlike if and only if the associated matrices $A_1$ and $A_2$ are coherent.
There is a vast literature on mappings on the space of matrices that preserve the coherency relation or some similar relation that involves rank. 
See for example Hua's series of work \cite{Hu1,Hu2,Hu3,Hu4,Hu5,Hu6,Hu7,Hu8} and related results \cite{HuS}, \cite{SS}, \cite{Se3}. See also the survey article \cite{LP}. 
This kind of results have many applications, for example in the geometry of algebraic homogeneous spaces, see \cite{Cho}, and in the study of symmetries of certain quantum structures, see \cite{Mol}, \cite{Se1}, \cite{Se2}, and the references therein.

By considering a mapping satisfying the assumption in Theorem \ref{MM}, we get a mapping $\varphi\colon H_2\to H_2$ satisfying
\[
\mathrm{rank}\, (A - B ) \le 1 \Rightarrow \mathrm{rank}\, (\varphi(A) - \varphi(B))\leq 1,
\]
or equivalently, 
\begin{equation}\label{det}
\det (A - B ) = 0 \Rightarrow \det(\varphi(A) - \varphi(B))=0
\end{equation}
for every pair $A, B\in H_2$.
It is easily seen that this is also equivalent to the condition
\[
\mathrm{rank}\, (\varphi(A) - \varphi(B))\leq \mathrm{rank}\, (A - B )
\]
for every pair $A, B\in H_2$.
Therefore, by considering $H_2$ instead of $\mathcal{M}_4$, we get at least three formulations of this condition. 

This paper is basically formulated in terms of $2\times 2$ hermitian matrices instead of the Minkowski spacetime.
We avoid using results on the Minkowski spacetime like Theorem \ref{ftc}.
Therefore, the reader needs no prerequisite knowledge about the geometry of Minkowski spacetime. 
Instead, throughout the paper we freely use facts about $2\times 2$ matrices. We believe that in such a way the results are more accessible for the general audience.

\subsection{Structure of the paper} 
In the next section, we will introduce the space $\overline{H_2}\supset H_2$. 
We extend the coherency relation $\sim$ on $H_2$ to $\oH$, and show that the space $\overline{H_2}$ endowed with the coherency relation can be identified with the $2\times 2$ unitary group $U_2$ via the Cayley transform. 
We also show that $\oH$ can be identified with $\overline{\mathcal{M}_4}$, which appears in the literature as a compactification of $\mathcal{M}_4$.  

Let $\mathcal{A}\subset \oH$. 
A mapping $\varphi\colon \mathcal{A}\to \oH$ is called a \emph{coherency preserver}\index{coherency preserver}, or it preserves coherency (in one direction), if it satisfies 
\[
A\sim B\Rightarrow \varphi(A) \sim \varphi(B)
\]
for every pair $A, B\in \mathcal{A}$.
Instead of lightlikeness preserving mappings on $\mathcal{M}_4$, we will think of the equivalent problem of coherency preserving mappings on $H_2$.
In fact, we will study mappings defined on a subset of $\oH$ rather than on $H_2$.
An advantage of working with $\oH$ (or $\overline{\mathcal{M}_4}$) rather than $H_2$ (or $\mathcal{M}_4$) derives from the fact that there are more symmetries in $\oH$ than in $H_2$. 

Section \ref{prelim} collects basic concepts concerning the coherency relation of $\oH$. 
In Subsection \ref{automorphism} we introduce the concept of an automorphism, that is, a bijective mapping of $\oH$ that preserves coherency in both directions. 
We define two special classes of automorphisms: affine automorphisms and standard automorphisms (Definition \ref{automorphisms}). 
It is known that an automorphism is always standard. For the sake of completeness, we give a proof of this fact in Subsection \ref{apply}.
We remark here that the concept of standard automorphism coincides with that of conformal transformation in the literature. \index{conformal transformation}

We then introduce basic notions like projections (together with its relation to the Bloch representation, Subsection \ref{projectionsx}), lines (Subsection \ref{line}), and surfaces (Subsection \ref{ssss}). We give some properties of such sets. We also study the relative position of three points in $\oH$ (Subsection \ref{three points}).
In the last subsection of that section (Subsection \ref{identity-type theorem}), we collect results about the identity-type theorem in our setting.

The main part of our paper starts in Section \ref{standard}.
We study coherency preserving mappings from a subset of $\oH$ to $\oH$. 
Such a map is said to be standard if it extends to a standard automorphism of $\oH$.
The first main result is Theorem A, which gives a sufficient condition for a coherency preserving mapping to be standard. 
To demonstrate the potential of this theorem, we show that the classical version of the fundamental theorem of chronogeometry can be obtained easily from Theorem A (Subsection \ref{apply}).

In Section \ref{non}, we first introduce two important types of coherency preserving maps (Definitions \ref{degeneratefirst}, \ref{degeneratesecond}) in correspondence with the latter two items in Theorem \ref{MM}. 
The key result is Theorem B.
Applications of Theorem B are given in Subsection \ref{applici}.
We prove that every coherency preserver from $\mathcal{U}$ into $\oH$ is either standard or of one of the two types if $\mathcal{U}$ is either a matrix interval in $H_2$ or the whole space $\oH$. 
This together with the identity-type theorem in Subsection \ref{identity-type theorem} shows that a coherency preserver defined on an open connected subset of $\oH$ is either standard or locally degenerate in a certain sense (Theorem \ref{locally}).
We also demonstrate that Theorem \ref{MM} and the main result in Lester's article \cite{Les} can be obtained easily from our theorems.
In Subsection \ref{concrete}, we give a more concrete description of coherency preservers of the two types. In particular, we study the case where the domain is either $H_2$ or $\oH$.
A rather long proof of Theorem B is given in the last section. 
The proof is split into three cases.

\subsection{Further research directions} 
Before closing the current section, let us mention the possibility of generalizing our work. 
In this paper, we study everything in the setting of $4$-dimensional Minkowski spacetime or its compactification. 
Most of our discussion, after translating the concepts in a suitable manner, is valid for a mapping from (a subset of the compactification of) the Minkowski spacetime of arbitrary dimension $\geq 5$ to itself, even though the picture in terms of hermitian matrices does not make sense.
In the $3$-dimensional case there is at least one point where the argument in the 4-dimensional space cannot be modified easily. 
More precisely, it is Lemma \ref{134}.
We do not know whether the same conclusions as in the 4-dimensional case hold in the 3-dimensional case.

There are further possible directions of research that seem highly challenging. 
For example, what happens if we consider mappings on a space endowed with a more general symmetric bilinear form instead of the Lorentz--Minkowski indefinite inner product?
Is it possible to give a more general result for coherency preserving mappings on the space of hermitian matrices of an arbitrary size?
How about mappings satisfying \eqref{det}?


\section{The compactification $\oH$ of $H_2$}\label{compactification}
Whenever appropriate, matrices will be identified with linear operators.
We use the symbol $I$ \index{$I$} for the unit matrix and $0$ for the zero matrix.
For any complex matrix $A$, we denote by $A^t$ \index{$A^t$} the transpose of $A$, by $A^\ast$ \index{$A^\ast$} the conjugate transpose of $A$, and by $\mathrm{tr}\, A$ \index{$trA$@$\mathrm{tr}\, A$} the trace of $A$. Let $i,j$ be integers,
$1 \le i,j \le 2$. By $E_{ij}$ \index{$E_{ij}$} we denote the $2 \times 2$ matrix whose all entries are zero but the $(i,j)$-entry which is equal to $1$.
Vectors in $\mathbb{C}^2$ will be represented by $2 \times 1$ complex matrices. Every $2 \times 2$ complex matrix of rank one is of the form $xy^\ast$ for some nonzero vectors
$x,y \in \mathbb{C}^2$. 
A \emph{projection} \index{projection} is a matrix $P$ satisfying $P=P^2=P^*$.
If $x = y$ is a vector of norm one, then $xx^\ast$ is a projection of rank one, and every projection of rank one is of this form. Let $\{ e_1 ,  e_2 \}$ be the standard basis of $\mathbb{C}^2$\index{$e_i$}. Then $E_{ij} = e_i e_{j}^\ast$.

With $\le$ \index{$\le$} we denote the usual partial order (Loewner order\index{Loewner order}) on $H_2$. That is, $A\leq B$ means that both eigenvalues of $B-A$ are at least $0$.
For $A,B \in H_2$ we will write $A < B$ \index{$<$} if $B-A$ is a positive invertible matrix, that is, both eigenvalues of $B-A$ are positive.
Note that the order relation $A<B$ can be interpreted as the following situation in special relativity under the identification via the mapping $\xi$ as in \eqref{zacet}: 
The spacetime event corresponding to $A$ is in the past of that corresponding to $B$.

For a $2 \times 2$ matrix $A$,  $\sigma(A)$ \index{$sigmaA$@$\sigma(A)$} denotes its spectrum, that is, the set of eigenvalues of $A$.
Let $\cP$ \index{$P$@$\mathcal{P}$} denote the collection of all rank one projections in $H_2$.
For $P\in \mathcal{P}$, we write $P^\perp:= I-P\in \mathcal{P}$\index{$P^\perp$}. 
Let $\oR=\mathbb{R}\cup\{\infty\}$\index{$\infty$}\index{$R$@$\oR$} denote the one-point compactification of $\mathbb{R}$. Later in this paper, we will frequently use the symbols $(a,\infty] :=(a, \infty)\cup \{\infty\}$\index{$a$1@$(a,\infty]$}, $[a,\infty] :=[a, \infty)\cup \{\infty\}\subset \oR$ \index{$a$2@$[a,\infty]$} for $a\in \mathbb{R}$.

In this section, we introduce a space $\oH\supset H_2$ in a somewhat intuitive manner that involves only $2\times 2$ hermitian matrices. 
Then we endow the space $\oH$ with the coherency relation $\sim$ that extends the usual coherency relation in $H_2$.
We will show that the space $\oH$ endowed with the coherency relation can be identified via the Cayley transform with the $2\times 2$ unitary group $U_2$, and also with the compactification $\overline{\mathcal{M}_4}$ of $\mathcal{M}_4$.

Before we proceed, let us give one easy lemma. 
Let $H_{2}^{++}$\index{$H_2^{++}$} denote the set of all positive invertible matrices in $H_2$, and let $H_{2}^{--}$\index{$H_2^{--}$} denote the set of negative invertible matrices in $H_2$.
Let $H_{2}^{+-}$ \index{$H_2^{+-}$} denote the set of matrices in $H_2$ having one positive eigenvalue and one negative eigenvalue. 

\begin{lemma}\label{apbq}
Let $P,Q\in \cP$ satisfy $P\neq Q$. Let $a,b\in \mathbb{R}$. 
\begin{itemize}
\item If $a,b>0$, then $aP+bQ\in H_2^{++}$. 
\item If $a,b<0$, then $aP+bQ\in H_2^{--}$. 
\item If $ab<0$, then $aP+bQ\in H_2^{+-}$. 
\end{itemize}
More generally, if $a,b$ are nonzero complex numbers, then $aP+bQ$ is an invertible matrix.
\end{lemma}
\begin{proof}
There is no loss of generality in assuming that $P=E_{11}$. 
Then
\[
Q=\left[ \begin{matrix} c & e^{it}\sqrt{c-c^2}  \cr e^{-it}\sqrt{c-c^2} & 1-c \cr \end{matrix} \right]
\]
for some $0\leq c<1$ and some $t\in [0,2\pi)$. 
Now the verification of the statement is easily done by the equation
\[
\det(aE_{11}+bQ) = (a+bc)\cdot b(1-c)-b^2(c-c^2) = ab(1-c)
\]
and by looking at the $(1,1)$-entry of $aE_{11}+bQ$.
\end{proof}
\subsection{Definition of $\oH$}
We consider the collection $\oH$ \index{$H$@$\oH$} of all formal sums $aP+bP^{\perp}$ for $P\in \cP$ and $a, b\in \oR$, with the following rules:
\begin{itemize}
\item  $aP+bP^\perp=bP^\perp+aP$ for any $a, b\in {\oR}$ and any $P\in \cP$,  
\item $a P+a P^{\perp}=a Q+a Q^{\perp}$ for any $a\in {\oR}$ and any $P, Q\in \cP$. 
\end{itemize}
Then it is clear that the set $H_2$ embeds into $\oH$ in a natural manner.
Using this embedding, we will always regard elements of $H_2$ as elements of $\oH$.
Therefore, for $a\in \mathbb{R}$ and $P, Q\in \cP$, we have $a P = a P +0P^\perp$ and $aI =aP+aP^\perp =aQ+aQ^\perp$.
We also use the symbols ${\oi}:=\infty P+\infty P^{\perp}=\infty Q+\infty Q^{\perp}$\index{$\infty$2@$\oi$} and $\infty P := \infty P +0P^\perp$\index{$Pinfty$@$\infty P$}.
We extend the coherency relation $\sim$ on $H_2$ to $\oH$ in the following way\index{$\sim$2@$\sim$ (on $\oH$)}: 
\begin{itemize}
\item For $A, B\in H_2$, we define $\sim$ as before, i.e., $A\sim B \iff \mathrm{rank}\, (B - A ) \le 1$. 
\item If $a\in \bR$, $A\in H_2$, and $P\in\cP$, then $\infty P+a P^{\perp}\sim A$ if and only if $P^{\perp} AP^{\perp}=aP^{\perp}$ (note that the condition $P^{\perp} AP^{\perp}=aP^{\perp}$ is equivalent to $\mathrm{tr}\, (P^\perp A) = a$).
\item If $a, b\in \bR$ and $P, Q\in \cP$, then $\infty P+a P^{\perp}\sim \infty Q + bQ^{\perp}$ if and only if $P=Q$. 
\item ${\oi} \sim \infty P+aP^{\perp}$ for every $a\in {\oR}$ and every $P\in \cP$, and ${\oi} \not\sim A$ for every $A\in H_2$.
\end{itemize}
For $A, B\in \oH$, we define $d(A, A)=0$, $d(A, B)=1$ if $A\neq B\sim A$, and $d(A, B)=2$ if $A\not\sim B$\index{$d$}. It is easily seen that $d$ satisfies the axioms of distance. 

For $A\in \oH$, the collection $\cC_A$ \index{$C(A)$@$\cC_A$} of all elements $B\in \oH$ satisfying $A\sim B$ is called the \emph{cone} \index{cone} with vertex $A$. The following are easy to verify.
\begin{itemize}
\item If $A\in H_2$, then 
\[
\cC_A=\{B\in H_2\,:\, A\sim B\} \cup \{\infty P+\mathrm{tr}\, (P^\perp A) P^\perp\,:\, P\in \cP\}.
\]
\item If $P\in \cP$ and $a\in \bR$, then 
\[
\cC_{\infty P+ aP^{\perp}}=\{\infty P+ bP^{\perp}\,:\, b\in {\oR}\} \cup \{B\in H_2\,:\, P^{\perp}BP^{\perp}=aP^\perp\}.
\] 
\item $\cC_{\oi} =\oH\setminus H_2$.
\end{itemize}

\subsection{The relation between $\oH$ and $U_2$}\label{oHU2}
In order to demonstrate that the above definition of $\sim$ is natural, we will identify $\oH$ with $U_2$\index{$U_2$}, the group of all $2 \times 2$ unitary matrices. The identification is given by the Cayley transform \index{Cayley transform}
\begin{equation}\label{cayley}
f(aP+bP^{\perp}) = f(a)P+f(b)P^{\perp},
\end{equation}
 where $f\colon {\oR}\to\mathbb{T}$ is determined by 
\begin{equation}\label{cayleyf}
 f(t)=\frac{t-i}{t+i}, \ \ \ t\in \mathbb{R},\ \ \  \text{and}\ \ \ f(\infty)=1.
\end{equation}
Here, $\mathbb{T}$ \index{$T$@$\mathbb{T}$} denotes the group of all complex numbers of modulus one.
See also Uhlmann's article \cite{Uh}.

The coherency relation on $U_2$ is defined in the same way as on the set of hermitian matrices, that is, unitary matrices $U$ and $V$ are coherent if and only if $\mathrm{rank}\, (U-V) \le 1$. 
More generally, for an arbitrary pair of $2\times 2$ complex matrices $A, B$, we say $A$ and $B$ are coherent and write $A\sim B$ \index{$\sim$3@$\sim$ (for a pair of $2\times 2$ matrices)} when $\mathrm{rank}\, (A-B) \le 1$.

\begin{lemma}\label{tr1}
Let $A$ be a $2\times 2$ complex matrix of rank at most one.
Then $A\sim I$ holds if and only if $\mathrm{tr}\, A=1$. 
\end{lemma}
\begin{proof}
By choosing a suitable basis of $\mathbb{C}^2$, one may assume that $A$ is of the form
\[
\left[ \begin{matrix} \mathrm{tr}\,A  & 0  \cr * & 0 \cr \end{matrix} \right]. 
\]
It is now easy to get the desired conclusion.
\end{proof}
It is clear that the map $A \mapsto f(A)$, $A \in \oH$, is a bijection of $\oH$ onto $U_2$. 
Moreover, the image $f(H_2)$ equals 
the set of all unitary matrices $U \in U_2$ with the property that $1 \not\in \sigma (U)$.

We will show that for any pair $A,B \in \oH$ we have $A \sim B$ if and only if $f(A) \sim f(B)$. We first note that
\begin{equation}\label{mmmm}
{t-i \over t+i} = 1 - {2i \over t+i}.
\end{equation}
Therefore, for any  $A,B \in H_2$ we have
\[
\begin{split}
f(A) - f(B) &= \left( I - 2i(A+iI)^{-1} \right) -  \left( I-2i(B+iI)^{-1} \right)\\
&=-2i(A+iI)^{-1} + 2i(B+iI)^{-1} \\
&= 2i(A+iI)^{-1} ((A+iI)-(B+iI))(B+iI)^{-1}\\
&= 2i(A+iI)^{-1} (A-B) (B+iI)^{-1}.
\end{split}
\]
Hence for  $A,B \in H_2$, we have $A \sim B$ if and only if $f(A) \sim f(B)$. 

Let us consider the case $A= \infty P+a P^{\perp}$ for some $a \in \bR$ and $P \in \cP$. 
Then $U:=f(A) = P+f(a)P^\perp$. Note that $f(a)\in \mathbb{T}\setminus\{1\}$.
\begin{claim}\label{uv}
Let $B\in \oH$ and $V:=f(B)\in U_2$. 
Then $V$ is coherent to $U$ if and only if one of the following holds: 
\begin{itemize}
\item $B\notin H_2$ and there is $\lambda\in \mathbb{T}$ such that $V=P+\lambda P^\perp$. 
\item $B\in H_2$ and $P^\perp BP^\perp = aP^\perp$. 
\end{itemize}
\end{claim}
\begin{proof}
Note that  $U\sim V$ is equivalent to $U-I\sim V-I$. Note also that $U-I= (f(a)-1)P^\perp\neq 0$. 
Assume that $B\notin H_2$, or equivalently, $1\in \sigma(V)$. 
Then $V-I$ is a scalar multiple of a rank one projection.   
Observe that $V-I$ is a scalar multiple of $P^\perp$ if and only if there is $\lambda\in \mathbb{T}$ such that $V=P+\lambda P^\perp$. 
If these conditions hold, we clearly have $U\sim V$. 
If $V-I=bQ$ for some $P^\perp\neq Q\in \mathcal{P}$ and $0\neq b\in \mathbb{C}$, then $V-U=bQ- (f(a)-1)P^\perp$ is invertible by Lemma \ref{apbq}. 
It follows that $U$ is not coherent to $V$ in this case. 

Assume that $B\in H_2$, or equivalently, $1\notin \sigma(V)$. 
We have $f(a)-1=-2i(a+i)^{-1}$ and $V-I= -2i(B+iI)^{-1}$ by \eqref{mmmm}.
Thus the condition $U-I\sim V-I$ is equivalent to $(a+i)^{-1}P^\perp\sim (B+iI)^{-1}$, which is in turn equivalent to $(a+i)^{-1}(B+iI)P^\perp \sim I$.
By Lemma \ref{tr1}, this is further equivalent to $1=\mathrm{tr}\, ((a+i)^{-1}(B+iI)P^\perp) = (a+i)^{-1}(\mathrm{tr}\,(BP^\perp) +i)$. 
This is equivalent to $\mathrm{tr}\,(BP^\perp) =a$, which means $P^\perp BP^\perp =aP^\perp$.
\end{proof}

This claim clearly implies that for $B\in \oH$ we have $A\sim B$ if and only if $f(A)\sim f(B)$. 
In the case where $A = {\oi}$ and $B$ is any element of $\oH$, we have $I = f(A) \sim f(B)$ if and only if $1$ is an eigenvalue of $f(B)$, which is equivalent to $B= \infty P+aP^{\perp}$ for some $P\in \cP$ and some $a \in {\oR}$.
We have shown that the map $f\colon \oH \to U_2$ is an isomorphism with respect to the coherency relation.

\subsection{The relation between $\oH$ and $\overline{\mathcal{M}_4}$}\label{ohm4}
The compactification $\overline{\mathcal{M}_4}$ \index{$M_4$2@$\overline{\mathcal{M}_4}$} of $\mathcal{M}_4$ is a concept studied by both mathematicians and physicists, notably by R. Penrose. 
Essentially the same space can be introduced in several ways, and it is called by different names such as the conformal compactification\index{conformal compactification}, the conformal Minkowski space\index{conformal Minkowski space}, the compactified Minkowski space\index{compactified Minkowski space}, etc.
The space $\overline{\mathcal{M}_4}$ is visualized with the so-called Penrose's diagram\index{Penrose's diagram}. 
Those readers who are familiar with this concept are encouraged to think of a visual image to see what is going on in each of the arguments in the subsequent sections. 
See for example  \cite[Section 5.1]{HaE}, \cite{Pe} for more information about the space $\overline{\mathcal{M}_4}$.  

In this subsection, we first introduce the space $\overline{\mathcal{M}_4}$ in accordance with Lester's article \cite{Les0}.
The space $\overline{\mathcal{M}_4}$ is naturally endowed with a binary relation that extends the lightlikeness relation in $\mathcal{M}_4$.
We give a complete proof of the fact that $\overline{\mathcal{M}_4}$ endowed with this binary relation can be identified with $\oH$ endowed with our coherency relation.
The results in the current subsection will not be used in the rest of this paper, so those readers who are new to the space $\overline{\mathcal{M}_4}$ may skip to Section \ref{prelim}.

Let us denote by $(\cdot \, , \cdot)$ \index{$(\cdot \, , \cdot)$} the usual inner product on $\mathbb{R}^6$. We further denote 
\[
\index{$M$2@$\overline{M}$}
\overline{M} = \left[ \begin{matrix}   -1 & 0 & 0 & 0 & 0 & 0 \cr
0 & -1 & 0 & 0 & 0 & 0 \cr
0 & 0 & -1 & 0 & 0 & 0 \cr
0 & 0 & 0 & 1 & 0 & 0 \cr  
 0 & 0 & 0 & 0 & 0 & -{1\over 2}\cr
0 & 0 & 0 & 0 & -{1 \over 2} & 0 \cr\end{matrix} \right] .
\]
The symbol $\mathbb{P} ( \mathbb{R}^6 )$ \index{$P(R^6)$@$\mathbb{P} ( \mathbb{R}^6 )$} stands for the projective space over $\mathbb{R}^6$, that is,
\[
\mathbb{P} ( \mathbb{R}^6 ) = \{ [X] \, : \, X \in  \mathbb{R}^6 \setminus \{ 0 \} \}.
\]
Here, $[X]$ \index{$X$@$[X]$} denotes the one-dimensional subspace spanned by the nonzero vector $X$. Next we introduce the symmetric bilinear form $\langle \cdot , \cdot \rangle \colon  \mathbb{R}^6 \times  \mathbb{R}^6 \to  \mathbb{R}$ defined by
\[
\langle X , Y \rangle = (\overline{M} X,Y), \ \ \ X,Y \in  \mathbb{R}^6 ,
\]
and the corresponding quadratic form $q\colon \mathbb{R}^6 \to  \mathbb{R}$ \index{$q$} defined by
\[
q(X) = \langle X,X \rangle = ( \overline{M} X,X), \ \ \ X \in  \mathbb{R}^6 .
\]
Note that the symbol $\langle \cdot , \cdot \rangle$ has been used before to denote the Lorentz--Minkowski indefinite inner product on $\mathcal{M}_4$. When using this symbol it will be always clear from the context which of the two bilinear forms is on our mind. 
Clearly,
\begin{equation}\label{tuniz}\index{$\langle\cdot, \cdot\rangle$}
q(x,y,z,t,h,n) = -x^2 -y^2-z^2 + t^2 -hn , \ \ \ (x,y,z,t,h,n) \in \mathbb{R}^6 .
\end{equation}
For every nonzero $X \in \mathbb{R}^6$ and every nonzero real number $s$ we have  $q(X)=0$ if and only if $q(sX) =0$. We define
\[
\overline{ \mathcal{M}_4 } = \{ [X] \, : \, [X] \in \mathbb{P} ( \mathbb{R}^6 ) \, \ \text{and} \ \, q(X) = 0 \}.
\]
We say that $[X], [Y] \in \overline{ \mathcal{M}_4 }$ are coherent, $[X] \sim [Y]$\index{$\sim$4@$\sim$ (on $\overline{\cM_4}$)}, if $\langle X , Y \rangle = 0$. The set  $\overline{ \mathcal{M}_4 }$ equipped with the coherency relation is called the conformal Minkowski space.

In the next step we will classify points in $\overline{ \mathcal{M}_4 }$ and at the same time we will construct a bijective map $\xi$ \index{$xi$@$\xi$} mapping $\overline{ \mathcal{M}_4 }$ onto $\oH$ which preserves coherency in both directions. Let $[X] = [ (x,y,z,t,h,n) ]$ be a point in $\overline{ \mathcal{M}_4 }$. We will distinguish two possibilities.

We start with the possibility that $h \not = 0$. All points in $\overline{ \mathcal{M}_4 }$ with this property will be called finite points\index{finite point}.
If  $[X] = [ (x,y,z,t,h,n) ]$  is a finite point
then we can assume with no loss of generality that $h=1$. It follows from $q(X) = 0$ that
\[
X =  (x,y,z,t,1,  -x^2 -y^2-z^2 + t^2).
\]
If the spacetime event $(x,y,z,t)\in \mathcal{M}_4$ is denoted by $r$, then we can write shorter
\[
X = (r,1, \langle r,r \rangle).
\]
For such a point $[X] \in \overline{ \mathcal{M}_4 }$ we define
\[
\xi ([X]) =  \left[ \begin{matrix}  t-z & x+iy \cr x-iy & t+z \cr\end{matrix} \right].
\]
Clearly, $\xi$ is a bijection of the set of all finite points in  $\overline{ \mathcal{M}_4 }$ onto $H_2$. 

It is easy to see that the coherency relation on the set of finite points in the conformal Minkowski space corresponds to the coherency relation on $H_2$, that is, to the lightlikeness in Minkowski space. 
Indeed, let $[X], [Y] \in \overline{ \mathcal{M}_4 }$ be finite points,
\[
X =  (r_1,1, \langle r_1,r_1 \rangle)
\]
and
\[
Y =  (r_2,1, \langle r_2,r_2 \rangle).
\]
Then
we have 
\[
\langle X, Y \rangle  =  \langle r_1 , r_2 \rangle -{1 \over 2} ( \langle r_1,r_1 \rangle +  \langle r_2,r_2 \rangle) = -{1 \over 2} \langle r_1-r_2,r_1-r_2 \rangle.
\]
Therefore, the map $r\mapsto  [(r,1, \langle r,r \rangle)]$ is an embedding of the space $\mathcal{M}_4$ endowed with the lightlikeness relation into the space $\overline{\mathcal{M}_4}$ endowed with the relation $\sim$.

When $h=0$ we further distinguish two possibilities. The first one is that $x=y=z=t=0$. There is only one such point in the conformal Minkowski space that will be denoted by $[e]$, where $e = (0,0,0,0,0,1)$. \index{$e$}
We define $\xi ([e]) = {\oi}$. It is clear that for $[X] \in \overline{ \mathcal{M}_4 }$ we have $\langle X, e \rangle \not= 0$ if and only if $[X]$ is a finite point in $\overline{ \mathcal{M}_4 }$.

It remains to consider points $[X] = [(x,y,z,t,0,n)]$ with the property that at least one of the coordinates $x,y,z,t$ is nonzero. Then because of $q(X) = 0$ and \eqref{tuniz} we have $t\not=0$ and without loss of generality we can assume that $t=1/2$. So each such point $[X]$ can be represented by
\begin{equation}\label{musi}
X =\left(x,y,z, {1 \over 2}, 0 , a \right) ,
\end{equation}
where $x^2 + y^2 + z^2 = 1/4$, 
and we define
\[
\xi ([X]) = \xi \left( \left[ \left(x,y,z, {1 \over 2}, 0 , a \right) \right] \right) = \infty P + aP^\perp,
\]
where 
\[
P = \left[ \begin{matrix}  {1 \over 2} -z & x+iy \cr x-iy & {1 \over 2} +z \cr\end{matrix} \right] \ \ \ \text{and} \  \ \ 
P^\perp = \left[ \begin{matrix}  {1 \over 2}   + z & -x-iy \cr -x+ iy & {1 \over 2} - z \cr\end{matrix} \right]
\]
is an orthogonal pair of projections of rank one.

We show that for every $X$ of the form \eqref{musi} and for every $[Y] \in \overline{ \mathcal{M}_4 }\setminus [e]$ we have $[X] \sim [Y] \iff \xi ([X]) \sim \xi ([Y])$.
If both $X$ and $Y$ are of the form \eqref{musi},
\[
X =\left(x_1,y_1,z_1, {1 \over 2}, 0 , a_1 \right) = (r_1, 0, a_1) 
\]
and 
\[
Y =\left(x_2,y_2,z_2, {1 \over 2}, 0 , a_2 \right) = (r_2, 0, a_2)
\]
with $r_j = \left( x_j, y_j , z_j , 1/2 \right)$, $j=1,2$,
then $\langle X , Y \rangle = 0 \iff \langle r_1 , r_2 \rangle =0$ which by a straightforward application of the Cauchy--Schwarz inequality for vectors in $\mathbb{R}^3$ is equivalent to $r_1 = r_2$, that is, $\xi ([X]) \sim \xi ([Y])$.

If
\[
X =\left(x_1,y_1,z_1, {1 \over 2}, 0 , a_1 \right) = (r_1, 0, a_1)
\]
and $Y= (r_2, 1, n)$ with $r_2 = (x_2 , y_2, z_2 , t_2)$ and $n= \langle r_2 , r_2 \rangle$, then
$\langle X, Y \rangle = \langle r_1 , r_2 \rangle - (1/2)a_1$.
On the other hand, $\xi ([X]) \sim \xi ([Y])$ if and only if $\mathrm{tr}\, (P^\perp A) = a_1$, where
\[
P^\perp = \left[ \begin{matrix}  {1 \over 2}  + z_1 & -x_1 -iy_1 \cr -x_1+ iy_1 & {1 \over 2} - z_1 \cr\end{matrix} \right] \ \ \ \text{and} \ \ \ 
A=  \left[ \begin{matrix}  t_2 -z_2 & x_2 +iy_2 \cr x_2 -iy_2  & t_2 +z_2 \cr\end{matrix} \right].
\]
A straightforward calculation shows that $[X] \sim [Y] \iff \xi ([X]) \sim \xi ([Y])$ in  this case as well.

We have shown that $\oH$ can be identified either with the conformal Minkowski space, or with $U_2$. 
The subset $H_2 \subset \oH$ corresponds to the set of all finite points in $\overline{ \mathcal{M}_4 }$, which can be further identified with the classical Minkowski space. Alternatively, the
 subset $H_2 \subset \oH$ corresponds to the set of all unitary matrices $U \in U_2$ with the property that $1 \not\in \sigma (U)$.


\section{Basic concepts}\label{prelim}
In the current section, we introduce basic concepts on $\oH$ and give their properties. 

\subsection{Automorphisms}\label{automorphism}
An \emph{automorphism} \index{automorphism} of $\oH$ is a bijective map $\varphi \colon \oH \to \oH$ that preserves coherency in both directions, i.e., 
\begin{equation}\label{ABAB}
A \sim B \iff \varphi (A) \sim \varphi (B)
\end{equation}
for every pair $A,B \in \oH$.
Let $c \in \{-1 , 1\}$, $S$ be an invertible $2 \times 2$ complex matrix, and $T \in H_2$. 
It is easily seen that the mapping $\varphi\colon H_2\to H_2$ given by 
\begin{equation}\label{prvap}
\varphi (A) = c SAS^\ast + T
\end{equation}
for every $A \in H_2$; or
\begin{equation}\label{drugad}
\varphi (A) = c SA^t S^\ast + T
\end{equation}
for every $A \in H_2$, is a bijection satisfying \eqref{ABAB} for every pair $A,B \in H_2$.

\begin{remark}\label{m4h2}
It is known that a bijection $\varphi\colon H_2\to H_2$ satisfying \eqref{ABAB} for every pair $A,B \in H_2$ needs to be of the form \eqref{prvap} or \eqref{drugad}. 
For the sake of completeness, we will give a proof of this fact in Subsection \ref{apply}, which leads to our proof of the fundamental theorem of chronogeometry.
See also \cite[Subsection 5.3]{Hua}.
\end{remark}

In what follows, we show that a map $\varphi\colon H_2\to H_2$ of the form \eqref{prvap} or \eqref{drugad} extends to an automorphism of $\oH$.
We first introduce the rule $-\infty = \infty$ and define
\[
-( a P + bP^{\perp}) =  (-a) P +(-b)P^{\perp}, \ \ \ P\in \mathcal{P}, \ \, a,b \in \oR.\index{$A$0@$-A$}
\]
Note that this is well-defined, and that if $A\in H_2$ then $-A$ defined in this manner coincides with $-A$ as a matrix.
It is easy to see that the map $X \mapsto -X$, $X  \in \oH$, is a bijection of $\oH$ onto itself and for every pair $X,Y \in \oH$ we have
\[
X \sim Y \iff - X \sim -Y.
\]
Thus, the map $X  \mapsto -X$, $X \in \oH$, is an automorphism of $\oH$. Clearly, this automorphism is the inverse of itself, that is, $-(-X) = X$, $X \in \oH$.

Secondly, let $B \in H_2$. We define
\[
{\oi} +B = {\oi}
\]
and
\[
( \infty P + aP^{\perp}) + B =  \infty P + ( a + \mathrm{tr}\, (P^\perp B) )P^{\perp}  , \ \ \ P\in \mathcal{P}, \ \, a \in \mathbb{R}.
\]
Then $A+B$ \index{$A+B$} is defined for every $A\in \oH\setminus H_2$.
For $A\in H_2$, $A+B$ is defined as a matrix, as usual.
\begin{claim}
The map $X \mapsto X + B$, $X \in \oH$, is an automorphism of $\oH$. 
\end{claim}
\begin{proof}
We need to show that the map preserves coherency in both directions. We will verify only the special case when $X \in H_2$ and $Y =  \infty P + aP^{\perp}$ for some real $a$ and some $P \in \mathcal{P}$. (The other cases are easier or similar.) Then $X \sim Y$ if and only if $P ^\perp X P ^\perp = a P ^\perp$ which is equivalent to $\mathrm{tr}\, (P ^\perp X) = a$.

On the other hand, $ X + B \sim Y+ B = \infty P + (a + \mathrm{tr}\,( P^{\perp} B)) P^\perp$ if and only if $P ^\perp (X+B) P ^\perp = (a + \mathrm{tr}\, (P^{\perp} B)) P^\perp$, which is true if and only if $\mathrm{tr}\, ( P ^\perp (X+B) P ^\perp ) = a + \mathrm{tr}\, (P^{\perp} B)$. The last equality is easily seen to be equivalent to $\mathrm{tr}\, (P ^\perp X) = a$.

For the verification of bijectivity see the next paragraph.
\end{proof}

A mapping of the form $X\mapsto X+B$ will be called a \emph{translation}. \index{translation}
It is trivial to see that the map $X \mapsto X + (- B)$, $X \in \oH$, is the inverse of this mapping. 
We will write shortly $A + (-B) = A-B$\index{$A-B$}, $A \in \oH$. Clearly, we have $A-B = - (-A + B)$ for every $A \in \oH$ and every $B \in H_2$.
We will occasionally write $B+A$ instead of $A+ B$ and $B-A$ instead of $-A+ B$ for $A \in \oH$ and $B \in H_2$. For a subset $\mathcal{A}\subset \oH$ and $B\in H_2$, we will sometimes use the symbols like $-\mathcal{A}=\{-A\,:\, A\in \mathcal{A}\}$ and $\mathcal{A}+B=\{A+B\,:\, A\in \mathcal{A}\}$ etc.
 
Thirdly, we define
\[
( a P + bP^{\perp})^t =  a P^t +  b (P^t)^{\perp}  , \ \ \ P\in \mathcal{P}, \ \,a,b\in \oR. \index{$A^t$}
\]
Note that this is well-defined, and that if $A\in H_2$ then $A^t$ defined in this manner coincides with the usual transpose of the matrix.
It is easy to see that $X \mapsto X^t$, $X \in \oH$, is an automorphism of $\oH$ with $(X^t)^t = X$ for every $X \in \oH$.

Fourthly, let $S$ be any invertible $2 \times 2$ complex matrix. If $P$ is any projection of rank one, then $SPS^\ast$ is a positive rank one matrix.
Therefore,
\begin{equation}\label{rankone}
Q= {1 \over \mathrm{tr}\, (SPS^\ast)  } \, SPS^\ast 
\end{equation}
is a projection of rank one. We define
\[
S\,  {\oi}\,  S^\ast = {\oi}
\]
and
\[
S ( \infty P + aP^{\perp})S^\ast =  \infty Q +   a \mathrm{tr}\, (Q^\perp SS^\ast) \, Q^{\perp}  , \ \ \ P\in \mathcal{P}, \ \, a \in \mathbb{R},
\]
where $Q$ is defined by \eqref{rankone}.  
Then $SAS^*$ \index{$SAS^*$} is defined for every $A\in \oH\setminus H_2$.
For $A\in H_2$, $SAS^*$ is defined as a matrix, as usual.
\begin{claim}\label{sas*}
The map $X \mapsto SXS^\ast $, $X \in \oH$, is an automorphism of $\oH$. 
\end{claim}
\begin{proof}
We need to verify that for every pair $X,Y \in \oH$ we have
\begin{equation}\label{sxs}
X \sim Y \iff SXS^\ast \sim SYS^\ast.
\end{equation}
The verification of this equivalence is trivial in the following cases:
\begin{itemize}
\item at least one of $X$ and $Y$ is equal to ${\oi}$,
\item both $X$ and $Y$ belong to $H_2$, and
\item both $X$ and $Y$ belong to $\oH\setminus H_2$.
\end{itemize}
Thus, it remains to consider the case when 
$X \in H_2$ and $Y =  \infty P + aP^{\perp}$ for some real $a$ and some $P \in \mathcal{P}$. Then $X \sim Y$ if and only if $\mathrm{tr}\, (P^\perp X) = a$. And we have $SXS^\ast \sim SYS^\ast$ if and only if
\[
\mathrm{tr}\, (Q^\perp SXS^\ast) = a \mathrm{tr}\, (Q^\perp SS^\ast),
\]
or equivalently, 
\[
\mathrm{tr}\, (S^\ast Q^\perp SX) = a \mathrm{tr}\, (S^\ast Q^\perp S).
\]
Note that $S^\ast Q^\perp S$ is a positive matrix of rank one. 
We show that the equality
\[
\mathrm{tr}\, (P^\perp X)={1 \over \mathrm{tr}\, (S^\ast Q^\perp S ) } \, 
\mathrm{tr}\, (S^\ast Q^\perp SX)
\]
holds for all $X\in H_2$. 
It suffices to show that the positive matrix $S^\ast Q^\perp S$ of rank one is a (positive) scalar multiple of $P^\perp$. 
This is easily seen by the following equation: 
\[
S^\ast Q^\perp SP = S^\ast Q^\perp SP S^\ast (S^\ast)^{-1} =  S^\ast  (Q^\perp SP S^\ast) (S^\ast)^{-1} = S^\ast \cdot 0\cdot (S^\ast)^{-1} =0.
\]

For the verification of bijectivity see the next claim.
\end{proof}

\begin{claim}\label{SS-1}
The inverse of the map $X \mapsto SXS^\ast $, $X \in \oH$,  is the map $X \mapsto S^{-1}X(S^{-1})^\ast$, $X \in \oH$. 
\end{claim}
\begin{proof}
We know that \eqref{sxs} holds for every pair $X,Y \in \oH$. 
In the same way, we see that 
\[
X \sim Y \iff S^{-1}X(S^{-1})^\ast \sim  S^{-1}Y(S^{-1})^\ast
\]
for every pair $X,Y \in \oH$.
If we denote by $\varphi$ and $\psi$ the maps from $\oH$ to itself defined by $\varphi (X) = SXS^\ast$ and $\psi (X) =  S^{-1}X(S^{-1})^\ast$, $X \in \oH$, then 
we have 
\[
X \sim Y \iff \psi ( \varphi (X)) \sim \psi (\varphi (Y))
\]
for every pair $X,Y \in \oH$, and $\psi (\varphi  (X)) = X$ for every $X \in H_2 \cup \{ {\oi} \}$.
It is easy to see that for any pair of points $X,Y \in \oH\setminus H_2$ we have 
\begin{equation}\label{mrda}
\mathcal{C}_X \cap H_2 = \mathcal{C}_Y \cap H_2 \iff X=Y.
\end{equation}
It follows that  $\psi (\varphi  (X)) = X$ for every $X \in \oH$, and in the same way we see that $\varphi (\psi  (X)) = X$ for every $X \in \oH$.
Hence, the inverse of the automorphism $X \mapsto SXS^\ast$ is $X \mapsto S^{-1} X (S^{-1})^\ast$.
\end{proof}

We introduce the rules $\infty^{-1}= 0$, $0^{-1}=\infty$ in ${\oR}$.
For $A\in \oH$, we define $A^{-1}$ \index{$A^{-1}$} in the following manner:
\[
(aP +bP^{\perp})^{-1} = a^{-1} P+ b^{-1} P^{\perp}, \ \ \ P\in \mathcal{P}, \ \, a, b \in  {\oR}.
\]
Note that this is well-defined and compatible with the usual inverse of matrices when $a, b\in \bR\setminus \{0\}$.
\begin{claim}\label{cl}
The map $X \mapsto X^{-1}$, $X \in \oH$, is an automorphism of $\oH$. 
\end{claim}
\begin{proof}
Let us consider everything inside $U_2$. 
Let $f$ be defined as in \eqref{cayley} and \eqref{cayleyf}.
For a unitary matrix $U\in U_2$, take $a,b\in {\oR}$ and $P\in \cP$ such that $U=f(a) P + f(b) P^\perp$.
Then $aP+bP^\perp$ is the corresponding point in $\oH$, and $(aP+bP^\perp)^{-1}=a^{-1}P+b^{-1}P^\perp$. 
It follows that $f(a^{-1}P+b^{-1}P^\perp)=f(a^{-1})P+f(b^{-1})P^\perp$. 
Since \eqref{cayleyf} clearly implies $f(t^{-1})=-\overline{f(t)}$ for every $t\in {\oR}$, we obtain 
\[
f(a^{-1})P+f(b^{-1})P^\perp = -\overline{f(a)}P - \overline{f(b)}P^\perp = -U^*. 
\]
Now, we obtain the desired conclusion as a consequence of the following fact.  
For every pair of unitaries $U, V\in U_2$, we have $U \sim V \iff -U^* \sim -V^*$.
\end{proof}
The automorphism $X\mapsto X^{-1}$ of $\oH$ will be called the \emph{inversion}\index{inversion}.
The inverse mapping of the inversion is the inversion itself.

\begin{definition}\label{automorphisms}
An automorphism $\varphi\colon \oH\to \oH$ of the form
\[
\varphi (X) = c SXS^\ast + B, \ \ \ X \in \oH;
\]
or
\[
\varphi (X) = c SX^t S^\ast + B, \ \ \ X \in \oH
\]
for some $c \in \{-1 , 1 \}$, some invertible $2 \times 2$ complex matrix $S$, and some $B \in H_2$ will be called an \emph{affine automorphism}\index{affine automorphism}. 
We say that an automorphism of $\oH$ is \emph{standard} \index{standard automorphism} if it is written as a composition of finitely many affine automorphisms and inversions.
\end{definition}
It is easily seen that the set of standard automorphisms forms a group.
It is known that every automorphism of $\oH$ is standard. 
For the sake of completeness, we will give a proof of this fact in Subsection \ref{apply}.
We also see the following. 

\begin{lemma}\label{simsim}
The set of affine automorphisms forms a group.
\end{lemma}
\begin{proof}
Let $\varphi_1, \varphi_2$ be two affine automorphisms. 
We show that $\varphi_2\circ \varphi_1$ is an affine automorphism. 
Let us consider the case where $\varphi_i$ is of the form $X\mapsto S_iXS_i^\ast + B_i$ for some invertible $2 \times 2$ complex matrix $S_i$, and some $B_i \in H_2$,  $i=1,2$. 
Let $\varphi$ be the affine automorphism $X\mapsto (S_2S_1)X(S_2S_1)^* + S_2B_1S_2^* +B_2$. 
It is easily seen that $\varphi^{-1}\circ \varphi_2\circ \varphi_1(X)=X$ for every $X\in H_2$. 
This leads to $\varphi^{-1}\circ \varphi_2\circ \varphi_1(X)=X$ for every $X\in \oH$ (as in the discussion in the proof of Claim \ref{SS-1}), hence we get $\varphi_2\circ \varphi_1=\varphi$. 
The other cases can be considered in an analogous manner. 
Similarly, we may show that $\varphi_1^{-1}$ is an affine automorphism.
\end{proof}

In fact, the class of affine automorphisms coincides with the group of symmetries on $H_2$, that is, the group of all bijections on $H_2$ that preserve coherency in both directions (Theorem \ref{ftcgelem} and Corollary \ref{corp}).
Clearly, the class of standard automorphisms is bigger than that of affine automorphisms. 
This is the reason why we introduce the space $\oH$. 
If we work with $H_2$ together with symmetries on $H_2$ instead of $\oH$, then we cannot obtain most of the results in the subsequent subsections (particularly in Subsections \ref{existence}, \ref{line}, \ref{ssss}, \ref{three points}).

\subsection{Bloch representation}\label{projectionsx}\index{Bloch representation}
The symbol $\| \cdot \|$ \index{$\lVert \cdot \rVert$} denotes the usual operator norm.
A matrix $A\in H_2$ of trace one with $0\leq A$ is called a \emph{density matrix}\index{density matrix}.
We first recall some known facts about the Bloch representation \index{Bloch representation} of the set of all density matrices in $H_2$. 
The reader will note that the idea is the same as in Introduction where we have identified spacetime events with $2\times 2$ hermitian matrices. 

It is easy to see that the map
given by
\begin{equation}\label{blochmap}
 \left[ \begin{matrix} 1/2 - z & x+iy \cr x-iy & 1/2 +z \cr \end{matrix} \right] \mapsto (x,y,z) \in \mathbb{R}^3 
\end{equation}
is a bijection of the set of all density matrices in $H_2$ onto the closed ball in $\mathbb{R}^3$ centered at the origin with radius $1/2$, which we call the \emph{Bloch ball}\index{Bloch ball}. 
The sphere of the Bloch ball is called the \emph{Bloch sphere}\index{Bloch sphere}, and we denote it by $\mathcal{S}^2$\index{$S^2@$\mathcal{S}^2$}. 
Observe that the above mapping sends the set $\cP$ of $2\times2$ projections of rank one onto $\mathcal{S}^2$. 

Let
\[
A =  \left[ \begin{matrix} 1/2 - z_1 & x_1+iy_1 \cr x_1-iy_1 & 1/2 +z_1 \cr \end{matrix} \right] \ \ \ \text{and} \ \ \ B =  \left[ \begin{matrix} 1/2 - z_2 & x_2+iy_2 \cr x_2-iy_2 & 1/2 +z_2 \cr \end{matrix} \right]
\]
be any density matrices and $u = (x_1 , y_1 , z_1)$ and $v = (x_2, y_2, z_2)$ the corresponding points in the Bloch ball. 
Because $A-B$ is a $2\times 2$ hermitian matrix with trace zero, we have 
\begin{equation}\label{bb}
\| A - B \|^2 = - \det (A-B)= (x_2 - x_1)^2 + (y_2 - y_1)^2 + (z_2 - z_1)^2,
\end{equation}
that is, the norm distance between $A$ and $B$ is equal to the Euclidean distance in the Bloch ball.

If in addition $A=P, B=Q$ are projections, then the right-hand side of \eqref{bb} is further equal to
\[
 ( u-v, u-v ) = \frac{1}{4} + \frac{1}{4} - 2 ( u,v ) = \frac{1- \cos \alpha}{2}
= \sin^2 \frac{\alpha}{2},
\]
where $\alpha\in [0,\pi]$ is the angle between $u$ and $v$ in $\mathcal{S}^2$, that is, $\cos \alpha = 4 ( u , v )$.
Therefore, $\arcsin \| P - Q \|$ equals $\alpha/2$, which is the geodesic distance between $u$ and $v$ in the Bloch sphere (that is, the length of the shorter arc in the great circle passing through $u$ and $v$). 
We define \index{$d^g$} 
\begin{equation}\label{sooce}
d^g (P,Q) =\arcsin \| P - Q \|
\end{equation}
for $P, Q\in \cP$. 
For us, one important fact about this distance function is the following obvious property. 

\begin{lemma}\label{P}
If $P,Q \in \cP$ with $d^g (P,Q) = s$ and $n$ is a positive integer and $t$ a real number such that $nt \ge s$, then there exist points $P_0 = P, P_1, \ldots, P_n = Q$ in $\cP$ such that $d^g(P_{j-1}, P_j ) \le t$, $j= 1,2,\ldots, n$.
\end{lemma}
\begin{proof}
This is trivial by thinking about the geodesic distance in $\mathcal{S}^2$.
\end{proof}

The following theorem will be used in Section \ref{standard}.
\begin{theorem}[Special case of the non-bijective version of Wigner's theorem]\index{Wigner's theorem}\label{wigner}
Let $\mu \colon \mathcal{P} \to \mathcal{P}$ be a mapping satisfying
\begin{equation}\label{wigass}
\mathrm{tr}\, (\mu (P) \mu (Q) ) = \mathrm{tr}\, (PQ), \ \ \ P,Q \in \mathcal{P}.
\end{equation}
Then there exists a $2 \times 2$ unitary matrix $U$ such that either 
\[
\mu (P) = UPU^\ast,\ \ P \in \mathcal{P},
\]
or 
\[
\mu (P) = UP^tU^\ast,\ \ P \in \mathcal{P}.
\] 
\end{theorem}

See {\cite{Geh}} for an elementary proof of this theorem.

\begin{remark}\label{wignerremark}
Since $(P-Q)^2$ is a scalar multiple of $I$, we get 
\[
\| P - Q \|^2 =\| (P - Q)^2 \|= \frac{1}{2}\mathrm{tr}\,\left( (P-Q)^2\right) = \frac{1}{2}\mathrm{tr}\, (P-PQ-QP+Q) =1-\mathrm{tr}\,(PQ).
\]
From this, we see that \eqref{wigass} is equivalent to 
\[
\| \varphi (P) - \varphi (Q) \| = \| P - Q \|, \ \ \ P,Q \in \mathcal{P}.
\]
\end{remark}

We continue collecting some lemmas concerning $\cP$ that will be used in Section \ref{standard}.
\begin{lemma}\label{matrum}
Let $P,Q \in \mathcal{P}$. Then
\[
\| P -Q \|^2 + \| P^\perp - Q\|^2 = 1.
\]
\end{lemma}
\begin{proof}
This is clear from \eqref{bb} and the fact that $P$ and $P^{\perp}$ correspond to antipodal points in the Bloch sphere.
\end{proof}

\begin{lemma}\label{veryeasy}
For $A\in \oH$, the following are equivalent:
\begin{itemize}
\item $A \sim 0$ and $A\sim I$,
\item $A \in \mathcal{P}$.
\end{itemize}
\end{lemma}
\begin{proof}
If $A \in \mathcal{P}$, then obviously $A \sim 0$ and $A\sim I$. If $A \sim 0$, then $A = tP$ for some $P \in \mathcal{P}$ and some $t \in {\oR}$. Clearly, $tP \sim I$ if and only if $t=1$. This completes the proof.
\end{proof}

For $A,B \in \oH$ with $d(A,B) = 2$, let $\mathcal{S}_{A,B}$ \index{$S(A,B)$@$\mathcal{S}_{A,B}$} denote the set $\mathcal{C}_A \cap \mathcal{C}_B$.

\begin{lemma}\label{mcmcv}
Let $a$ be a nonzero real number and $P \in \mathcal{P}$. 
Then
\begin{equation}\label{eqmcmcv}
\mathcal{S}_{(1-a)P, P + aP^\perp} = \{ (1-a)P + aR\, : \, R \in \mathcal{P} \}.
\end{equation}
\end{lemma}
\begin{proof}
It is an easy consequence of Lemma \ref{veryeasy} that
\[
\mathcal{S}_{0, aI} = \{  aR\, : \, R \in \mathcal{P} \}.
\]
For any $A,B \in H_2$ with $d(A,B) = 2$ we have
\[
\mathcal{S}_{A + (1-a)P, B + (1-a) P} =  \{X+(1-a)P \,:\, X\in\mathcal{S}_{A,B}\}.
\] 
Combining the two equalities with $A=0$, $B=aI$, we get \eqref{eqmcmcv}.
\end{proof}

\begin{lemma}\label{mikar}
Let $P,Q \in \mathcal{P}$ and $a$ be a positive real number, $a < 1/2$. Then the intersection
\[
\{ (1-a) P + aR \, : \, R\in \mathcal{P} \} \, \cap \, \{ (1-a) Q + aR \, : \, R\in \mathcal{P} \}
\]
is nonempty if and only if
\[
\| P - Q \| \le { a \over 1-a }.
\]
\end{lemma}
\begin{proof}
Since the mapping \eqref{blochmap} is a restriction of an affine mapping, we see that the sets $\{ (1-a) P + aR \, : \, R\in \mathcal{P} \}$ and $\{ (1-a) Q + aR \, : \, R\in \mathcal{P} \}$ are sent by this mapping to spheres with radius $a/2$ that are both inscribed in $\mathcal{S}^2$.
Thus we need to verify that if
$\mathcal{S}_1$, $\mathcal{S}_2$ are spheres of radius $a/2$ that are inscribed in $\mathcal{S}^2$ with contact points $P, Q\in \mathcal{S}^2$, respectively, 
then $\mathcal{S}_1$ and $\mathcal{S}_2$ intersect if and only if the Euclidean distance between $P$ and $Q$ is at most $a/(1-a)$.
The proof of this claim is fairly easy and we omit it.
\end{proof}

\subsection{Existence of standard automorphisms with certain properties}\label{existence}
\begin{lemma}\label{pocelo}
Let $A,B \in \oH$. If $d(A,B) = 2$, then there is a standard automorphism $\varphi \colon \oH \to \oH$ such that $\varphi (A) = 0$ and $\varphi (B) = {\oi}$.
\end{lemma}
\begin{proof}
There exists a standard automorphism $\psi \colon \oH \to \oH$ such that $\psi (B) = {\oi}$. 
Indeed, if $B\in H_2$ then take the map $X\mapsto (X-B)^{-1}$. 
If $B=aP+\infty P^\perp$ with $P\in \cP$, $a\in \mathbb{R}$, then take the map $X\mapsto ((X+(1-a)P)^{-1}-P)^{-1}$.
If $B=\oi$, let $\psi$ be the identity mapping.
From $d(\psi (A), \psi (B)) = d(A,B) = 2$, we deduce that $\psi(A) \in H_2$. Hence, the standard automorphism
\[
X \mapsto \psi (X) - \psi (A), \ \ \ X \in \oH,
\]
has the desired property.
\end{proof}

\begin{lemma}\label{pjanck}
Let $A,B \in \oH$ with $d(A,B) =2$. Then for any $X,Y \in \mathcal{S}_{A,B}$ we have either $X=Y$, or $d(X,Y) =2$.
\end{lemma}

\begin{proof}
Applying Lemma \ref{pocelo}, we see that there is no loss of generality in assuming that $A=0$ and $B=\oi$. In that case, we have ${\mathcal{S}}_{A,B}=\{\infty P\,:\, P\in \cP\}$, and the desired conclusion follows trivially.
\end{proof}

\begin{lemma}\label{pocelo1}
Let $A,B \in \oH$. If $d(A,B) = 1$, then there is a standard automorphism $\varphi \colon \oH \to \oH$ satisfying
\[
\varphi (A) = 0\ \  \text{and}\ \ \varphi (B) = E_{11}=\left[ \begin{matrix} 1 &0 \cr 0 & 0\cr \end{matrix} \right].
\]
\end{lemma}

\begin{proof}
Without loss of generality, we may assume that $A=0$. Then $B= sP$ for some rank one projection $P$ and some nonzero $s \in {\oR}$. If $s\not=\infty$, then we can find an invertible $2 \times 2$ matrix $T$ such that $T(sP)T^\ast = \pm E_{11}$ and we are done. In the case where $s= \infty$, we can find a standard automorphism $\varphi_1$ of $\oH$ such that $\varphi_1 (0) = 0$ and $\varphi_1 ({\oi}) = I$. Because $\varphi_1 (\infty P) \sim \varphi_1 (0)=0$ and  $\varphi_1 (\infty P) \sim \varphi_1  ({\oi}) = I$, Lemma \ref{veryeasy} implies that $\varphi_1 (B) = Q$ for some rank one projection $Q$. Multiplying $\varphi_1$ with an appropriate unitary similarity, we get a standard automorphism $\varphi$ with the desired properties.
\end{proof}

Let $A,B \in H_2$ and $A \leq B$. We denote by $[A,B] \subset H_2$ \index{$AB$1@$[A,B]$} the matrix interval \index{matrix interval}
\[
[A,B] = \{ X \in H_2 \, : \, A \le X \le B \}.
\]
Further, let $(A,B) = \{ X \in H_2 \, : \, A < X < B \}$ \index{$AB$2@$(A,B)$} and $[A,B) = \{ X \in H_2 \, : \, A \le  X < B \}$\index{$AB$3@$[A,B)$}, $(A,B] = \{ X \in H_2 \, : \, A <  X \leq B \}$ \index{$AB$4@$(A,B]$} when $A<B$. 
In the language of special relativity, these correspond to the set of spacetime events in the future of a fixed event and in the past of another event, possibly with suitable boundary points.

\begin{lemma}\label{Porder}
Let $A, B, C, D\in H_2$ satisfy $A<B$ and $C<D$.
Then there is an affine automorphism $\psi$ of $\oH$ satisfying $\psi([A, B])=[C, D]$ and 
\[
X\leq Y\iff \psi(X)\leq \psi(Y),\quad X< Y\iff \psi(X)< \psi(Y)
\]
for any pair $X, Y\in H_2$.
\end{lemma}
\begin{proof}
Let us define the affine automorphism $\psi_1$ by 
\[
\psi_1(X)=(B-A)^{1/2}X(B-A)^{1/2} + A,\ \ X\in \oH.
\]
Then we have 
$\psi_1([0, I])=[A, B]$ and 
\[
X\leq Y\iff \psi_1(X)\leq \psi_1(Y),\quad X< Y\iff \psi_1(X)< \psi_1(Y)
\]
for any pair $X, Y\in H_2$.
Similarly, we define the affine automorphism $\psi_2$ by 
\[
\psi_2(X)=(D-C)^{1/2}X(D-C)^{1/2} + C,\ \ X\in \oH. 
\]
We see that the affine automorphism $\psi_2\circ\psi_1^{-1}$ satisfies the desired properties.
\end{proof}

\begin{corollary}\label{A<B}
Let $A, B\in H_2$ satisfy $A<B$. 
Then $\mathcal{S}_{A,B}\subset [A,B]\subset H_2$. 
\end{corollary}
\begin{proof}
By the preceding lemma, it suffices to consider the case $A=0$ and $B=I$. 
In this case, we have $\mathcal{S}_{0,I}=\cP\subset [0,I]\subset H_2$, as desired.
\end{proof}

\begin{lemma}\label{order}
Let $C \in (0,I)$. 
There is a standard automorphism $\psi$ of $\oH$ such that 
\[
\psi([0, I])=[0, I],\ \  \psi(0)=0,\ \  \psi(I)=I,\ \  \text{and}\ \  \psi((1/2)I)=C.
\] 
\end{lemma}

\begin{proof}
Consider the standard automorphism $\psi_1$ of $\oH$ defined by $\psi_1(X)=X^{-1}-I$, $X\in \oH$. 
It is easily seen by the definition of the mapping that 
\[
\psi_1(0)=\oi,\ \ \psi_1(I)=0,\ \  \psi_1((1/2)I)=I,
\]
\[
\psi_1([0, I])=\oHp:=\{aP+bP^{\perp}\,:\, a, b\in [0, \infty],\, P\in \mathcal{P}\}, \index{$Hp$@$\oHp$}
\] and that $\psi_1(C)$ is a positive invertible matrix in $H_2$. 
Consider another standard automorphism $\psi_2$ defined by $\psi_2(X)=\psi_1(C)^{1/2}X\psi_1(C)^{1/2}$, $X\in \oH$. 
It satisfies 
\[
\psi_2(0)=0,\ \ \psi_2(\oi)=\oi,\ \ \psi_2(I)=\psi_1(C),\ \ \text{and}\ \  \psi_2(\oHp)=\oHp.
\] 
It follows that the composition $\psi_1^{-1}\circ\psi_2\circ \psi_1$ satisfies the desired property.
\end{proof}

The above two lemmas clearly imply the following.

\begin{corollary}\label{cororder}
Let $A, B, C, A',B',C'\in H_2$ satisfy $A<C<B$ and $A'<C'<B'$.
Then there is a standard automorphism $\psi$ of $\oH$ satisfying 
\[
\psi([A, B])=[A', B'],\ \  \psi(A)=A',\ \  \psi(B)=B',\ \  \text{and}\ \  \psi(C)=C'.
\] 
\end{corollary}

\subsection{Lines}\label{line}\index{line}
A subset $\mathcal{A} \subset \oH$ is said to be coherent \index{coherent set} if $A \sim B$ for every pair $A,B \in \mathcal{A}$. A maximal coherent set is called a \emph{line}\index{line}. Clearly, if $\varphi \colon \oH \to \oH$ is an automorphism and $\ell \subset \oH$ is a line, then $\varphi ( \ell)$ is a line as well.

\begin{lemma}\label{pq}
Let $a,b \in {\oR}\setminus\{0\}$ and $P,Q \in \mathcal{P}$. 
Then $aP \sim bQ$ if and only if $P=Q$.
\end{lemma}
\begin{proof}
It is clear that $P=Q$ implies $aP \sim bQ$.
Let us prove that $aP \sim bQ$ implies $P=Q$.
Lemma \ref{apbq} verifies this when $a,b\in \mathbb{R}\setminus\{0\}$. 
Assume that $b\in \mathbb{R}\setminus\{0\}$ and $\infty P\sim bQ$. 
Then we have $bP^\perp QP^\perp =0$, thus $P=Q$. 
Similarly, $a\in \mathbb{R}\setminus\{0\}$ and $a P\sim \infty Q$ imply $P=Q$. 
Finally, it is trivial that $\infty P=\infty Q$ implies $P=Q$.
\end{proof}

\begin{lemma}\label{jojhcer}
Let $\ell$ be a line. Then either there exist $A\in H_2$ and $P \in \mathcal{P}$ such that
\[
\ell = \{ aP+A \, : \, a \in \oR \}
\]
or there exists $P \in \mathcal{P}$ such that
\[
\ell = \{ \infty P + aP^\perp \, : \, a \in {\oR}  \}.
\]
\end{lemma}

\begin{proof}
We first consider the possibility that there exists $A \in \ell \cap H_2$. Then $\ell' = \ell - A = \{ X - A \, : \, X \in \ell \}$ is a line that contains $0$. 
We know that $\mathcal{C}_0 = \{ aP \, : \, P \in \mathcal{P}, \ \, a \in {\oR} \}$. Since $\ell'$ is a maximal coherent set, the preceding lemma implies that $\ell' = \{ aP \, : \, a \in {\oR} \}$ for some $P \in \mathcal{P}$. 
Hence 
\[
\ell = \ell' + A =  \{ aP+A \, : \, a \in \oR \}.
\]

The remaining possibility is that $\ell \subset \oH \setminus H_2$. Since $\ell$ is a maximal coherent set, there exists a point in $\ell$ that is different from ${\oi}$. Thus, there exists $P \in \mathcal{P}$ and $a_0 \in \mathbb{R}$ such that $\infty P + a_0 P^\perp \in \ell$. The set of all points in $ \oH \setminus H_2$ that are coherent to $\infty P + a_0 P^\perp$ is $ \{ \infty P + aP^\perp \, : \, a \in {\oR}  \}$. It is easy to see that this is a maximal coherent set. Hence this set is equal to $\ell$.
\end{proof}

\begin{lemma}\label{simpll}
Let $A,B \in \oH$ be such that $d(A,B) = 1$. Then there exists exactly one line $\ell$ in $\oH$ such that $A,B \in \ell$, which equals $\cC_A\cap \cC_B$.
\end{lemma}

\begin{proof}
Let $\varphi \colon \oH \to \oH$ be an automorphism. Then there exists exactly one line passing through $A$ and $B$ if and only if there exists exactly one line passing through $\varphi(A)$ and $\varphi(B)$. Thus, by Lemma \ref{pocelo1},
with no loss of generality we may assume that $A=0$ and $B=E_{11}$. 
In that case, it is obvious from Lemma \ref{pq} that $\cC_A\cap\cC_B = \{ a E_{11} \, : \, a\in {\oR} \}$ is the unique line that contains both $A$ and $B$.
\end{proof}

\begin{lemma}\label{simplll}
Let $\ell$ be a line in $\oH$ and $A$ a point in $\oH$ such that $A\not\in \ell$. Then there is exactly one point on the line $\ell$ that is coherent to $A$.
\end{lemma}

\begin{proof}
With no loss of generality, we may assume that $A = \oi$. 
From $A = \oi \not\in \ell$ and Lemma \ref{jojhcer}, we see that there are
$B\in H_2$ and $P \in \mathcal{P}$ such that
\[
\ell = \{ aP+B  \, : \, a \in \oR \},
\]
and then clearly, $\infty P+B = \infty P + \mathrm{tr}\, (P^\perp B) P^\perp$ is the unique point on $\ell$ that is coherent to $A = \oi$.
\end{proof}

For $A,B \in \oH$ with $d(A,B) = 1$, we denote by $\ell_{A,B}$ \index{$l_A,B$@$\ell_{A,B}$} the unique line passing through $A$ and $B$. Note that by Lemma \ref{simpll} we have $\ell_{A,B} = \mathcal{C}_A \cap \mathcal{C}_B$. 
Denoting $\mathcal{C}_A \cap \mathcal{C}_B$ with two different symbols (depending on the distance between $A$ and $B$) makes sense because the sets $\mathcal{S}_{A,B}$ and $\ell_{A,B}$ are of completely different shapes. 

\begin{remark}\label{ceruk}
We will quite often use the following observation. Assume that $A,B,C \in \oH$ are points such that $d(A,B) =1$ and $d(A,C) =2$. Then there is a unique point $D \in \oH$ that is coherent to $A$ and $B$ and $C$. Indeed, let us denote by $\ell$ the unique line passing through $A$ and $B$. Then, clearly $C \not\in \ell$ and there is a unique $D \in \ell$ such that $D \sim C$.
\end{remark}

\begin{corollary}\label{veryeasy2}
Let $s,t$ be nonzero distinct real numbers and $Q \in \mathcal{P}$. Assume that $A\in \oH$ satisfies $A \sim sI$, $A\sim tI$, and $A\sim sQ$. Then $A = sQ + tQ^\perp$.
\end{corollary} 
\begin{proof}
Since the set $\{sQ, sI, A\}$ is coherent, it is contained in the line $\ell_{sQ,sI}$, thus we have $A = sQ + aQ^\perp$ for some $a\in {\oR}$. 
By $A\sim tI$ we obtain $a=t$.
\end{proof}

\begin{corollary}\label{triangle}
Let $\ell_1 , \ell_2 , \ell_3$ be pairwise different lines in $\oH$.  Assume that $\ell_1 \cap \ell_2 \not= \emptyset$, 
$\ell_1 \cap \ell_3 \not= \emptyset$,  and $\ell_2 \cap \ell_3 \not= \emptyset$. Then there exists $A \in \oH$ such that
\[
\ell_1 \cap \ell_2 = \ell_1 \cap \ell_3 =\ell_2 \cap \ell_3  = \{ A \}.
\]
\end{corollary}

\begin{proof}
By  Lemma \ref{simpll}, the intersection of any two  of the above three lines is a singleton. Denote $\ell_1 \cap \ell_2 = \{ A \}$,  $\ell_1 \cap \ell_3 = \{ B \}$, and $ \ell_2 \cap \ell_3  = \{ C \}$. We need to show that $A=B=C$. Assume on the contrary that this is not true, say $A\not=B$. Then $A,B$ are two distinct points on $\ell_1$ and $C$ is coherent to both of them. By Lemma \ref{simplll}, we have $C \in \ell_1$, too. 
Thus 
\[
C\in \ell_1\cap \ell_2\cap \ell_3 = (\ell_1\cap \ell_2)\cap (\ell_1 \cap \ell_3)=\{A\}\cap\{B\} = \emptyset,
\]
a contradiction.
\end{proof}

\subsection{Surfaces}\label{ssss}\index{surface}
\begin{lemma}\label{herz}
Let $\ell_1$ and $\ell_2$ be distinct lines in $\oH$.
\begin{itemize}
\item If $\ell_1 \cap \ell_2 \not= \emptyset$, then there exists a standard automorphism $\varphi$ of $\oH$ such that
\[
\varphi ( \ell_1) = \{ a E_{11} \, : \, a \in {\oR} \} \ \ \ \text{and} \ \ \ \varphi ( \ell_2) = \{ a E_{22} \, : \, a \in {\oR} \}.
\]
\item   If $\ell_1 \cap \ell_2 = \emptyset$, then there exists a standard automorphism $\varphi$ of $\oH$ such that
\[
\varphi ( \ell_1) = \{ a E_{11} \, : \, a \in {\oR} \} \ \ \ \text{and} \ \ \ \varphi ( \ell_2) = \{  a E_{11}+\infty E_{22}  \, : \, a \in {\oR} \}.
\]
\end{itemize}
\end{lemma}

\begin{proof}
Assume first that $\ell_1 \cap \ell_2 \not=\emptyset$. Then, by Lemma \ref{simpll}, the intersection of these two lines is a singleton. With no loss of generality, we may assume that  $\ell_1 \cap \ell_2= \{ 0 \}$. Then there are $P,Q \in \mathcal{P}$ with $P \not=Q$ and $\ell_1 = 
\{ a P \, : \, a \in {\oR} \}$ and  $\ell_2 = 
\{ a Q \, : \, a \in {\oR} \}$. If $u,v$ are unit vectors spanning the image of $P$ and the image of $Q$, respectively, then they are linearly independent and there exists an invertible linear operator $S \colon \mathbb{C}^2 \to \mathbb{C}^2$ such that $Su= e_1$ and $Sv = e_2$. 
It follows that 
\[
SPS^\ast = Suu^*S^* = (Su)(Su)^* = e_1e_1^*= E_{11}
\]
and similarly, $SQS^\ast = E_{22}$. 
Let $\psi$ denote the automorphism $X\mapsto SXS^*$. 
Then $\psi(\ell_1)$ is a line that contains $0$ and $E_{11}$, hence $\psi(\ell_1)= \{ a E_{11} \, : \, a \in {\oR} \}$.
Similarly, $\psi(\ell_2) = \{ a E_{22} \, : \, a \in {\oR} \}$.

Assume next that $\ell_1 \cap \ell_2 =\emptyset$. Since $\ell_1\cup \ell_2$ is clearly not a coherent set, by Lemma \ref{pocelo} we can assume that $0 \in \ell_1$ and $\oi \in \ell_2$. It follows that there are $P,Q \in \mathcal{P}$ such that $\ell_1 = 
\{ a P \, : \, a \in {\oR} \}$ and  $\ell_2 = 
\{ \infty Q + aQ^\perp \, : \, a \in {\oR} \}$. We have $P \not=Q$, since otherwise $\infty P = \infty Q$ would belong to the intersection of $\ell_1$ and $\ell_2$. 
We see that there exists an invertible $2 \times 2$ matrix $S$ such that $SPS^\ast = E_{11}$ and $SQS^\ast = E_{22}$.
Let $\psi$ denote the automorphism $X\mapsto SXS^*$. 
Then $\psi(\ell_1) = \{ a E_{11} \, : \, a \in {\oR} \}$ as before. 
Since $\psi(\oi)=\oi$ and $\psi(\infty Q)=\infty E_{22}$, we get $\psi(\ell_2) = \{  a E_{11} +\infty E_{22}  \, : \, a \in {\oR} \}$.
\end{proof}

\begin{definition}\label{surface}
Let $\ell_1$ and $\ell_2$ be lines in $\oH$ with $\ell_1 \cap \ell_2 = \emptyset$. The \emph{surface} \index{surface} along $\ell_1$ and $\ell_2$ is the union of all lines $\ell \subset \oH$ satisfying $\ell \cap \ell_1 \not= \emptyset$ and $\ell \cap \ell_2 \not= \emptyset$. 
\end{definition}

\begin{lemma}\label{standsurface}
The surface along the lines  $\{ a E_{11} \, : \, a \in {\oR} \}$ and $\{  \infty E_{22}  + a E_{11} \, : \, a \in {\oR} \}$ is
\[
\{ a E_{11} + b E_{22} \, : \, a,b \in {\oR} \}.
\]
\end{lemma}

\begin{proof}
Denote $\ell_1 = \{ a E_{11} \, : \, a \in {\oR} \}$ and $\ell_2 = \{  \infty E_{22}  + a E_{11} \, : \, a \in {\oR} \}$.
For every $a \in {\oR}$, the point $\infty E_{22}  + a E_{11}$ is the unique point on $\ell_2$ that is coherent to $aE_{11} \in \ell_1$. Hence, there is exactly one line passing through $aE_{11}$ that intersects $\ell_2$ and this line is $\{ a E_{11} + b E_{22} \, : \, b \in {\oR} \}$. Consequently, the union of all lines that intersect  $\ell_1$ and $\ell_2$ is
$\{ a E_{11} + b E_{22} \, : \, a,b \in {\oR} \}$.
\end{proof}

We will call the set 
\begin{equation}\label{ds}
\{ a E_{11} + b E_{22} \, : \, a,b \in {\oR} \}
\end{equation}
the \emph{diagonal surface}\index{diagonal surface}.
A direct consequence of Lemmas \ref{herz} and \ref{standsurface} is that any surface is a $\varphi$-image of the diagonal surface for some standard automorphism $\varphi$ of $\oH$.

\begin{corollary}\label{ilja}
Let $\Pi \subset \oH$ be a surface. 
\begin{itemize}
\item Assume that $A,B \in \Pi$ satisfy $d(A,B) = 1$. Then the line $\ell_{A, B}$ passing through $A$ and $B$ is contained in $\Pi$. 
\item Assume that $\ell_1$ and $\ell_2$ are lines contained in $\Pi$ and $\ell_1 \cap \ell_2 = \emptyset$. Then $\Pi$ is the surface along $\ell_1$ and $\ell_2$.
\end{itemize}
\end{corollary}

\begin{proof}
Without loss of generality, we may assume that $\Pi$ is the diagonal surface. Then the verification of both statements is easy.
\end{proof}

\begin{lemma}\label{svinj}
Let $Q \in \mathcal{P}$. Assume that $\Pi \subset \oH$ is a surface containing $0$, $(1/2)I$, $I$, and $Q$. Then
\[
\Pi = \{ a Q + b Q^\perp \, : \, a,b \in {\oR} \}.
\]
\end{lemma}

\begin{proof}
We have $0,Q \in \Pi$ and therefore, by Corollary \ref{ilja} the line $\{ a Q  \, : \, a \in {\oR} \}$ is contained in $\Pi$. In particular, $(1/2)Q\in \Pi$. 
It follows that the line $\ell_1 := \{Q+aQ^\perp\,:\, a\in {\oR}\}$ passing through $Q$ and $I$, and the line $\ell_2 := \{(1/2)Q+aQ^\perp\,:\, a\in {\oR}\}$ passing through $(1/2)Q$ and $(1/2)I$, are both contained in $\Pi$. 
Observe that the surface $\{ a Q + b Q^\perp \, : \, a,b \in {\oR} \}$ also contains $\ell_1$ and $\ell_2$.
Since $\ell_1\cap \ell_2=\emptyset$, the second item of Corollary \ref{ilja} leads to the desired conclusion.
\end{proof}

Let $P \in \mathcal{P}$ and consider the set
\begin{equation}\label{sp}\index{$\square_P$}
{\square}_P = \{ a P + b P^\perp \, : \, a,b \in [0,1] \}.
\end{equation}
It is contained in the surface $\{ a P + b P^\perp \, : \, a,b \in {\oR} \}$.

\begin{lemma}\label{segment}
Let $P \in \mathcal{P}$ and $A, B\in {\square}_P$ satisfy $d(A, B)=1$. 
Then every point of ${\square}_P$ is coherent to some point of $\ell_{A, B} \cap {\square}_P$. 
\end{lemma}
\begin{proof}
By the assumption, we see that one of the following holds. 
Either there are $p,q,r\in [0,1]$ with $A=pP+qP^\perp$, $B=pP+rP^\perp$, and $q\neq r$, or there are $p,q,r\in [0,1]$ with $A=pP+qP^\perp$, $B=rP+qP^\perp$, and $p\neq r$. 
We have $\ell_{A, B}=\{pP+tP^\perp\,:\, t\in {\oR}\}$ in the former case and $\ell_{A, B}=\{tP+qP^\perp\,:\, t\in {\oR}\}$ in the latter case. 
It is now clear to see that the desired conclusion holds.
\end{proof}

\begin{lemma}\label{where}
Let $P \in \mathcal{P}$ and assume that $\varphi \colon {\square}_P \to \oH$ is a coherency preserver. Then $\varphi ({\square}_P)$ is contained in a cone or in a surface.
\end{lemma}

\begin{proof}
The first possibility we will consider is that there exists $a\in [0,1]$ such that the line segment $s_a := \{ aP + bP^\perp \, : \, b \in [0,1] \}$ is mapped to a singleton $\{ A \}$. Then clearly, $\varphi ({\square}_P)$ is contained in the cone with vertex $A$.

So, we may assume from now on that for each line segment $s_a$, $0 \le a \le 1$, there is a unique line $\ell_a \subset \oH$ such that $\varphi (s_a) \subset \ell_a$.

Assume next that there is a pair $a_1 , a_2 \in [0,1]$ such that $\ell_{a_1} \cap \ell_{a_2} =\emptyset$. We will show that in this case $\varphi ({\square}_P)$ is contained in the surface along $\ell_{a_1}$ and $\ell_{a_2}$. 
We observe that for any line segment $s'_b := \{ aP + bP^\perp \, : \, a \in [0,1] \}$, $0 \le b \le 1$, there is a line $\ell'_b \subset \oH$ such that $\varphi (s'_b) \subset \ell'_b$. Moreover, each $\ell'_b$, $0 \le b \le 1$, intersects both $\ell_{a_1}$ and $\ell_{a_2}$, as desired.

It remains to consider the case where for every pair $a_1 , a_2 \in [0,1]$ either $\ell_{a_1} = \ell_{a_2}$, or $\ell_{a_1}$ and $\ell_{a_2}$ intersect in one point. 
The desired conclusion holds trivially if $\ell_{a_1} = \ell_{a_2}$ for every pair $a_1 , a_2 \in [0,1]$. 
So, the final possibility we need to consider is that there are $a_1 , a_2 \in [0,1]$ such that  $\ell_{a_1}$ and $\ell_{a_2}$ intersect in one point $A$.
Let $0 \le a \le 1$. 
If $\ell_{a}\in\{\ell_{a_1}, \ell_{a_2}\}$, then we have $A\in \ell_a$. 
If $\ell_{a}\notin\{\ell_{a_1}, \ell_{a_2}\}$, then Corollary \ref{triangle} shows that $A\in \ell_a$. 
Thus, we get $\varphi ({\square}_P)\subset \cC_A$ in this case.
\end{proof}

We know that a surface is determined by any pair of lines that belong to the surface and do not intersect. The situation is different for pairs of lines that do intersect. 
To see this, we need the following lemma. 

\begin{lemma}\label{inftysurface}
Let $P, Q\in \cP$ be distinct elements. 
Then the points $0, \infty P, \infty Q, \oi$ are contained in no single cone, and are contained in the unique surface 
\begin{equation}\label{PI}
\Pi_{P, Q}:=\{aP+bQ\,:\, a, b\in \bR\}\cup \ell_{\infty P, \oi}\cup \ell_{\infty Q, \oi}.\index{$PiPQ$@$\Pi_{P, Q}$}
\end{equation}
\end{lemma}
\begin{proof}
To show that these points are contained in no single cone, assume that $A\in \oH$ is coherent to these points. Since $0\sim A\sim \oi$, we see that $A=\infty R$ for some $R\in \cP$. 
We also have $\infty P\sim A=\infty R\sim \infty Q$, which implies $P=R=Q$, contradicting the assumption. 

Let $\ell_1$ be the line passing through $0$ and $\infty P$, and $\ell_2$ the line passing through $\infty Q$ and $\oi$. Then $\ell_1=\{aP\,:\, a\in {\oR}\}$ and $\ell_2=\{\infty Q+ aQ^\perp\,:\, a\in {\oR}\}$. 
For each $t\in\mathbb{R}$, the unique point in $\ell_2$ that is coherent to $tP$ is $\infty Q+tP=\infty Q+t\mathrm{tr}\,(PQ^\perp)Q^\perp$. Using Lemma \ref{jojhcer}, we easily see that
the line passing through $tP$ and $\infty Q+tP$ is $\{sQ+tP\,:\, s\in \oR\} =\{tP+sQ\,:\, s\in \mathbb{R}\}\cup \{\infty Q+t\mathrm{tr}\,(PQ^\perp)Q^\perp\}$. 
The unique point in $\ell_2$ that is coherent to $\infty P$ is $\oi$. 
This completes the proof.
\end{proof}

Consider two lines $\ell_1=\ell_{0, \infty E_{11}}, \ell_2=\ell_{\infty E_{11}, \oi}$ in the diagonal surface. 
There are infinitely many surfaces that contain both $\ell_1$ and $\ell_2$. 
Indeed, for every $P\in \mathcal{P}\setminus\{E_{11}\}$, we have $0, \infty E_{11}, \oi\in \Pi_{E_{11}, P}$, which leads to $\ell_1\cup \ell_2\subset \Pi_{E_{11}, P}$.

\subsection{Three points}\label{three points}
Next, we will be interested in a relative position of any three points $A_1 , A_2 , A_3 \in \oH$ with $d(A_1 , A_2 ) = d(A_1 , A_3 ) = d(A_2 , A_3 ) = 2$.

\begin{lemma}\label{threepoints}
Let $A_1 , A_2 , A_3 \in \oH$ satisfy $d(A_1 , A_2 ) = d(A_1 , A_3 ) = d(A_2 , A_3 ) = 2$. Then one and only one of the following holds.
\begin{itemize}
\item There exists a standard automorphism $\varphi \colon \oH \to \oH$ with the property that 
\[
\varphi (A_1) = 0,\ \varphi (A_2) = \oi,\ \  \text{and}\ \  \varphi (A_3) = I.
\]
\item There exists a standard automorphism $\varphi \colon \oH \to \oH$ with the property that 
\[
\varphi (A_1) = 0,\ \ \varphi (A_2) = \oi,\ \ \text{and}\ \ \varphi (A_3) = J := \left[ \begin{matrix} 1 & 0 \cr 0 & -1 \cr \end{matrix} \right]. \index{$J$}
\]
\end{itemize}
\end{lemma}

\begin{proof}
By Lemma \ref{pocelo}, there is a standard automorphism $\psi \colon \oH \to \oH$ with the property that $\psi (A_1) = 0$ and $\psi (A_2) = \oi$. 
By $d(\psi (A_3) , \oi) = d(\psi (A_3) , \psi (A_2)) = 2$, we get $\psi (A_3)\in H_2$.
Because $d(\psi (A_3) , 0) = d(\psi (A_3) , \psi (A_1)) = 2$, the $2 \times 2$ hermitian matrix $\psi (A_3)$ is invertible.
In the case where $\psi (A_3)<0$, we replace the automorphism $\psi$ by the automorphism $X \mapsto -\psi (X)$, $X \in \oH$. Hence, we have either 
\[
\psi (A_1) = 0,\ \ \psi (A_2) = \oi,\ \ \text{and}\ \ \psi (A_3)\in H_2^{++},
\]
 or 
\[
\psi (A_1) = 0,\ \ \psi (A_2) = \oi, \ \ \text{and} \ \ \psi (A_3)\in H_2^{+-}.
\] 
In the first case, we can find an invertible $2 \times 2$ complex matrix $S$ such that $S \psi (A_3) S^\ast = I$.
In the second case, we can find an invertible $2 \times 2$ complex matrix $S$ such that $S \psi (A_3) S^\ast = J$. Then the standard automorphism $\varphi (X)= S\psi (X) S^\ast$, $X \in \oH$, has one of the above two desired properties. 

It remains to verify that one cannot find an automorphism of $\oH$ that sends $0$, ${\oi}$, and $I$, to $0$, ${\oi}$, and $J$, respectively. 
This is a consequence of the following lemma.
\end{proof}

\begin{lemma}\label{noauto}
We have 
\begin{equation}\label{empty}
\cC_0\cap \cC_{\oi}\cap \cC_I=\emptyset
\end{equation}
but
\begin{equation}\label{nempty}
 \cC_0\cap \cC_{\oi}\cap \cC_J=\{\infty P\,:\, P^\perp JP^\perp =0\}\neq \emptyset.
\end{equation}
Thus there is no automorphism $\varphi$ of $\oH$ with 
\[
\varphi(0)=0,\ \ \varphi(\oi)=\oi,\ \  \text{and}\ \ \varphi(I)=J.
\] 
\end{lemma}
\begin{proof}
Every element of $\cC_0\cap \cC_{\oi}$ is of the form $\infty P$ for some $P\in \cP$. 
It is clear that $\infty P\notin \cC_I$, so \eqref{empty} is established. 
We have $P^\perp JP^\perp=0$ for 
\[
P=\frac{1}{2}\left[ \begin{matrix} 1 & 1 \cr 1 & 1 \cr \end{matrix} \right]\in \cP,
\]
thus we obtain \eqref{nempty}.
\end{proof}

\begin{definition}
Let $A_1 , A_2 , A_3 \in \oH$ satisfy $d(A_1 , A_2 ) = d(A_1 , A_3 ) = d(A_2 , A_3 ) = 2$. We say that the triple $A_1 , A_2 , A_3$ is in \emph{timelike} position \index{timelike} if $\cC_{A_1} \cap \cC_{A_2}\cap \cC_{A_3}=\emptyset$. Otherwise, the triple $A_1 , A_2 , A_3$ is said to be in \emph{spacelike} position\index{spacelike}.
The preceding lemmas imply that the triple $A_1 , A_2 , A_3$ is in timelike position if and only if the first option of Lemma \ref{threepoints} holds.
\end{definition}
%

\begin{example}\label{IJ}
The triple $0, (1/2)I, I$ is in timelike position. Indeed, the automorphism $\varphi \colon \oH \to \oH$ given by $\varphi (X) = -I + (I-X)^{-1}$ satisfies 
\[
\varphi (0) = 0,\ \ \varphi ((1/2)I) = I,\ \  \text{and}\ \ \varphi (I) = {\oi}.
\] 
On the other hand, the triple $0, (1/2)J, J$ is in spacelike position. Indeed, we have 
\[
\varphi (0) = 0,\ \ \varphi ((1/2)J) = J,\ \ \text{and}\ \ \varphi (J) = {\oi}, 
\]
where the automorphism $\varphi \colon \oH \to \oH$ is given by $\varphi (X) = -J + (J-X)^{-1}$, $X \in \oH$.
\end{example}

\begin{remark}\label{remark02}
Let us recall that the triple $0, (1/2)I, I$ corresponds to spacetime events $(0,0,0,0)$, $(0,0,0,1/2)$, and $(0,0,0,1)$, respectively. Thus, we have the origin $(0,0,0)$ of the space $\mathbb{R}^3$ at three different moments, $t_1 =0$, $t_2 = 1/2$, and $t_3 = 1$.  
On the other hand, the triple $0, (1/2)J, J$ corresponds to spacetime events $(0,0,0,0)$, $(0,0,-1/2, 0)$, and $(0,0,-1,0)$, respectively. Thus we have three different points on the line through the origin of $\mathbb{R}^3$ all of them considered at the moment $t=0$. 
This motivates our choice of terminology above.
\end{remark}

\begin{example}\label{definite}
For an invertible $A\in H_2$, the triple $0, \oi, A$ is in timelike position if and only if $A\in H_2^{++}\cup H_2^{--}$. 
To show this, consider the image of these points by a suitable affine automorphism sending $0$ to $0$.
\end{example}

For $P\in \cP$, the triple $(1/2)P, (1/2)P^\perp, I$ is in spacelike position. Indeed, we have $\mathcal{C}_{(1/2)P} \cap \mathcal{C}_{(1/2)P^\perp} \cap \mathcal{C}_I\neq \emptyset$. 
A concrete description of the left-hand side follows.
\begin{lemma}\label{nnjk}
Let $P \in \mathcal{P}$. Then
\[
\mathcal{C}_{(1/2)P} \cap \mathcal{C}_{(1/2)P^\perp} \cap \mathcal{C}_I = \{ Q + (1/3)Q^\perp \, : \, Q\in \mathcal{P}, \, \mathrm{tr}\, (PQ) = 1/2 \}.
\]
\end{lemma}

Note that for any two projections $P,Q \in \mathcal{P}$, we have $\mathrm{tr}\, (PQ) = 1/2$ if and only if $\mathrm{tr}\, (PQ^\perp) = 1/2$.

\begin{proof}
Without loss of generality, we may assume that $P= E_{11}$. Then the set $\{ Q\in \mathcal{P} \, : \, \mathrm{tr}\, (PQ) = 1/2 \}$ is equal to
\[
\left\{ {1\over 2}\left[ \begin{matrix} {1} &  z \cr   \overline{z} &  {1}  \cr\end{matrix} \right] \ : \, z \in {\mathbb{T}} \right\}.
\]
Clearly, for every $z \in \mathbb{T}$ we have
\[
\left({1\over 2} \left[ \begin{matrix} {1} &  z \cr   \overline{z} &  {1}  \cr\end{matrix} \right] \right)^\perp =  {1\over 2} \left[ \begin{matrix} {1} &  - z \cr   -\overline{z} &  {1}  \cr\end{matrix} \right]
\]
and consequently, 
\[
 \left\{ Q + {1\over 3}Q^\perp \, : \, Q\in \mathcal{P}, \, \mathrm{tr}\, (PQ) = {1 \over 2} \right\} = 
\left\{ {1\over 3}\left[ \begin{matrix} {2} &   z \cr   \overline{z} &  {2}  \cr\end{matrix} \right] \ : \, z \in \mathbb{T} \right\} .
\]

In the next step, we will verify that 
\[
\mathcal{B} := \mathcal{C}_{(1/2)E_{11}} \cap \mathcal{C}_{{(1/2)}E_{22}} \cap \mathcal{C}_I\subset H_2.
\] Clearly, ${\oi}\not\in \mathcal{B}$. If $\infty R + bR^\perp \in \mathcal{B}$ for some real number $b$ and some $R \in \mathcal{P}$, then from 
$\infty R + bR^\perp  \sim I$ we conclude that $b=1$. But then  $\infty R + R^\perp  \sim (1/2)E_{11}$ implies $(1/2)R^\perp E_{11} R^\perp = R^\perp$. Considering the trace, we get 
a contradiction.

Thus, $\mathcal{B} \subset H_2$. 
Let $t,s$ be real numbers and $\alpha$ a complex number. 
Then 
\[
A =  
\left[ \begin{matrix}t &  \alpha \cr   \overline{\alpha} &  s  \cr\end{matrix} \right] \in \mathcal{B}
\]
holds if and only if 
\[
\det (A - (1/2)E_{11}) = \det (A - (1/2)E_{22}) = \det (A - I) = 0,
\]
which is further equivalent to
\[
ts - {1 \over 2}s = |\alpha|^2, \ \ \ ts - {1 \over 2}t = |\alpha|^2, \ \ \ \text{and} \ \ \ (t-1)(s-1)  = |\alpha|^2 .
\]
Real numbers $t,s$ and a complex number $\alpha$ satisfy the above system of equations if and only if $t=s= 2/3$ and $| \alpha | = 1/3$. Hence, $A \in \mathcal{B}$ if and only if
\[
A = {1\over 3}\left[ \begin{matrix} {2} &  z \cr   \overline{z} &  {2}  \cr\end{matrix} \right]
\] 
for some $z \in \mathbb{T}$.  
\end{proof}

\begin{lemma}\label{c(s)}
Let $A, B\in H_2$ satisfy $A<B$. 
Then $X\in H_2$ is coherent to some element of $\mathcal{S}_{A, B}$ if and only if 
\begin{equation}\label{o1}
X\notin (A, B),\ \  X\leq B,\ \ \text{and} \ \ X\not< A,
\end{equation}
or
\begin{equation}\label{o2}
X\notin (A, B),\ \ X\geq A,\ \ \text{and}\ \ X\not> B 
\end{equation}
holds.
\end{lemma}
\begin{proof}
By Lemma \ref{Porder}, we may assume that $A=0$ and $B=I$ without loss of generality.
Then we have $\mathcal{S}_{A, B}=\mathcal{S}_{0, I}=\mathcal{P}$. 
In that case, it is easy to check that \eqref{o1} or \eqref{o2} holds if and only if one of the eigenvalues of $X$ is in $[0,1]$ and the other one is in $(-\infty, 0]\cup [1, \infty)$.
We need to verify that $X$ satisfies this condition if and only if $X$ is coherent to some element of $\mathcal{P}$.

Assume first that $X \sim P$ for some $P \in \mathcal{P}$, that is, $X = P +R$ for some $R \in H_2$ of rank at most one. 
If $R\geq 0$, then $P\leq X$ and thus the larger eigenvalue of $X$ is at least $1$. 
On the other hand, for a unit vector $y\in \mathbb{C}^2$ that is perpendicular to the range of $R$, we have $(Xy,y)=( (P+R) y,y) = (Py,y) \in [0,1]$.
This together with $X\geq 0$ implies that the smaller eigenvalue of $X$ is in $[0,1]$. 
Similarly, $R\leq 0$ implies that one of the eigenvalues of $X$ is in $[0,1]$ and the other is in $(-\infty, 0]$. 

To prove the other direction, assume that one of the eigenvalues of $X$ is in $[0,1]$ and the other is in $(-\infty, 0]\cup [1, \infty)$. 
Then $X=aQ+bQ^\perp$ for some $Q\in \cP$, $a\in [0,1]$, and $b\in (-\infty, 0]\cup [1, \infty)$.
It follows that $\det(X-Q)\cdot \det(X-Q^\perp)\leq 0$. 
Combining this with the fact that $\cP$ is connected, we see that there is $P\in \cP$ satisfying $\det(X-P)=0$, as desired.
\end{proof}

\begin{corollary}\label{matrixtriple}
Let $A, B\in H_2\subset \oH$ satisfy $A<B$. 
Let $X\in H_2$ be a matrix with $d(A, X)=d(B,X)=2$. 
Then the triple $A, B, X$ is in spacelike position if and only if \eqref{o1} or \eqref{o2} holds.
\end{corollary}
From this corollary, we see for example that a triple $A,B, C\in H_2$ with $A<C<B$ is in timelike position.

\begin{corollary}\label{SJS}
Let $A\in H_2\cap \mathcal{S}_{0, J}$. Then the triple $(1/2)A, (1/2)(J-A), J$ is in timelike position.
\end{corollary}
\begin{proof}
From $0\sim A\sim J$, we get $0\sim J-A\sim J$. 
It follows that there are projections $P, Q\in \cP$ and $a,b\in \mathbb{R}$ such that $A=aP$ and $J-A=bQ$.
We get $aP+bQ=J$. By the definition of $J$, we see that $P\neq Q$ and $0 = {\rm tr}\, J = {\rm tr}\, (aP + bQ) = a+b$. Thus, $a=-b$. If $a < 0$, then $aI \le aP \le aP + bQ = J$, and hence, $a\leq -1<1\leq b$. Similarly, if $a>0$, then $b\leq -1<1\leq a$.

Assume that $a\leq -1<1\leq b$. Then we have $A=aP<bQ=J-A$ by Lemma \ref{apbq}. 
Moreover, since $J-(1/2)A=(1/2)aP+bQ$ and $J-(1/2)(J-A)=aP+(1/2)bQ$, Lemma \ref{apbq} also implies that $J-(1/2)A, J-(1/2)(J-A)\in H_2^{+-}$.
Thus the preceding corollary implies that the triple $(1/2)A, (1/2)(J-A), J$ is in timelike position.
Similarly, if $b\leq -1<1\leq a$, then we see that the triple $(1/2)(J-A), (1/2)A, J$ is in timelike position.
\end{proof}

\subsection{Identity-type theorem}\label{identity-type theorem}\index{identity-type theorem}
The set $U_2$ of $2\times 2$ unitaries is endowed with the usual topology as a closed subset of a finite-dimensional vector space.   
We endow $\oH$ with the topology that comes from the identification with $U_2$ as in Subsection \ref{oHU2}. 
From the discussion there, it is easy to see that $H_2\subset \oH$ is open and that the relative topology given to the subset $H_2\subset \oH$ coincides with the usual topology of $H_2$.

\begin{lemma}\label{srtsrt}
For any pair $A, B\in \oH$, there is a homeomorphic automorphism $\varphi\colon \oH\to\oH$ satisfying $\varphi(A)=B$.
\end{lemma}
\begin{proof}
This can be easily seen by considering $U_2$ instead of $\oH$. 
Indeed, if $V,V'\in U_2$, then the mapping $\varphi\colon U_2\to U_2$ defined by $U\mapsto V'V^*U$, $U\in U_2$, is clearly a homeomorphism satisfying $\varphi(V)=V'$ and 
\[
U\sim U'\iff \varphi(U)\sim \varphi(U')
\] 
for every pair $U, U'\in U_2$.
\end{proof}

\begin{remark}
In fact, it turns out that every automorphism of $\oH$ is a homeomorphism, but we will not need such a general fact.
\end{remark}

\begin{lemma}
The interior of a cone in $\oH$ is empty. 
Let $\ell\subset \oH$ be a line and $\mathcal{U}\subset \oH$ an open subset. 
If $\ell\cap \mathcal{U}$ is nonempty, then it has infinitely many points.  
\end{lemma}
\begin{proof}
Observe that a unitary $U\in U_2$ is coherent to $I$ if and only if one of the eigenvalues of $U$ is $1$. 
It is easily seen that the set of all such unitaries has empty interior in $U_2$. 
This together with Lemma \ref{srtsrt} proves the first statement.
Let us show the second statement. 
By Lemma \ref{srtsrt} again, it suffices to consider the case $0\in \ell\cap \mathcal{U}$. 
Then $\ell=\{tP\,:\, t\in \oR\}$ for some $P\in \cP$, and $H_2\cap \mathcal{U}$ is an open neighborhood of $0$ in $H_2$. It follows that $\ell\cap \mathcal{U}$ has infinitely many points.
\end{proof}

\begin{lemma}\label{ua}
Let $\mathcal{U}\subset \mathcal{A}\subset \oH$ and assume that $\mathcal{U}$ is open in $\oH$. 
Let $\varphi\colon \mathcal{A}\to \oH$ be a coherency preserving map. 
Assume that $\varphi(X)=X$ for every $X\in \mathcal{U}$. 
If $A\in \mathcal{A}$ satisfies $\mathcal{C}_A \cap \mathcal{U}\neq \emptyset$, then $\varphi(A)=A$. 
\end{lemma}
\begin{proof}
Let $A\in \mathcal{A}$ satisfy $\mathcal{C}_A \cap \mathcal{U}\neq \emptyset$.
Since $\mathcal{U}$ is open and the interior of a line is empty, we may take a pair of distinct lines $\ell_1$, $\ell_2$ passing through $A$ such that $\ell_1\cap \mathcal{U}$ and $\ell_2\cap \mathcal{U}$ are nonempty. 
Then, by the openness of $\mathcal{U}$ again we see that $\ell_1\cap \mathcal{U}$ and $\ell_2\cap \mathcal{U}$ have infinitely many points. 
Observe that $\varphi(\ell_1\cap \mathcal{A})$ and $\varphi(\ell_2\cap \mathcal{A})$ are coherent sets containing $\varphi(\ell_1\cap \mathcal{U})=\ell_1\cap \mathcal{U}$ and $\varphi(\ell_2\cap \mathcal{U})=\ell_2\cap \mathcal{U}$, respectively. 
It follows that $\varphi(\ell_1\cap \mathcal{A})\subset \ell_1$ and $\varphi(\ell_2\cap \mathcal{A})\subset \ell_2$. 
Thus we conclude that 
\[
\varphi(A)\in \varphi(\ell_1\cap \mathcal{A})\cap \varphi(\ell_2\cap \mathcal{A})\subset \ell_1\cap \ell_2=\{A\}.
\] 
\end{proof}

\begin{lemma}\label{smtq}
Let $A,B,C \in H_2$ satisfy $A < B < C$. Let $\varphi \colon [A,C ] \to \oH$ be a coherency preserver. Assume that $\varphi (X) = X$ for every $X \in [A,B]$. Then $\varphi (X) = X$ for every $X \in [A,C]$.
\end{lemma}
\begin{proof}
By Lemma \ref{Porder}, we may, and we will assume that $A=0$ and $C= I$ without loss of generality.
Hence, we are considering $B \in H_2$ with $0 < B <  I$ and a coherency preserver $\varphi \colon [0,I ] \to \oH$ satisfying $\varphi (X) = X$ for every $X \in [0,B]$.
Let $t_0$ be the smaller eigenvalue of $B$. Then $0 < t_0 <1$. 

We consider the subset $\mathcal{U}_1:=(0, t_0I)\subset [0, I]$, which is open in $\oH$. 
Note that $\mathcal{U}_1\subset [0,B]$.
If one of the eigenvalues of a matrix $A\in [0, I]$ lies in $(0, t_0)$, then $\mathcal{C}_A\cap \mathcal{U}_1$ is clearly nonempty, so Lemma \ref{ua} implies $\varphi(A)=A$. 
Next we define the open subset $\mathcal{U}_2$ of $\oH$ to be the set formed of all matrices $X$ in $(0, I)$ such that $\sigma(X)\cap (0, t_0)\neq \emptyset$.
We already know that $\varphi(X)=X$ for every $X\in \mathcal{U}_2$. 
If $A\in [0, I]\setminus (\{0, I\} \cup \mathcal{P})$, then it is easy to see that $\mathcal{C}_A\cap \mathcal{U}_2$ is nonempty, thus we obtain $\varphi(A)=A$ again by Lemma \ref{ua}.
We know $\varphi(0)=0$ from the assumption.
From $\varphi(P+tP^\perp)=P+tP^\perp$ for every $t\in (0, 1)$ and $P\in \mathcal{P}$,  it clearly follows that $\varphi(I)=I$. (Indeed, we have $\varphi(I)\in \ell_{P+tP^\perp, P+sP^\perp} = \ell_{P, I}$ for every $0<t<s<1$ and $P\in \cP$.) Finally, for $P \in \mathcal{P}$ we know that $\varphi (P) \in {\mathcal{S}}_{\varphi(0), \varphi(I)}= \mathcal{P}$ and $\varphi(P) \sim \varphi((1/2)P)=(1/2)P$, and therefore $\varphi (P) =P$.
\end{proof}

\begin{corollary}\label{smtq1}
Let $A,B,C \in H_2$ with $A < B < C$. Let $\varphi \colon [A,C ] \to \oH$ be a coherency preserver. Assume that $\varphi (X) = X$ for every $X \in [B,C]$. Then $\varphi (X) = X$ for every $X \in [A,C]$.
\end{corollary}

\begin{proof}
We define a coherency preserver $\varphi_1 \colon [0, C-A] \to \oH$ by
\[
\varphi_1 (X) = C - \varphi (C-X), \ \ \ X \in [0, C-A].
\]
It is straightforward to check that the restriction of $\varphi_1$ to $[0, C-B]$ is the identity map. By Lemma \ref{smtq}, $\varphi_1$ is the identity on the whole matrix interval $[0, C-A]$, which is equivalent to $\varphi (X) = X$ for every $X \in [A,C]$.
\end{proof}

In the proof of the next corollary, we will use the following simple fact: Let $\mathcal{U} \subset H_2$ be an open subset and $A \in \mathcal{U}$. Then there exist $B,C \in H_2$ such that $B < A < C$ and $[B,C] \subset \mathcal{U}$. Indeed, we only need to observe that for each $\varepsilon > 0$ we have $[ A- \varepsilon I, A + \varepsilon I ] = \{ D \in H_2 \, : \, \| D-A \| \le \varepsilon \}$.

\begin{corollary}\label{smtq2}
Let $A,B\in H_2$ with $A < B$. Let $\varphi \colon [A,B] \to \oH$ 
be a coherency preserver. Assume that 
the set $\{X\in [A, B]\,:\, \varphi (X) = X\}$ as a subset of $H_2$ has nonempty interior.
Then $\varphi (X) = X$ for every $X \in [A,B]$.
\end{corollary}

\begin{proof}
By the assumption together with the remark above, we may take $C, D\in (A, B)$ with $C<D$ such that $\varphi(X)=X$ for every $X\in [C, D]$. 
By Lemma \ref{smtq}, $\varphi$ is the identity on the matrix interval $[C, B]$. By Corollary \ref{smtq1}, $\varphi$ is the identity on the whole matrix interval $[A,B]$.
\end{proof}

\begin{corollary}\label{kiakia}
Let $A,B\in H_2$ with $A < B$. Let $\varphi \colon [A,B] \to \oH$ 
be a coherency preserver. 
Assume that there exists an automorphism $\psi$ of $\oH$
such that the set $\{X\in [A, B]\,:\, \varphi (X) = \psi (X) \}$ as a subset of $H_2$ has nonempty interior.
Then $\varphi (X) = \psi (X)$ for every $X \in [A,B]$.
\end{corollary}

\begin{proof}
After introducing a new coherency preserver defined by
\[
X \mapsto \psi^{-1} (\varphi (X)), \ \ \ X \in [A,B],
\]
we may assume with no loss of generality that $\psi$ is the identity. Then the desired conclusion follows directly from Corollary \ref{smtq2}.
\end{proof}

We continue with an identity-type theorem\index{identity-type theorem}. 

\begin{theorem}\label{identity}
Let $\mathcal{A}$ be a nonempty connected open subset of $\oH$ and let $\varphi \colon \mathcal{A} \to \oH$ be a coherency preserving map. Let $\mathcal{D}$ be a nonempty open subset of $\mathcal{A}$. If an automorphism $\psi$ of $\oH$
satisfies $\varphi (X) = \psi (X)$ for every $X \in \mathcal{D}$, then $\varphi (X) = \psi (X)$ for every $X \in \mathcal{A}$.
\end{theorem}

\begin{proof}
We denote by $\mathcal{L} \subset \mathcal{A}$ the set of all points $A \in \mathcal{A}$ with the property that 
$\varphi (X) = \psi(X)$ for every $X$ in some neighborhood of $A$. It is trivial to see that $\mathcal{D}\subset \mathcal{L}$ and that $\mathcal{L}$ is an open subset of $\mathcal{A}$.
It remains to prove that $\mathcal{L}$ is a closed subset of $\mathcal{A}$. 

Hence, assume that $A \in \mathcal{A}$ and that there exists a sequence $A_n \in \mathcal{L}$, $n\geq 1$, such that $\lim A_n = A$. 
By Lemma \ref{srtsrt}, there is a homeomorphic automorphism $\psi_1\colon \oH\to\oH$ satisfying $\psi_1(A)=0$. 
Thus $\psi_1(A_n)$ tends to $0$ as $n\to\infty$. 
Let $\varepsilon$ be a positive number such that $[- \varepsilon I, \varepsilon I] \subset \psi_1(\mathcal{A})$. There exists a positive integer $m$ such that $\psi_1(A_m) \in 
(- \varepsilon I, \varepsilon I)$. 
By Corollary \ref{kiakia}, we have 
 $\varphi\circ \psi_1^{-1} (X) = \psi\circ \psi_1^{-1}(X)$ for every $X \in [- \varepsilon I, \varepsilon I]$. Thus, we get $A\in \mathcal{L}$. 
This completes the proof.
\end{proof}

\begin{remark}
The assumption of openness in Theorem \ref{identity} is essential. 
For example, since $tI\not\sim sI$ for distinct $s,t\in \mathbb{R}$, we see that any mapping from  $\{tI\,:\, t\in \mathbb{R}\}$ to $\oH$ preserves coherency even when $\varphi(tI)=tI$ for every $t>0$.

The assumption of connectedness is also essential. 
For example, consider the mapping $\varphi\colon H_2^{++}\cup H_2^{--}\to H_2$ defined by 
\[
\varphi(X)=X, \ \ X\in H_2^{++},\ \ \text{and}\ \ \varphi(H_2^{--})=\{0\}. 
\]
Since $X\not\sim Y$ for any $X\in H_2^{++}$ and $Y\in H_2^{--}$, we see that $\varphi$ preserves coherency.
The identity mapping on $\oH$ is an automorphism that extends $\varphi|_{H_2^{++}}$, but $\varphi$ is not the identity mapping on $H_2^{--}$.   
\end{remark}

We close this section with two lemmas whose motivation is close to the identity theorem. 

\begin{lemma}\label{sss}
Let $A, B\in \oH$ satisfy $d(A, B)=2$. 
If $C\in \oH$ has the property $C\sim X$ for every $X\in \mathcal{S}_{A, B}$ then $C\in \{A, B\}$.
\end{lemma}

\begin{proof}
Without loss of generality, we may assume $A=0$ and $B=\oi$. 
In this case, we have $\mathcal{S}_{A, B}=\{\infty P\,:\, P\in \mathcal{P}\}$, and it is easy to get to the desired conclusion.
\end{proof}

\begin{lemma}\label{smtr}
Let $A \in H_2$, $0 < A \le I$, and $B \in \oH$. Assume that for every $P \in \mathcal{P}$ and every real $s \in [0,1]$ we have
\[
A \sim sP \Rightarrow B \sim sP.
\]
Then $B=A$ or $B=0$.
\end{lemma}

\begin{proof}
By Lemma \ref{sss}, it suffices to show that $\mathcal{S}_{0, A}\subset \{sP\,:\, s\in [0,1], P\in \mathcal{P}\}$. 
Using the standard automorphism $X\mapsto A^{1/2} XA^{1/2}$, $X\in \oH$, we see
\[
\mathcal{S}_{0, A}= \{A^{1/2}PA^{1/2}\,:\, P\in \mathcal{S}_{0, I}\}.
\]
Observe that for $P\in \mathcal{S}_{0, I}=\mathcal{P}$ the norm of the rank-one positive operator $A^{1/2}PA^{1/2}$ is at most $1$,
which leads to the desired claim.
\end{proof}


\part{Main results and their proofs}

\section{Standard coherency preservers}\label{standard}
Let $\mathcal{A}$ be a subset of $\oH$ and $\varphi \colon \mathcal{A} \to \oH$ a map.  We say that $\varphi$ is a \emph{standard coherency preserver} \index{standard coherency preserver} if $\varphi$ extends to a standard automorphism of $\oH$. 
The main result of the current section is the following one.
\begin{theorema}\index{Theorem A}
Let $\mathcal{A}$ be a connected open subset of $\oH$ and $\varphi \colon \mathcal{A} \to \oH$ a coherency preserving map. Assume that there exist $A,B,C \in \mathcal{A}\cap H_2$ such that 
\begin{itemize}
\item $A <C< B$,
\item $[A,B] \subset \mathcal{A}$, and 
\item the triple $\varphi (A), \varphi (C) , \varphi (B)$ is in timelike position.
\end{itemize}
Then $\varphi$ is standard.
\end{theorema}

\subsection{Proof of Theorem A}\label{ptand}
We first show the following proposition.
\begin{proposition}\label{unitary}
Let $\varphi\colon [0, I]\to \oH$ be a coherency preserver satisfying 
\begin{equation}\label{mid}
\varphi (0) = 0, \ \ \ \varphi ((1/2)I) = (1/2)I, \ \ \ \text{and} \ \ \ \varphi (I) = I.
\end{equation}
Then there exists a unitary $U\in U_2$ satisfying either 
\[
\varphi(X)=UXU^*,\ \ \ X\in [0, I],
\]
or 
\[
\varphi(X)=UX^tU^*, \ \ \ X\in [0, I].
\]
\end{proposition}
To prove this, we assume that $\varphi\colon [0, I]\to \oH$ is a coherency preserver satisfying \eqref{mid}.
By Lemma \ref{veryeasy}, we have $\varphi (\mathcal{P}) \subset \mathcal{P}$.
Lemma \ref{veryeasy} also shows that the triple $0, (1/2)I, I$ is not contained in any cone.

We fix $P \in \mathcal{P}$. Since $0, (1/2)I, I \in \varphi ({\square}_P)$, where ${\square}_P$ is defined as in \eqref{sp}, Lemma \ref{where} yields that $\varphi ({\square}_P)$ is contained in some surface. By Lemma \ref{svinj}, we have
\[
\varphi ({\square}_P) \subset 
\{ a \varphi (P) + b \varphi (P)^\perp \, : \, a,b \in {\oR} \}.
\]

\begin{claim}\label{0121}
The equation
\begin{equation}\label{iris}
\varphi (aP + bP^\perp ) = a \varphi (P) + b \varphi (P)^\perp
\end{equation}
holds for any $a, b\in \{0, 1/2, 1\}$.
\end{claim}
\begin{proof}
We have $\varphi ( (1/2)P) \sim \varphi ((1/2)I) =(1/2)I$ and $\varphi ((1/2)P) \sim \varphi (0) = 0$, and by a simple modification of Lemma \ref{veryeasy} we see that $\varphi ( (1/2)P) = (1/2)Q$ for some $Q \in \mathcal{P}$. But $\varphi ( (1/2)P) \sim \varphi (P)\in \mathcal{P}$ and therefore,
\[
\varphi ((1/2)P) = (1/2)\varphi (P).
\]
We continue with $\varphi ((1/2)P + P^\perp)$. By
\[
(1/2)P + P^\perp \sim (1/2)I, I, (1/2)P,
\] 
we get
\[
\varphi ((1/2)P + P^\perp) \sim (1/2)I, I, (1/2) \varphi (P).
\]
Using Corollary \ref{veryeasy2}, we conclude that
\[
\varphi ((1/2)P + P^\perp) = (1/2) \varphi (P)+ \varphi(P)^\perp.
\]
Now,
\[
P^\perp \sim 0, I, (1/2)P + P^\perp
\]
yields
\[
\varphi(P^\perp) \sim 0, I, (1/2) \varphi (P)+ \varphi(P)^\perp.
\]
Since $\ell_{I, (1/2) \varphi (P)+ \varphi(P)^\perp} = \{t\varphi(P)+\varphi(P)^\perp\,:\,  t\in \oR\}$, we get
\[
\varphi (P^\perp) = \varphi (P)^\perp.
\]
Using the same arguments as above, we see that
\[
\varphi ((1/2)P^\perp) = (1/2)\varphi (P)^\perp \ \ \ \text{and} \ \ \  \varphi (P + (1/2)P^\perp) = \varphi (P) + (1/2)\varphi (P)^\perp .
\]
\end{proof}

\begin{claim}\label{013}
Equation \eqref{iris} holds for any $a, b\in \{0, 1/3, 1\}$. 
\end{claim}
\begin{proof}
Take $Q \in \mathcal{P}$ with $\mathrm{tr}\, (PQ) = 1/2$. 
Claim \ref{0121} implies that 
\[
\varphi ((1/2)Q) = (1/2)\varphi (Q),\ \ \varphi ((1/2)Q^\perp) = (1/2)\varphi (Q)^\perp.
\] 
By Lemma \ref{nnjk}, we have
\[
P + (1/3)P^\perp \in 
\mathcal{C}_{(1/2)Q} \cap \mathcal{C}_{(1/2)Q^\perp} \cap \mathcal{C}_I.
\]
Thus we obtain
\begin{equation}\label{sanbunnoichi}
\begin{split}
\varphi ( P + (1/3) P^\perp) &\in \mathcal{C}_{(1/2)\varphi (Q)} \cap \mathcal{C}_{(1/2)\varphi (Q)^\perp} \cap \mathcal{C}_I\\
&= \{ R + (1/3) R^\perp \, : \, R \in \mathcal{P}, \, \mathrm{tr}\, (\varphi (Q)R) = 1/2 \}.  
\end{split}
\end{equation}
Because $P + (1/3) P^\perp\in \ell_{P, I}$, we have 
\begin{equation}\label{phip}
\varphi ( P + (1/3) P^\perp)\in \ell_{\varphi(P), \varphi(I)}=\ell_{\varphi(P), I}  =
\{ \varphi (P) + a\varphi(P)^\perp   \, : \, a \in  {\oR} \}.
\end{equation}
It follows from \eqref{sanbunnoichi} and \eqref{phip} that
\[
\varphi (P + (1/3) P^\perp ) = \varphi (P) + (1/3) \varphi(P)^\perp,
\]
and similarly,
\[
\varphi (P^\perp + (1/3) P ) = \varphi (P)^\perp + (1/3) \varphi(P) .
\]
Further, $\varphi ((1/3)P)\in \ell_{0, \varphi(P)}$
and $\varphi ((1/3)P)   \sim  \varphi (P)^\perp + (1/3) \varphi(P)$. 
It follows from Lemma \ref{simplll} that 
\[
\varphi ( (1/3) P ) = (1/3) \varphi(P)  \ \ \ \text{and similarly,} \ \ \ \varphi ( (1/3) P^\perp ) = (1/3) \varphi(P)^\perp.
\]
Finally, $\varphi ((1/3)I)$ belongs to the lines 
\[
\ell_{(1/3)\varphi(P), (1/3)\varphi(P) + \varphi (P)^\perp} \not=\ell_{(1/3)\varphi(P)^\perp, (1/3)\varphi(P)^\perp + \varphi (P)}
\]
which intersect in the point $(1/3)I$. Thus,
$\varphi ((1/3)I) = (1/3)I$.
\end{proof}

In what follows, we sometimes make use of the coherency preserving map $\varphi_1 \colon [0,I] \to \oH$ defined by $\varphi_1 (X) = I - \varphi ( I -X)$, $X \in [0, I]$. Clearly, 
\[
\varphi_1 (0) = 0,\ \ \varphi_1 ((1/2)I) = (1/2)I,\ \ \varphi_1 (I) = I,\ \ \text{and}\ \ \varphi_1 (P) = \varphi (P).
\] 

\begin{claim}\label{0123}
Equation \eqref{iris} holds for any $a, b\in \{0, 1/3, 2/3, 1\}$. 
\end{claim}
\begin{proof}
Applying Claim \ref{013} to the mapping $\varphi_1$ in place of $\varphi$, we have \eqref{iris} for every $a, b\in \{0, 2/3, 1\}$.

Let $a,b \in \{ 0, 1/3, 2/3, 1\}$. 
We have shown that \eqref{iris} holds when one of the numbers $a,b$ equals $0$ or $1$. 
In the remaining cases, we have $\{aP + bP^\perp\}=\ell_{aP, aP+P^\perp}\cap \ell_{bP^{\perp}, P+bP^\perp}$. Thus 
\[
\begin{split}
\varphi(aP + bP^\perp)&\in\ell_{\varphi(aP), \varphi(aP+P^\perp)}\cap \ell_{\varphi(bP^{\perp}), \varphi(P+bP^\perp)}\\
&=\ell_{a\varphi(P), a\varphi(P)+\varphi(P)^\perp}\cap \ell_{b\varphi(P)^{\perp}, \varphi(P)+b\varphi(P)^\perp}\\
&= \{a\varphi(P) + b\varphi(P)^\perp\},
\end{split}
\]
and consequently, $\varphi(aP + bP^\perp)=a\varphi(P) + b\varphi(P)^\perp$.
\end{proof}

\begin{claim}
For every positive integer $n$, equation \eqref{iris} holds for any 
\[
a,b \in \left\{ 0, {1 \over 3}  \left({2 \over 3}\right)^{n-1}, \left({2 \over 3}\right)^{n} \right\}.
\]
\end{claim}

\begin{proof}
The case $n=1$ is true by the previous claim. So, assume that the above is true for some positive integer $n$. 
We introduce a new coherency preserving map $\tau \colon [0, I ] \to \oH$ by
\[
\tau (X) = \left (3/2 \right)^n \varphi \left(\left(2/3\right)^n X \right), \ \ \ X \in [0,I].
\]
Using the induction hypothesis, we see that
\[
\tau (0) = 0, \ \ \ \tau ((1/2)I) = (1/2)I, \ \ \ \text{and} \ \ \ \tau (I) = I.
\]
Moreover, $\tau (P) = \varphi (P)$.
But then we already know that
\[
\tau (a' P + b' P^\perp ) = a' \varphi (P) + b' \varphi (P)^\perp
\]
for any $a',b' \in \{ 0, 1/3 , 2/3 \}$. This immediately implies that \eqref{iris} holds for any 
\[
a,b \in \left\{ 0, {1 \over 3}  \left({2 \over 3}\right)^{n}, \left({2 \over 3}\right)^{n+1} \right\},
\]
as desired.
\end{proof}

\begin{claim}
For every positive integer $n$, we have
\begin{equation}\label{sita}
\varphi \left( \left(1-\left(\frac{2}{3}\right)^n \right) P \right) = \left(1-\left(\frac{2}{3}\right)^n\right) \varphi ( P )
\end{equation}
and
\begin{equation}\label{ue}
\varphi \left( P +\left(\frac{2}{3}\right)^n  P^\perp \right) = \varphi (P) + \left(\frac{2}{3}\right)^n \varphi (P)^\perp.
\end{equation}
\end{claim}

\begin{proof}
We know that
\begin{equation}\label{sipl}
\varphi ( (2/3)^n P ) = (2/3)^n \varphi (P) \ \ \ \text{and}  \ \ \ \varphi ( (2/3)^n P^\perp ) = (2/3)^n \varphi(P)^\perp
\end{equation}
for every positive integer $n$. 

Using \eqref{sipl} for the coherency preserving map $\varphi_1$, we arrive at
\[
\varphi ( P + (1-(2/3)^n ) P^\perp ) = \varphi (P) + (1-(2/3)^n) \varphi (P)^\perp 
\]
and 
\[
\varphi ( (1 -(2/3)^n ) P + P^\perp ) = (1 - (2/3)^n ) \varphi (P) + \varphi(P)^\perp
\]
for every positive integer $n$. 

We now use the fact that 
\[
\varphi ( (1-(2/3)^n ) P ) \sim 0,\, \varphi (P),\, \varphi ( (1 -(2/3)^n ) P + P^\perp )
\]
combined with Lemma \ref{simplll} to deduce that
\[
\varphi ( (1-(2/3)^n ) P ) = (1-(2/3)^n) \varphi ( P ) 
\]
and similarly, 
\[
\varphi ( (1 - (2/3)^n ) P^\perp ) = (1 - (2/3)^n ) \varphi(P)^\perp
\]
for every positive integer $n$. 

The last equation remains true if we replace $\varphi$ by $\varphi_1$.  This yields
\[
\varphi ( P +(2/3)^n  P^\perp ) = \varphi (P) + (2/3)^n \varphi (P)^\perp 
\]
for every positive integer $n$. 
\end{proof}

\begin{claim}
There is a $2\times 2$ unitary matrix $U\in U_2$ such that 
\[
\text{either}\ \ \varphi (P) = UPU^\ast,\ \  P \in \mathcal{P},\ \ \text{or}\ \  \varphi (P) = UP^tU^\ast,\ \ P \in \mathcal{P}.
\]
\end{claim}
\begin{proof}
Let $a$ be a positive real number $< 1/2$. Using Lemmas \ref{mcmcv} and \ref{mikar}, we see that a pair $P,Q \in \mathcal{P}$ satisfies
\[
([0,I]\supset) \,\, \mathcal{S}_{(1-a)P, P + aP^\perp} \cap \mathcal{S}_{(1-a)Q, Q + aQ^\perp} \not= \emptyset
\]
if and only if
\[
\| P - Q \| \le { a \over 1-a } .
\]
Therefore, the two equations \eqref{sita} and \eqref{ue} imply that for any positive integer $n > 2$, we have
\[
\| P -Q \| \le { (2/3)^n \over 1 - (2/3)^n } \Rightarrow \| \varphi (P ) - \varphi (Q) \| \le { (2/3)^n \over 1 - (2/3)^n }.
\]
Applying \eqref{sooce}, we conclude that there exists a sequence of positive real numbers $a_n$, $n\geq 1$, such that $\lim a_n = 0$ and
\[
d^g (P,Q) \le a_n \Rightarrow d^g (\varphi (P) , \varphi (Q)) \le a_n
\]
for every positive integer $n$.

Let $P,Q \in \mathcal{P}$, $P \not= Q$, and $ \varepsilon >0 $. Take a positive integer $n$ such that $a_n <  \varepsilon$. Find a positive integer $k$ such that
\[
d^g (P,Q) \le k a_n < d^g (P,Q) + \varepsilon.
\]
By Lemma \ref{P}, we can find rank one projections $P_0 = P, P_1 , \ldots, P_{k-1}, P_k = Q$ such that $d^g (P_{j-1},P_j) \le a_n$. But then
\[
\begin{split}
d^g (\varphi (P),\varphi (Q)) 
&\le d^g (\varphi (P_0),\varphi (P_1)) + \cdots +  d^g (\varphi (P_{k-1}),\varphi (P_k)) \\
& \leq  k a_n < d^g (P,Q) + \varepsilon.
\end{split}
\]
Therefore,
\[
d^g (\varphi (P),\varphi (Q)) \le d^g (P,Q), \ \ \ P,Q \in \mathcal{P}.
\]
Applying \eqref{sooce} once more, we conclude that
\[
\| \varphi (P) - \varphi (Q) \| \le \| P - Q \|, \ \ \ P,Q \in \mathcal{P}.
\]
Finally, it follows from Lemma \ref{matrum} and $\varphi (P^\perp) = \varphi(P)^\perp$ that 
\[
\| \varphi (P) - \varphi (Q) \| = \| P - Q \|, \ \ \ P,Q \in \mathcal{P}.
\]
Now, Wigner's theorem (Theorem \ref{wigner}, see also Remark \ref{wignerremark}) completes the proof. 
\end{proof}

Hence, after composing $\varphi$ with the standard automorphism of $\oH$ given by $X \mapsto U^\ast X U$, 
and possibly with the standard automorphism of $\oH$ given by $X \mapsto X^t$, we get to the situation where
\[
\varphi (P) = P, \ \ \ P \in \mathcal{P}.
\]

The following claim finally completes the proof of Proposition \ref{unitary}.

\begin{claim}
In the above situation, we have $\varphi (A) = A$ for every $A\in [0, I]$.
\end{claim}
\begin{proof}
Fix $P \in \mathcal{P}$ and a real number $t$, $1/3 \le t \le 1$. 
By $0\sim tP\sim P$, $\varphi(0)=0$, and $\varphi(P)=P$, there exists $s \in {\oR}$ such that $\varphi (tP) = sP$. 
It follows from Claim \ref{013} that for every $Q \in \mathcal{P}$ we have 
\[
\varphi ( (1/3)Q + Q^\perp) = (1/3)Q + Q^\perp.
\]
Further,
\[
(1/3)Q + Q^\perp -tP = ((1/3)Q + Q^\perp)\, (I - (3Q + Q^\perp) (tP)).
\]
Therefore, by Lemma \ref{tr1}, the condition $(1/3)Q + Q^\perp  \sim tP$ holds if and only if 
\[
1 = \mathrm{tr}\,( (3Q + Q^\perp) (tP)) = \mathrm{tr}\, ( (2Q + I)(tP)) = t (1 + 2 \mathrm{tr}\, (QP)).
\]
Because $(1-t) / (2t) \in [0,1]$, we can find a projection $Q \in \mathcal{P}$ such that  $\mathrm{tr}\, (QP) = (1-t) / (2t)$, and then by the above, 
 $(1/3)Q + Q^\perp  \sim tP$. It follows that  $(1/3)Q + Q^\perp  \sim sP$, which combined with Lemma \ref{simplll} yields that $t=s$.

Hence, for every real $t$, $1/3 \le t \le 1$, and every $P \in \mathcal{P}$, we have $\varphi (tP) = t P$. 
Using the coherency preserving map $\varphi_1 \colon X\mapsto I - \varphi (I -X)$, we obtain $\varphi (tP+P^\perp) = t P+P^\perp$ for every real $t$, $0\le t \le 2/3$, and every $P \in \mathcal{P}$. 
Since $\varphi(tP)\sim 0, \varphi(P), \varphi (tP+P^\perp)$, we get $\varphi(tP)=tP$ when $0\le t \le 2/3$ by Lemma \ref{simplll}. 
We conclude that
\[
\varphi (tP) = tP, \ \ \ P \in \mathcal{P}, \ \, t \in [0,1].
\]
Let now $A$ be any element of $[0,I]$ of rank two. Applying Lemma \ref{smtr} with $B = \varphi (A)$, we see that $\varphi (A) = A$ or $\varphi (A) = 0$. 

We have shown that $\varphi (A) \in \{0, A\}$ holds for every $A \in [0,I]$.
We now use the obtained result for the map $\varphi_1\colon X \mapsto I - \varphi ( I -X)$ to conclude that  $\varphi (A) \in \{A, I\}$ for every $A \in [0,I]$. Hence we obtain $\varphi (A) = A$ for every $A \in [0,I]$. 
\end{proof}

\begin{corollary}\label{pqyl}
Let $\varphi : [0,I] \to \oH$ be a coherency preserver with the property that $\varphi (0)$, $\varphi((1/2)I)$, and $\varphi (I)$ are in timelike position. Then $\varphi$ is standard.
\end{corollary}
\begin{proof}
This statement is a straightforward consequence of Proposition \ref{unitary} and Example \ref{IJ}.
\end{proof}

\begin{proof}[Proof of Theorem A]
Suppose that $\varphi\colon \mathcal{A}\to \oH$ satisfies the assumption of Theorem A.
By Corollary \ref{cororder}, there is a standard automorphism $\psi_1$ of $\oH$ satisfying 
\[
\psi_1([0,I])=[A, B],\ \ \psi_1(0)=A,\ \ \psi_1(I)=B, \ \ \text{and} \ \ \psi_1((1/2)I)=C.
\]
It follows that the restriction of the coherency preserver $\varphi\circ \psi_1$ to $[0,I]$ satisfies the assumption of Corollary \ref{pqyl}. 
Thus, Corollary \ref{pqyl} and Theorem \ref{identity} imply that $\varphi\circ \psi_1$ is standard, which further shows that $\varphi$ is standard, too.
\end{proof}

\subsection{Applications of Theorem A}\label{apply}
We begin with a useful lemma.
\begin{lemma}\label{path}
Let $\mathcal{U}$ be an open connected subset of $H_2$. 
For any pair $X, Y\in \mathcal{U}$, there are $X_1=X, X_2, \ldots, X_n=Y\in \mathcal{U}$ such that $X_1\sim X_2\sim \cdots\sim X_n$.
\end{lemma}
\begin{proof}
For $X, Y\in \mathcal{U}$, we write $X\approx Y$ if there are $X_1=X, X_2, \ldots, X_n=Y\in \mathcal{U}$ such that $X_1\sim X_2\sim \cdots\sim X_n$.
It is clear that $\approx$ is an equivalence relation.

First, let us consider the case where $\mathcal{U}$ is an open interval. 
When $X<Y$, then we have $\mathcal{S}_{X, Y}\subset \mathcal{U}$ by Corollary \ref{A<B}. We may take any point $Z\in \mathcal{S}_{X, Y}\subset \mathcal{U}$, and then $X\sim Z\sim Y$. Thus $X\approx Y$.  
If $X, Y$ are arbitrary, then we may take $Z\in \mathcal{U}$ such that $X< Z$, $Y< Z$, hence we obtain $X\approx Z\approx Y$. 
Thus the proof is complete when $\mathcal{U}$ is an open interval. 

Let $\mathcal{U}$ be an arbitrary open connected subset of $H_2$.
Fix $X_0\in \mathcal{U}$, and let $\mathcal{U}_0$ denote the collection of $Y\in \mathcal{U}$ with $X_0\approx Y$.
Then the preceding paragraph shows that $\mathcal{U}_0$ and $\mathcal{U}\setminus\mathcal{U}_0$ are open.
Since $\mathcal{U}$ is connected, we obtain $\mathcal{U}_0=\mathcal{U}$.
\end{proof}

Let us prove the fundamental theorem of chronogeometry \index{fundamental theorem of chronogeometry} as a consequence of the preceding subsection.
\begin{theorem}\label{ftcgelem}
Let $\varphi \colon H_2 \to H_2$ be a bijective map satisfying 
\[
A \sim B \iff \varphi (A) \sim \varphi (B)
\]
for every pair $A,B \in H_2$. Then $\varphi$ extends to an affine automorphism. 
\end{theorem}

\begin{proof}
After composing $\varphi$ with a translation, we can assume with no loss of generality that $\varphi (0) =0$. 
As before, we define $d(A,B) = \mathrm{rank}\, (A-B) \in \{ 0,1,2 \}$, $A,B \in H_2$. It is clear that we have $d(A,B) = d( \varphi (A) , \varphi (B))$, $A, B \in H_2$. 
The set $H_2^2$ of invertible matrices in $H_2$ is formed of $A\in H_2$ with the property $d(A,0) = 2$. Thus we get 
\[
\varphi(H_2^2)=H_2^2.
\] 

Recall that $H_2^2$ is the disjoint union of three sets $H_2^{++}, H_2^{--}$, and $H_2^{+-}$.
We first claim that 
\[
\varphi ( H_{2}^{++} \cup H_{2}^{--} ) = H_{2}^{++} \cup H_{2}^{--} .
\]
The above equality follows directly from the following characterization of the set $H_{2}^{++} \cup H_{2}^{--}$: An invertible matrix $A \in H_2$ belongs to $H_{2}^{++} \cup H_{2}^{--}$ if and only if for each line $\ell$ passing through zero there exists $B \in \ell\cap H_2$ such that $A \sim B$. 
Indeed, let $A \in H_{2}^{++}$ and $\ell = \{ aP \, : \, a \in {\oR}\}$ where $P \in \mathcal{P}$. Then $A^{-1} > 0$ and therefore $\mathrm{tr}\, (A^{-1}P) = \mathrm{tr}\, (P A^{-1}P) >0$. It follows that there exists exactly one (positive) $a \in \mathbb{R}$ such that $\mathrm{tr}\, (A^{-1}(aP)) = 1$, or equivalently, $A \sim aP$ (see Lemma \ref{tr1}).
In the same way we treat the case when $A \in H_{2}^{--}$. So, assume finally that $A \in H_{2}^{+-}$. Then $A^{-1}$ is again in $H_{2}^{+-}$. 
Therefore, there exists a unit vector $x \in \mathbb{C}^2$ such that $x^ \ast A^{-1}x = 0$, yielding that $\mathrm{tr}\, (A^{-1}Q) = 0$, where $Q = xx^\ast$. 
It follows from Lemma \ref{tr1} that there is no point in $\{ aQ \, : \, a \in \mathbb{R} \}$ that is coherent to $A$.

Observe that if $A,B \in H_{2}^{++} \cup H_{2}^{--}$ and $A \sim B$, then either $A,B \in H_{2}^{++}$, or $A,B \in H_{2}^{--}$. 
Moreover, $H_2^{++}$ and $H_2^{--}$ are open and connected sets. 
Thus Lemma \ref{path} implies that
\[
\varphi ( H_{2}^{++} ) = H_{2}^{++}  \ \ \ \text{and} \ \ \  \varphi ( H_{2}^{--} ) = H_{2}^{--} \, ;
\]
or 
\[
\varphi ( H_{2}^{++} ) = H_{2}^{--}  \ \ \ \text{and} \ \ \  \varphi ( H_{2}^{--} ) = H_{2}^{++} \, .
\]

Therefore, we get $\varphi(-I)<0=\varphi(0)<\varphi(I)$ or $\varphi(-I)>0=\varphi(0)>\varphi(I)$.
Hence Corollary \ref{matrixtriple} implies that the triple $\varphi(-I), \varphi(0), \varphi(I)$ is in timelike position. 
Theorem A shows that $\varphi$ extends to an automorphism $\hat{\varphi}$ of $\oH$. 
Note that $\hat{\varphi}(H_2)=H_2$ implies $\hat{\varphi}(\oi)=\oi$.
It remains to show that $\hat{\varphi}$ is an affine automorphism.
By composing $\hat{\varphi}$ with an appropriate affine automorphism of $\oH$, we may assume $\hat{\varphi}(0)=0$ and $\hat{\varphi}(I)=I$ without loss of generality. (Recall that the composition of affine automorphisms is again an affine automorphism by Lemma \ref{simsim}.) 

Then the automorphism $\psi \colon X\mapsto X^{-1}-I$ of $\oH$ satisfies 
\[
\psi^{-1}\circ \hat{\varphi}\circ\psi(0) = \psi^{-1}\circ \hat{\varphi}(\oi) = \psi^{-1}(\oi)=0,
\]
\[
\psi^{-1}\circ \hat{\varphi}\circ\psi(I)=\psi^{-1}\circ \hat{\varphi}(0) = \psi^{-1}(0) =I,\ \ \text{and}
\]
\[
\psi^{-1}\circ \hat{\varphi}\circ\psi((1/2)I)=\psi^{-1}\circ \hat{\varphi}(I) =\psi^{-1}(I) =(1/2)I.
\]
Now we apply Proposition \ref{unitary}.
We see that there is a $2\times 2$ unitary matrix $U$ such that $\psi^{-1}\circ \hat{\varphi}\circ\psi(X)=UXU^*$ holds for every $X\in [0, I]$, or $\psi^{-1}\circ \hat{\varphi}\circ\psi(X)=UX^tU^*$ holds for every $X\in [0, I]$. 
By the identity-type theorem (Theorem \ref{identity}), the same equality holds for every $X\in \oH$. 
It is easy to see that $\psi(UXU^*)=U\psi(X)U^*$ and $\psi(X^t)=\psi(X)^t$ for every $X\in \oH$. 
It follows that $\hat{\varphi}(X)=UXU^*$ holds for every $X\in \oH$, or $\hat{\varphi}(X)=UX^tU^*$ holds for every $X\in \oH$. 
\end{proof}

\begin{proof}[Proof of Theorem \ref{ftc}]
Assume that $\phi\colon \cM_4\to \cM_4$ is a bijection that preserves lightlikeness in both directions. 
We need to show that $\phi$ is of the form \eqref{fiesa}. 
We may assume $\phi(0)=0$ without loss of generality. 
Let $\xi\colon \mathcal{M}_4\to H_2$ be the mapping as in \eqref{zacet}.
Then the mapping $\varphi:=\xi\circ \phi\circ \xi^{-1}\colon H_2\to H_2$ clearly satisfies the assumption of Theorem \ref{ftcgelem}. 
It follows that $\varphi$ extends to an affine automorphism.
In particular, we see that $\varphi$ is linear.  
It follows that $\phi=\xi^{-1}\circ \varphi\circ \xi$ is also linear, that is, $\phi (r) = Sr$ for some real $4\times 4$ matrix $S$. Our task is to show that $S=cQ$ for some positive real number $c$ and some Lorentz matrix $Q$.
We know that for every $r \in \mathcal{M}_4$ we have
\[
\langle r,r \rangle = 0 \iff \langle Sr, Sr \rangle = 0,
\]
that is,
\[
(M r,r ) = 0 \iff (S^t M Sr, r ) = 0.
\]
Here, $(\cdot , \cdot )$ \index{$(\cdot \, , \cdot)$} denotes the standard inner product on the space of all real quadruples, and $M$ is as in \eqref{Matrix}.  Choosing suitable spacetime events $r$ (say, $r = (\cos \gamma, \sin \gamma, 0 , 1), \ldots$) and applying the fact that $S^t MS$ is symmetric, it is trivial to deduce that
\[
S^t MS = dM
\]
for some real number $d$. The matrix $S^t MS$ has the same inertia as $M$, that is, it has three negative eigenvalues and one positive eigenvalue. It follows that $d$ is positive, and consequently,
\[
\left( {1 \over \sqrt{d}} S \right)^t \, M \, \left( {1 \over \sqrt{d}} S \right) = M,
\]
as desired.
\end{proof}

\begin{corollary}\label{corp}
Every automorphism of $\oH$ is standard. 
An automorphism $\varphi$ of $\oH$ is an affine automorphism if and only if $\varphi ({\oi}) = {\oi}$.
\end{corollary}
\begin{proof}
Let $\varphi$ be an automorphism of $\oH$. 
Since the triple $0, I, (1/2)I$ is timelike (Example \ref{IJ}), we see that the triple $\varphi(0), \varphi(I), \varphi((1/2)I)$ is also timelike.
Thus Theorem A together with Theorem \ref{identity} implies that $\varphi$ is standard.
If in addition $\varphi(\oi)=\oi$, then we get $\varphi(H_2)=H_2$. 
Thus Theorems \ref{ftcgelem} and \ref{identity} show that $\varphi$ is an affine automorphism.
Conversely, it is clear that an affine automorphism $\varphi$ satisfies $\varphi ({\oi}) = {\oi}$. 
\end{proof}

Therefore, we no longer need to make the distinction between an automorphism and a standard automorphism.

\begin{corollary}
Let $A, B\in H_2$. 
An automorphism $\varphi$ of $\oH$ satisfies $\varphi(A)=B$ if and only if there is an affine automorphism $\psi$ satisfying 
\[
\varphi(X)=  \psi((X-A)^{-1})^{-1}+B
\]
for every $X\in \oH$.
\end{corollary}
\begin{proof}
If $\psi$ is an affine automorphism, then
\[
\psi((A-A)^{-1})^{-1}+B=\psi(\oi)^{-1}+B=\oi^{-1}+B=B.
\]
Conversely, assume that an automorphism $\varphi$ of $\oH$ satisfies $\varphi(A)=B$.
Define two automorphisms $\varphi_1, \varphi_2$ by $\varphi_1(X):=X^{-1} +A$, $\varphi_2(X):=(X-B)^{-1}$, $X\in \oH$. 
Then $\psi:=\varphi_2\circ\varphi\circ\varphi_1$ is an automorphism. 
By $\varphi(A)=B$, we get $\psi(\oi)=\oi$, or equivalently, $\psi$ is an affine automorphism. 
We get
\[
\varphi(X)= \varphi_2^{-1}\circ\psi\circ\varphi_1^{-1}(X) = \psi((X-A)^{-1})^{-1}+B
\]
for every $X\in \oH$, as desired.
\end{proof}
In particular, if $A=B=0$, then $\varphi$ is of the form $\varphi(X)= \psi(X^{-1})^{-1}$.
We may also see that a general automorphism can be expressed as the composition of only a few affine automorphisms and inversions. 
To show this, let $\varphi$ be an automorphism of $\oH$. 
From the fact that the cone with vertex $\oi$ is mapped by $\varphi$ onto some cone, it is easy to see that there is $A\in H_2$ such that $B:=\varphi(A)\in H_2$. 
Thus the preceding corollary shows that $\varphi$ is expressed as the composition of two inversions and three affine automorphisms (two of which are translations).


\section{Degenerate coherency preservers}\label{non}
\subsection{Two types of degenerate coherency preservers}
We first give some examples of coherency preservers that are not standard. 
This is to motivate the definition of degenerate coherency preservers.
Those who want to get to the main result immediately may look at Definitions \ref{degeneratefirst}, \ref{degeneratesecond}, \ref{degspl}, and Theorem B, then skip to Subsection \ref{applici}.

For $n\in \{0,1,2\}$, let $H_2^n$ \index{$H_2^n$} denote the set of matrices in $H_2$ of rank $n$.
Observe that $H_2^2$ is the set of all invertible ($=$ nonsingular) matrices in $H_2$.

\begin{example}
Obviously, any mapping from $\mathcal{A}\subset \oH$ into a line $\ell$ is a coherency preserver. 
As a special case, let $P$ be any projection of rank one, and define $\varphi\colon H_2\to \oH$ by
\[
\varphi (A) = (\mathrm{tr}\, A) P, \ \ \ A \in H_2.
\]
This is a coherency preserver, and it is of interest to observe that $\varphi$ is linear. 
\end{example}

\begin{example}\label{exC}
Another simple example is the map $\varphi \colon H_2 \to \oH$ defined by
\[
\varphi (A) = A, \ \ A \in H_2\setminus H_2^2
\]
and 
\begin{equation}\label{H22}
\varphi (A) = 0,\ \ A \in H_2^2.
\end{equation}
A slight modification gives a much ``wilder" example. We  choose any function $\eta \colon \mathcal{P} \to \mathcal{P}$ and for every rank one projection $P$
we choose a function $f_P \colon \mathbb{R} \to {\oR}$ satisfying $f_P (0) = 0$. Define $\varphi\colon H_2\to \oH$ by
\[
\varphi (tP) = f_P (t) \eta (P),\ \ (t,P)\in \mathbb{R}\times \cP
\]
and \eqref{H22}.
It is easily seen from Lemma \ref{pq} that any such $\varphi$ preserves coherency. 
If $\eta$ and all functions $f_P$ are chosen to be surjective, then the image of $\varphi$ is the whole cone $\mathcal{C}_0$.
\end{example}
Observe that the range of a coherency preserver in the above example is contained in one cone. 
We will give the general form of such coherency preservers in Proposition \ref{genC}.
After some more reflection, one may encounter an example of a coherency preserver whose range slightly protrudes one cone. 
Here is such an example: 

\begin{example}\label{exoh}
Let $f\colon \mathcal{P}\to \mathcal{P}$ be any mapping, and define $\varphi\colon \oH\to \oH$ by 
\[
\varphi(0)=0,\ \ \varphi(\oH\setminus\cC_0)=\{I\},
\]
and 
\[
\varphi(tP)=f(P), \ \ P\in \cP,\ \ t\in \oR\setminus\{0\}.
\] 
It is easy and left to the reader to show that $\varphi$ is a coherency preserver.
Observe that $\varphi(\oH)=f(\cP)\cup\{0, I\}$ is not necessarily contained in one cone.
\end{example}
Note that $\varphi(\oH\setminus \{0\})=f(\cP)\cup\{I\}$ is contained in the cone $\mathcal{C}_I$, in this example. 
Let us give another example, which is seemingly more complicated. 
\begin{example}\label{varphi1}
We denote 
by $\mathcal{L} \subset H_{2}^1$ the set of all rank one matrices that are of the form
\[
t \,  \left[ \begin{matrix} s &  i\sqrt{ s (1-s)}  \cr -i \sqrt{ s (1-s)}  & 1 - s  \end{matrix} \right]
\]
for some nonzero real $t$ and some real $s$, $0 \le s \le 1$. 
We set
\[
\mathcal{K} = \left\{ \left[ \begin{matrix} 0 &  p  \cr p & 0  \end{matrix} \right] \ :  \, p \in \mathbb{R} \setminus \{ 0 \} \right\} \subset H_{2}^2.
\]
Let $g \colon \mathbb{R} \setminus \{ 0 \} \to {\oR}$ and $\eta\colon [0,1] \to \mathcal{P}$ be any maps. 
We define $\varphi_1 \colon H_2 \to \oH$ by
\[
\varphi_1 (0) = 0,\ \ \varphi_1 (H_{2}^1 \setminus \mathcal{L}) = \{E_{11}\}, \ \ \varphi_1 (H_{2}^2 \setminus \mathcal{K}) = \{I\},
\]
\[
\varphi_1 \left( t \,  \left[ \begin{matrix} s &  i\sqrt{ s (1-s)}  \cr -i \sqrt{ s (1-s)}  & 1 - s  \end{matrix} \right]  \right) = \eta (s),\ \ (t,s)\in (\mathbb{R} \setminus \{ 0 \})\times [0,1],
\]
and 
\[
\varphi_1 \left( \left[ \begin{matrix} 0 &  p  \cr p & 0  \end{matrix} \right] \right) = g(p) E_{22} + E_{11}, \ \ p\in \mathbb{R} \setminus \{ 0 \}.
\]
Then $\varphi_1$ is a coherency preserver.
\end{example}

\begin{proof}
Assume that $A,B \in H_2$ are coherent. Observe that rank one matrices are mapped to rank one matrices. Thus, if one of $A$ or $B$ is the zero matrix, then the other one is of rank at most one, and therefore $\varphi_1 (A) \sim \varphi_1 (B)$. The next case we will treat is the one where both $A$ and $B$ are of rank one, that is, $A=tP$ and $B=sP$ for some rank one projection $P$ and some nonzero scalars $t,s$.
Then clearly, $\varphi_1 (A) = \varphi_1 (B)$. We continue by treating the case where both $A$ and $B$ are invertible. Using the fact that $\varphi_1$ maps the set of invertible matrices into the line
\[
\ell = \left\{ tE_{22} + E_{11} \, :  \, t \in {\oR} \right\},
\]
we see that $\varphi_1 (A) \sim \varphi_1 (B)$ in this case, as well.

It remains to consider the case when one of $A,B$, say $A$, is of rank one, and $B$ is of rank two. Then clearly, $\varphi_1 (A) \in \mathcal{P}$. 
If $B \not\in \mathcal{K}$, then $\varphi_1 (B) = I$ and $\varphi_1 (A) \sim \varphi_1 (B)$.  If $B \in \mathcal{K}$, then Lemma \ref{jelensr} below implies $A\not\in \mathcal{L}$, which yields 
$\varphi_1 (A) = E_{11}$. Therefore, we have $\varphi_1 (A) \sim \varphi_1 (B)$ in this last case, too. 
\end{proof}

\begin{lemma}\label{jelensr}
Let $p,t$ be nonzero real numbers and let $s$ be a real number, $0 \le s \le 1$. Then the matrices
\[
A=t \,  \left[ \begin{matrix} s &  i\sqrt{ s (1-s)}  \cr -i \sqrt{ s (1-s)}  & 1 - s  \end{matrix} \right]\in \mathcal{L} \ \ \ \text{and} \ \ \  B=\left[ \begin{matrix} 0 &  p  \cr p  & 0  \end{matrix} \right]\in \mathcal{K}
\]
are not coherent.
\end{lemma}
\begin{proof}
Observe that $A$ is of rank one while $B$ is invertible. We have
$B - A = B (I - B^{-1} A)$.
It is straightforward to verify that $B^{-1} A$ is a trace zero complex matrix and since it is of rank one, it is a square-zero matrix. Consequently, $I - B^{-1} A$ is invertible, which further yields that $B - A$ is of rank two, as desired.
\end{proof}

For $\varphi_1$ in Example \ref{varphi1}, we see that the image $\varphi_1(H_2\setminus \{0\})$ is contained in the cone $\mathcal{C}_I$. 
Thus, for each of the above examples of coherency preservers $\varphi\colon \mathcal{A}\to \oH$, there are $A\in \mathcal{A}$, $B\in \varphi(\mathcal{A})$ with the property $\varphi(\mathcal{A}\setminus \{A\})\subset \mathcal{C}_B$. 
We give a name to this class of coherency preservers.
\begin{definition}\label{degeneratefirst}
A coherency preserver $\varphi \colon \mathcal{A} \to \oH$ defined on a subset $\mathcal{A}$ of $\oH$ is of \emph{type $(\mathcal{C})$} \index{coherency preserver of type $(\mathcal{C})$} if there exist $A\in \mathcal{A}$ and $B \in \varphi(\mathcal{A})$  such that $\varphi (\mathcal{A}\setminus \{A\}) \subset \mathcal{C}_B$.
\end{definition}

Is every non-standard coherency preserver of type $(\mathcal{C})$? The answer is no.

\begin{example}\label{varphi2}
Let $f\colon \mathbb{R} \to {\oR}$ be any function.
Let us define a map $\varphi_2 \colon H_2 \to \oH$ in the following way. 
We start by defining 
\[
\varphi_2 (tE_{11}) = f(t) E_{11}
\]
for every $t \in \mathbb{R}$. 
If a matrix $A\in H_2\setminus\ell_{0, E_{11}}$ is coherent to  $tE_{11}$ for some real number $t$, then such a $t$ is uniquely determined.
In this case, we set
\[
\varphi_2 (A) = E_ {22} + f(t) E_{11}.
\]
In the remaining case where $d(A, tE_{11}) = 2$ for all $t \in \mathbb{R}$ (equivalently, $A$ is of rank two and the $(2,2)$-entry of $A$ is zero), we define
\[
\varphi_2 (A) = E_{22}.
\]
Then $\varphi_2$ is a coherency preserving map.
\end{example}

\begin{proof}
We need to show that $\varphi_2 (A) \not\sim \varphi_2 (B)$ implies $A \not\sim B$. 
Observe that $\varphi(H_2\cap\ell_{0, E_{11}})\subset \ell_{0,E_{11}}$ and $\varphi(H_2\setminus\ell_{0, E_{11}})\subset \ell_{E_{22},I}$. 
Therefore, $\varphi_2 (A) \not\sim \varphi_2 (B)$ implies that one of the two matrices $A, B$ is in $H_2\cap\ell_{0, E_{11}}$ and the other is in $H_2\setminus\ell_{0, E_{11}}$. 
In that case, it is obvious from the definition of $\varphi_2$ that $\varphi_2 (A) \not\sim \varphi_2 (B)$ implies $A \not\sim B$.
\end{proof}

Observe that $\varphi_2$ defined in this manner is not necessarily of type $(\mathcal{C})$. 
However, it is a coherency preserver of type $(\ell)$ in the following sense.
\begin{definition}\label{degeneratesecond}
A coherency preserver $\varphi \colon \mathcal{A} \to \oH$ defined on a subset $\mathcal{A}$ of $\oH$ is of \emph{type $(\ell)$} \index{coherency preserver of type $(\ell)$}
if there are lines $\ell, \ell'$ such that $\varphi(\mathcal{A}\setminus \ell)\subset \ell'$.
\end{definition}

Note that the two types $(\mathcal{C})$, $(\ell)$ of coherency preservers are not disjoint. 
For example, any mapping  $\varphi\colon \mathcal{A}\to \oH$ whose range is contained in a single line is a coherency preserver of both types $(\mathcal{C})$ and $(\ell)$. 
Note also that in both types the range of $\varphi$ is contained in a ``small" subset of $\oH$. 
Indeed, if $\varphi$ is of type $(\mathcal{C})$, then its range is contained in a union of a point and a cone; and if $\varphi$ is of type $(\ell)$, then its range is contained in a union of two lines.

It will turn out that the two types of coherency preservers play a dominant role in the class of non-standard coherency preservers, so let us give a name to the union of these two classes.

\begin{definition}\label{degspl}
A coherency preserver $\varphi \colon \mathcal{A} \to \oH$ defined on a subset $\mathcal{A}$ of $\oH$ is \emph{degenerate} \index{degenerate coherency preserver} if it is of type $(\mathcal{C})$ or of type $(\ell)$.
\end{definition}

Here we present the key result of the current section. 
\begin{theoremb}\index{Theorem B}
Let $A, B\in H_2$ satisfy $A<B$ and $\varphi\colon [A, B]\to \oH$ be a coherency preserver that is not standard. Then $\varphi$ is degenerate. 
\end{theoremb}
We will give its lengthy proof in Section \ref{ci}. 
Before that, we give applications of Theorem B in the next subsection.
It will turn out that a coherency preserver is either standard or degenerate if it is defined on a subset that looks like a matrix interval.

However, it is not true that a non-standard coherency preserver defined on, say, an open connected subset of $\oH$ is always degenerate in the sense of Definition \ref{degspl}. 
We give some examples.
Let $\mathcal{U}\subset H_2$ be the set formed of all matrices $X$ in $H_2$ such that the difference of the largest eigenvalue of $X$ and the smallest eigenvalue is less than $1$. Clearly, it is open. It is also connected. Indeed, let $A,B \in \mathcal{U}$ with eigenvalues $t_1 \le t_2$ and $s_1 \le s_2$, respectively. It is trivial to find paths contained in $\mathcal{U}$ from $A$ to $((t_1 + t_2)/2)I$,   from $((t_1 + t_2)/2)I$ to $((s_1 + s_2)/2)I$, and from $((s_1 + s_2)/2)I$ to $B$.

\begin{lemma}\label{neksej}
Let  $\mathcal{U}\subset H_2$ be as above.
Assume that $A, B\in \mathcal{U}$ are coherent. Then $|\mathrm{tr}\,(A-B)|<2$. 
\end{lemma}
\begin{proof}
We have $B = A + tP$ for some real number $t$ and some rank one projection $P$. Assume $t \ge 0$.
 Then $\mathrm{tr}\, A \le \mathrm{tr}\, B$. We need to verify that $t < 2$. Let $p \le q$ be eigenvalues of $A$, $x\in \mathbb{C}^2$ a unit vector spanning the image of $P$ and $y$ a unit vector orthogonal to $x$. We have
\[
(Bx,x) = (Ax,x) +t(Px,x) \ge p + t
\]
and
\[
(By,y) = (Ay,y) +t(Py,y) \le q 
\]
and therefore, the difference of the larger eigenvalue of $B$ and the smaller eigenvalue of $B$ is no smaller than $(p+t) - q = t - (q-p) > t -1$. It follows that $1 > t-1$, as desired. 
If $t\leq 0$, then exchange the roles of $A$ and $B$.
\end{proof}

We first give an easy example.
\begin{example}
Let $a > 3$ be a real number.
Let $A_n\in \oH$, $n\in \mathbb{Z}$, satisfy $A_n\sim A_{n+1}$ for every $n$.
We define $\psi\colon \mathcal{U}\to \oH$ by $\psi(X)=A_n$ for $X\in \mathcal{U}$ with $an\leq\mathrm{tr}\, X<a(n+1)$, $n\in \mathbb{Z}$. 
Since every pair $X, Y\in \mathcal{U}$ with $X\sim Y$ satisfies $|\mathrm{tr}\,(X-Y)|<2$, it is easy to see that $\psi$ is a coherency preserver.
\end{example}

Next we give a more involved example that mixes coherency preservers of the two types $(\ell)$ and $(\mathcal{C})$.

\begin{example}
Let $\varphi_1 \colon H_2 \to \oH$ be a coherency preserver of type $(\mathcal{C})$ given in Example \ref{varphi1}. 
Let $\varphi_2 \colon H_2 \to \oH$ be a coherency preserver of type $(\ell)$ given in Example \ref{varphi2}. 
Let $\varphi_3 \colon H_2\to \mathcal{C}_0$ be any coherency preserver of type $(\mathcal{C})$. 
Assume that $a \in \mathbb{R}$, $a > 3$.
We define $\varphi\colon \mathcal{U}\to \oH$ by
\[
\varphi(X) =\begin{cases}\varphi_2(X+aI)-E_{22},& \mathrm{tr}\, X<-a,\\ \varphi_1(X)-I, & -a\leq \mathrm{tr}\, X<a,\\\varphi_3(X), & a\leq \mathrm{tr}\, X .\end{cases}
\]
Then $\varphi$ is a coherency preserver. 
\end{example}
\begin{proof}
Recall that $\varphi_1(X)=I$ if $X>0$ or $X<0$. 
Thus we have $\varphi(X)=0$ for $0< X\in \mathcal{U}$ with $\mathrm{tr}\,X< a$, and for $0>X\in \mathcal{U}$ with $\mathrm{tr}\,X \geq -a$. 

Assume that $X, Y\in \mathcal{U}$ are coherent. 
We show that $\varphi(X)$ is coherent to $\varphi(Y)$.
Since $Y-X$ is of rank one, we see that $X\leq Y$ or $Y\leq X$ holds. 
We may assume $X\leq Y$ without loss of generality. 
By Lemma \ref{neksej}, we get $|\mathrm{tr}\,(X-Y)|<2$. 
Obviously, one of the following holds: 
\begin{itemize}
\item The set $\{\mathrm{tr}\,X, \mathrm{tr}\,Y\}$ is contained either in $(-\infty, -a)$,  $[-a, a)$, or  $[a,\infty)$. 
\item $\mathrm{tr}\,X<-a\leq \mathrm{tr}\,Y$.
\item $\mathrm{tr}\,X<a\leq \mathrm{tr}\,Y$.
\end{itemize}

In the first case, it is clear that $\varphi(X)\sim \varphi(Y)$. 

Assume $\mathrm{tr}\,X<-a\leq \mathrm{tr}\,Y$. Then we have  $\mathrm{tr}\,Y < \mathrm{tr}\,X + 2 <-a+2<-1$. By the definition of $\mathcal{U}$, we see that $Y<0$, and we obtain $\varphi(Y)=\varphi_1(Y)-I =0$. 
On the other hand, we know that $\mathrm{tr}\,X> \mathrm{tr}\,Y -2 \ge -a-2$ and thus $\mathrm{tr}\,(X+aI)>a-2>1$, which leads to $X+aI>0$. From the definition of $\varphi_2$, we have $\varphi(X)=\varphi_2(X+aI)-E_{22}\in \{sE_{11}\,:\, s\in {\oR}\}$. Therefore, $\varphi(Y)=0\sim \varphi(X)$. 

If $\mathrm{tr}\,X<a\leq \mathrm{tr}\,Y$, a similar argument shows that $X>0$ and thus $\varphi(X)=\varphi_1(X)-I=0$. Since $\varphi(Y)=\varphi_3(Y)\in\mathcal{C}_0$, we have $\varphi(X)\sim\varphi(Y)$, as desired. 
\end{proof}

Combining ideas in the above two examples together, one may construct various coherency preservers involving (possibly infinitely) many degenerate coherency preservers.
One may also construct similar examples of coherency preservers on, say, the open connected subset $\{X\in H_2\,:\, -1<\mathrm{tr}\,X<1\}\subset H_2$.

\subsection{Applications of Theorem B}\label{applici}
In this subsection, we give several results assuming that Theorem B holds true. 
The proof of Theorem B will be given later in Section \ref{ci}.
Recall that $\oH$ is endowed with the topology that comes from the identification with $U_2$ as in Subsection \ref{oHU2}. 

\begin{lemma}\label{ccc}
Let $A_n\in \oH$, $n\geq 1$, be a sequence that converges to $A$ in the compact space $\oH$. 
Then 
\[
\bigcap_{N\geq 1}\bigcup_{n\geq N} \mathcal{C}_{A_n}\subset \cC_A.
\]
\end{lemma}
\begin{proof}
We may consider $U_2$ instead of $\oH$. 
Thus, let us think of $A_n$ and $A$ as elements in $U_2$.
The condition
\[
B\in \bigcap_{N\geq 1}\bigcup_{n\geq N} \mathcal{C}_{A_n}
\] 
means that $B\sim A_n$ for infinitely many $n$. 
Recall that $B\sim A_n$ in $U_2$ means that the rank of $B-A_n$ is at most $1$. 
Since $A_n\to A$, it follows that the rank of $B-A$ is at most $1$, too. 
\end{proof}

\begin{lemma}\label{elln}
Let $\ell_n\subset \oH$, $n\geq 1$, be a sequence of lines. 
Then there is a subsequence $\ell_{n_k}$ and a line $\ell$ satisfying
\[
\bigcap_{N\geq 1}\bigcup_{k\geq N} \ell_{n_k}\subset \ell.
\]
\end{lemma}
\begin{proof}
Take any $A_n\in \ell_n$ for each $n$. 
By passing to a subsequence, we may assume that $A_n$ converges to some point $A\in \oH$ in the compact space $\oH$. 
We may also assume $A=0$ without loss of generality (use Lemma \ref{srtsrt}). 
Set $\mathcal{U}=\{X\in H_2\,:\, \| X\| <1\}$, which is an open neighborhood of $A=0$. 
It is easy to see that we may take $B_n\in \ell_n\setminus \mathcal{U}$ for each $n$. 
By passing to a subsequence, we may assume that $B_n$ converges to some element $B\in \oH\setminus \mathcal{U}$. 
Since $A_n\sim B_n$ for each $n$, we have $A\sim B$ (to see this, consider the corresponding points in $U_2$). 
If 
\[
X\in \bigcap_{N\geq 1}\bigcup_{n\geq N} \ell_{n},
\]
then we have $A_n\sim X\sim B_n$ for infinitely many $n$. 
Since $A_n\to A$ and $B_n\to B$, we have $X\in \ell_{A,B}$.
Thus we obtain the desired conclusion.
\end{proof}

\begin{lemma}\label{gap}
Let $\mathcal{A}\subset H_2$ be a subset and $\varphi\colon \mathcal{A}\to \oH$ a coherency preserver. 
Assume that the following condition holds: For any pair $X, Y\in \mathcal{A}$, there are $X_1=X, X_2, \ldots, X_n=Y\in \mathcal{A}$ such that $X_1\sim X_2\sim \cdots\sim X_n$.
(By lemma \ref{path}, this condition is satisfied whenever $\mathcal{A}$ is open and connected.)
If $A\in \oH$ satisfies $\varphi(\mathcal{A}) \subset \cC_A$, then either $A\in \varphi(\mathcal{A})$, or there is a line $\ell$ passing through $A$ with $\varphi(\mathcal{A}) \subset \ell$.
\end{lemma}
\begin{proof}
We may assume $A=0$ without loss of generality.
The desired conclusion readily follows from Lemma \ref{pq}.
\end{proof}

Let $A, B \in H_2$ with $A < B$. Any of the sets $(A,B)$, $\{C \in H_2 \, : \, C > A \}$, $\{ C \in H_2 \, :\, C < A \}$, and $H_2$ will be called an \emph{open interval} \index{open interval} in $H_2$. 

\begin{theorem}\label{degenerate}
Let $\mathcal{U}$ be an open interval in $H_2$ and $\varphi\colon \mathcal{U}\to \oH$ a coherency preserver.  Then $\varphi$ is either standard or degenerate. 
\end{theorem}
\begin{proof}
We may find $\{A_n\,:\, n\in \mathbb{Z}\}\subset H_2$ such that 
\[
\cdots<A_{-2}<A_{-1}<A_0<A_1<A_2<\cdots
\]
and $\mathcal{U}=\bigcup_{n\geq 1} [A_{-n}, A_n]$ hold. 
The restriction $\varphi_n\colon [A_{-n}, A_n]\to \oH$ of $\varphi$ to $[A_{-n}, A_n]$ is obviously a coherency preserver. 
Assume that $\varphi$ is not standard. 
Then Theorem \ref{identity} implies that for every $n\geq 1$ the map $\varphi_n$ is not standard.
It follows from Theorem B that $\varphi_n$ is degenerate. 

Assume that there are infinitely many $n\geq 1$ such that $\varphi_n$ is of type $(\mathcal{C})$. 
By passing to a subsequence, we may assume that $\varphi_n$ is of type $(\mathcal{C})$ for every $n\geq 1$. 
For each $n\geq 1$, we may take $B_n\in [A_{-n}, A_n]$ and $C_n\in \oH$ such that $\varphi([A_{-n}, A_n]\setminus \{B_n\})\subset \mathcal{C}_{C_n}$. 
By passing to a subsequence again, we may assume that $\{B_n\}$ and $\{C_n\}$ converge in the compact space $\oH$. 
Let $B_n\to B$ and $C_n \to C$ as $n\to \infty$ in $\oH$. 
It follows that 
\[
\mathcal{U}\setminus\{B\} \subset \bigcup_{n\geq N}[A_{-n}, A_n]\setminus \{B_n\}
\]
for every $N\geq 1$. Therefore, we have
\[
\varphi(\mathcal{U}\setminus \{B\}) \subset \bigcap_{N\geq 1}\bigcup_{n\geq N}\varphi([A_{-n}, A_n]\setminus \{B_n\}) \subset \bigcap_{N\geq 1}\bigcup_{n\geq N} \mathcal{C}_{C_n}.
\]
This together with Lemma \ref{ccc} implies that $\varphi(\mathcal{U}\setminus \{B\})\subset \mathcal{C}_C$. Therefore, $\varphi$ is of type $(\mathcal{C})$ if $C\in \varphi(\mathcal{U})$. 
Assume that $C\notin \varphi(\mathcal{U})$. Since $\mathcal{U}\setminus\{B\}$ is open and connected, Lemma \ref{gap} shows that $\varphi(\mathcal{U}\setminus \{B\})$ is contained in one line. Thus $\varphi$ is of type $(\mathcal{C})$.

Assume that there are infinitely many $n\geq 1$ such that $\varphi_n$ is of type $(\ell)$. 
This case can be considered in a parallel manner as in the preceding case.
By passing to a subsequence, we may assume that $\varphi_n$ is of type $(\ell)$ for every $n\geq 1$. 
For each $n\geq 1$, we may take lines $\ell_n$, $\ell_n'$ such that $\varphi([A_{-n}, A_n]\setminus \ell_n)\subset \ell_n'$. 
By using Lemma \ref{elln} and passing to a subsequence, we may assume that $\bigcap_{N\geq 1}\bigcup_{n\geq N} \ell'_{n}\subset \ell'$ for some line $\ell'$. 
Observe that the sequence $\bigcap_{n\geq N} \ell_n$,  $N\geq 1$, of subsets of $\oH$ is expanding and that each of $\bigcap_{n\geq N} \ell_n$ is either empty or a singleton or a line. Thus we see that $\bigcap_{n\geq N} \ell_n$ is contained in a line $\ell$ for all $N\geq 1$.
It follows that 
\[
\varphi(\mathcal{U}\setminus \ell)\subset \bigcap_{N\geq 1}\bigcup_{n\geq N}\varphi([A_{-n}, A_n]\setminus \ell_n) \subset \bigcap_{N\geq 1}\bigcup_{n\geq N} \ell_n'\subset \ell'.
\]
Thus $\varphi$ is of type $(\ell)$.
\end{proof}
Almost the same proof shows that a coherency preserver defined on an interval of the form $[A, B)$, or $\{ C \in H_2 \,: \, C \ge A\}$, or $(A, B]$, or  $\{ C \in H_2 \,: \, C \le A\}$ is also either standard of degenerate.
We also obtain the following. 

\begin{theorem}\label{degeneratex}
Let $\varphi\colon \oH \to \oH$ be a coherency preserver. Then $\varphi$ is either standard or degenerate. 
\end{theorem}
\begin{proof}
For $n\geq 1$, set $\mathcal{A}_n:=\{aP+bP^\perp\,:\, P\in \cP,\ a,b\in [-n, \infty]\}$. 
Clearly, we have $\bigcup_{n\geq 1} \mathcal{A}_n =\oH$.
If $\varphi$ is not standard, then $\varphi$ restricted to each $\mathcal{A}_n$ is not standard, either. 
Observe that the automorphism $X\mapsto (X+(n+1)I)^{-1}$ sends $\mathcal{A}_n$ onto $[0, I]$. 
Therefore, Theorem B implies that $\varphi$ restricted to each $\mathcal{A}_n$ is degenerate. 
Hence the same argument as above shows that $\varphi$ is degenerate on $\oH$.
\end{proof}

Let us call a subset of $\oH$ a generalized open interval \index{generalized open interval} if it is an image of some open interval in $H_2$ by an automorphism of $\oH$. 
Let $\mathcal{U}$ be an open subset of $\oH$. 
Observe that for every $X\in \mathcal{U}$ there is a generalized open interval $\mathcal{I}$ such that $X\in \mathcal{I}\subset \mathcal{U}$. 
We say that a coherency preserver $\varphi\colon \mathcal{U}\to \oH$ is \emph{locally degenerate} \index{locally degenerate} if $\varphi$ is degenerate on every generalized open interval contained in $\mathcal{U}$. 
Here is a consequence of Theorem \ref{degenerate} with the identity-type theorem (Theorem \ref{identity}).

\begin{theorem}\label{locally}
Let $\mathcal{U}$ be an open connected subset of $\oH$. 
Then every coherency preserver $\varphi\colon \mathcal{U}\to \oH$ is either standard or locally degenerate.
\end{theorem}

Now, we are going to give a proof of Theorem \ref{MM}, which was the original goal of our research.
We need an easy lemma. 

\begin{lemma}\label{lll}
Let $\ell\subset \oH$ be a line and $\xi\colon \mathcal{M}_4 \to H_2$ the mapping as in \eqref{zacet}.
Then $\xi^{-1}(H_2\cap \ell)$ is contained in one lightlike line.
\end{lemma}
\begin{proof}
There is nothing to prove when $\ell\cap H_2=\emptyset$. 
If $\ell\cap H_2\neq\emptyset$, then Lemma \ref{jojhcer} implies that there exist $A\in H_2$ and $P \in \mathcal{P}$ such that
$\ell \cap H_2= \{ A + aP \, : \, a \in \mathbb{R} \}$. 
This can be rewritten as $\ell \cap H_2 = \{ B + t(2P)\, : \, t \in \mathbb{R} \}$, where $B = A - (\mathrm{tr}\, A)P$ is a trace zero matrix. The spacetime event $r_0$ that corresponds to $B$ via the identification \eqref{zacet} is of the form $r_0 = (x_0, y_0, z_0 , 0)$. Applying the Bloch representation (Subsection \ref{projectionsx}), we see that 
$\ell \cap H_2$ corresponds to the lightlike line $\{(x_0+tx, y_0+ty, z_0+tz, t)\,:\, t\in \mathbb{R}\}$ for some $(x_0, y_0, z_0), (x,y,z) \in \mathbb{R}^3$ with $x^2+y^2+z^2=1$.
\end{proof}

\begin{proof}[Proof of Theorem \ref{MM}]
Instead of working on $\mathcal{M}_4$, we may equivalently consider a coherency preserver $\varphi\colon H_2\to H_2$. 
By Theorem \ref{degenerate}, $\varphi\colon H_2\to H_2\subset \oH$ is either standard or degenerate. 

If $\varphi$ is standard, then $\varphi$ extends to an automorphism of $\oH$, which necessarily maps $\oi$ to itself. By Corollary \ref{corp}, the extension is an affine automorphism, and the proof is complete in this case (see Subsection \ref{apply}).

Assume that $\varphi$ is of type $(\mathcal{C})$. 
Then there are $A\in H_2$ and $B\in H_2$ such that $\varphi(H_2\setminus\{A\})\subset \cC_B\cap H_2$. It immediately follows that the second item of Theorem \ref{MM} holds. 

Assume that $\varphi$ is of type $(\ell)$. 
Then there are lines $\ell, \ell'\subset \oH$ satisfying $\varphi(H_2\setminus\ell)\subset H_2\cap\ell'$. 
It follows from Lemma \ref{lll} that the third item of Theorem \ref{MM} holds.
\end{proof}

A generalization of Alexandrov's theorem given by Lester is also a simple consequence of our main result.

\begin{theorem}[{\cite{Les}}]\label{June}
Let $\mathcal{U}$ be an open connected subset of $\oH$. 
If $\varphi\colon \mathcal{U}\to \oH$ satisfies 
\[
X\sim Y\iff \varphi(X)\sim \varphi(Y)
\] 
for every pair $X, Y\in \mathcal{U}$, then $\varphi$ is standard.
\end{theorem}
\begin{proof}
Theorem \ref{locally} tells us that $\varphi$ is either standard or locally degenerate. 
We need to show that $\varphi$ is not locally degenerate. 
By Lemma \ref{Porder}, we just need to show that a degenerate coherency preserver $\varphi\colon (0, I)\to \oH$ never satisfies $X\sim Y\iff \varphi(X)\sim \varphi(Y)$ for every pair $X, Y\in (0, I)$. 

We prove it. 
Assume that $\varphi\colon (0, I)\to \oH$ is a degenerate coherency preserver satisfying $X\sim Y\iff \varphi(X)\sim \varphi(Y)$ for every pair $X, Y\in (0, I)$. We need to obtain a contradiction.

If $\varphi$ is of type $(\cC)$, then there are $A\in (0,I)$ and $B\in \varphi((0,I))$ such that $\varphi((0,I)\setminus\{A\})\subset \cC_B$. 
Take $A<C<D<I$, Then we have $\mathcal{S}_{C,D}\subset [C,D]\subset (0,I)$ by Corollary \ref{A<B}. 
Since $C\not\sim D$, we see that $\varphi(C)\not\sim \varphi(D)$. 
This together with $\varphi(C),\varphi(D)\in \cC_B$ shows that $d(B, \varphi(C))=d(B,\varphi(D))=1$ and $\ell_{B, \varphi(C)}\cap \ell_{B,\varphi(D)}=\{B\}$.
It follows that 
\[
\varphi(\mathcal{S}_{C,D})\subset \cC_{\varphi(C)}\cap \cC_{\varphi(D)}\cap \cC_{B} =\ell_{B, \varphi(C)}\cap \ell_{B,\varphi(D)}=\{B\}.
\]
However, $\mathcal{S}_{C,D}\,(\subset (0,I))$ has infinitely many points. 
For any $X\neq Y$ in $\mathcal{S}_{C,D}$, we have $X\not\sim Y$, but $\varphi(X)=B=\varphi(Y)$, thus we obtain a contradiction. 

If $\varphi$ is of type $(\ell)$, then there are lines $\ell,\ell'$  such that $\varphi((0,I)\setminus\ell)\subset \ell'$.
It is clear that there is a pair $A,B\in (0,I)\setminus\ell$ with $A\not\sim B$, but we have $\varphi(A),\varphi(B)\in \ell'$, hence $\varphi(A)\sim \varphi(B)$, which contradicts our assumption.
\end{proof}

\subsection{More concrete description of degenerate coherency preservers}\label{concrete}
We would like to describe degenerate coherency preservers in a more concrete manner.
We remark that none of the results in this subsection depend on Theorem B.

Let $\mathcal{A}\subset \oH$ be a subset.
First, we study the special case where the image is contained in a cone. 
Consider a coherency preserver $\varphi\colon \mathcal{A}\to \oH$ such that the image $\varphi(\mathcal{A})$ is contained in a cone. 
By considering $\psi\circ \varphi$ for an appropriate automorphism $\psi$ of $\oH$ instead of $\varphi$, the description of such a map reduces to the case $\varphi(\mathcal{A})\subset \mathcal{C}_0$ without loss of generality.
We may give a complete characterization. 
 
\begin{proposition}\label{genC}
Let $\varphi\colon \mathcal{A}\to \mathcal{C}_0$ be a mapping. For $P\in \cP$ we denote $\mathcal{A}_P :=\varphi^{-1}(\ell_{0, P}\setminus \{0\})$. Then the following two conditions are equivalent. 
\begin{itemize}
\item The mapping $\varphi$ is a coherency preserver. 
\item $A\not\sim B$ holds for any $P,Q\in \cP$ with $P\neq Q$ and any $A\in \mathcal{A}_P$, $B\in \mathcal{A}_Q$.
\end{itemize}
\end{proposition}
\begin{proof}
The mapping $\varphi$ preserves coherency if and only if $\varphi(A)\not\sim \varphi(B)$ implies $A\not\sim B$ for every pair $A, B\in \cA$. 
By Lemma \ref{pq}, we see that $\varphi(A)\not\sim \varphi(B)$ holds if and only if $A\in \mathcal{A}_P$ and $B\in \mathcal{A}_Q$ hold for some $P,Q\in \cP$ with $P\neq Q$.
This leads to the desired conclusion.
\end{proof}

This generalizes Example \ref{exC}.
We may get more examples in the following manner. 
Let $\mathcal{B}\subset\mathcal{A}$ be a subset satisfying $B_1\not\sim B_2$ for any distinct $B_1, B_2\in \mathcal{B}$, e.g.,\ $\mathcal{B}=\{tI\,:\, t\in \bR\}$, $\mathcal{B}=\{A\in H_2\,:\, \mathrm{tr}\,A =0\}$.
Take any mapping $\varphi\colon \mathcal{A}\to\cC_0$ satisfying $\varphi(\mathcal{A}\setminus\mathcal{B})=\{0\}$. 
Then $\varphi$ satisfies the condition of Proposition \ref{genC}.

In what follows, we consider the case where the image is not contained in any single cone. 
Let us begin with type $(\ell)$. 
Consider a coherency preserver $\varphi\colon \mathcal{A}\to \oH$ of type $(\ell)$. 
We may take lines $\ell_1, \ell_2'$ such that $\varphi(\mathcal{A}\setminus \ell_1)\subset \ell_2'$. 
We may also take a line $\ell_1'$ such that $\varphi(\mathcal{A}\cap \ell_1)\subset \ell_1'$. 
If $\varphi(\mathcal{A})$ is not contained in any single cone, then we have $\ell_1'\cap\ell_2'=\emptyset$. 
By considering $\psi_2\circ \varphi\circ\psi_1$ for an appropriate pair $\psi_1, \psi_2$ of automorphisms of $\oH$ instead of $\varphi$, the general description of such a map reduces to the case $\ell_1=\ell_1'=\ell_{0, E_{11}}$ and $\ell_2'=\ell_{E_{22}, I}$. (To see this, use Lemma \ref{herz} a few times.)

\begin{proposition}\label{ljzggr}
A mapping $\varphi\colon \mathcal{A}\to \oH$ with $\varphi(\mathcal{A}\cap \ell_{0, E_{11}})\subset \ell_{0, E_{11}}$ and $\varphi(\mathcal{A}\setminus  \ell_{0, E_{11}})\subset \ell_{E_{22}, I}$ is a coherency preserver if and only if the equation $\varphi(\mathcal{A}\cap \mathcal{C}_A\setminus  \ell_{0, E_{11}}) \subset\{\varphi(A)+E_{22}\}$ holds for every $A\in \mathcal{A}\cap \ell_{0, E_{11}}$. 
\end{proposition}
\begin{proof}
Let $\varphi\colon \mathcal{A}\to \oH$ be a mapping satisfying $\varphi(\mathcal{A}\cap \ell_{0, E_{11}})\subset \ell_{0, E_{11}}$ and $\varphi(\mathcal{A}\setminus  \ell_{0, E_{11}})\subset \ell_{E_{22}, I}$. 
If $\varphi$ preserves coherency, then we get $\varphi(\mathcal{A}\cap \mathcal{C}_A\setminus  \ell_{0, E_{11}}) \subset \ell_{E_{22}, I} \cap \cC_{\varphi(A)}=\{\varphi(A)+E_{22}\}$ for every $A\in \mathcal{A}\cap \ell_{0, E_{11}}$. 

Conversely, assume that $\varphi(\mathcal{A}\cap \mathcal{C}_A\setminus  \ell_{0, E_{11}}) \subset \{\varphi(A)+E_{22}\}$ holds for every $A\in \mathcal{A}\cap \ell_{0, E_{11}}$. 
If $A, B\in \mathcal{A}$ satisfy $\varphi(A)\not\sim \varphi(B)$, then one of $A, B$ lies in $\mathcal{A}\cap \ell_{0, E_{11}}$ and the other is in $\mathcal{A}\setminus  \ell_{0, E_{11}}$. 
Then it follows from the assumption that $A\not\sim B$, as desired.
\end{proof}
Observe that this generalizes Example \ref{varphi2}.

Lastly, we consider type $(\mathcal{C})$. 
Let $\varphi\colon \mathcal{A}\to \oH$ be a coherency preserver of type $(\mathcal{C})$. 
We may take $A\in \mathcal{A}$ and $B\in \oH$ such that $\varphi(\mathcal{A}\setminus \{A\})\subset \mathcal{C}_B$. 
If $\varphi(\mathcal{A})$ is not contained in any single cone, then we have $\varphi(A)\not\sim B$. 
By considering $\psi_2\circ \varphi\circ\psi_1$ for an appropriate pair $\psi_1, \psi_2$ of automorphisms of $\oH$ instead of $\varphi$, the general description of such a map reduces to the case $A=0$, $B=I$, and $\varphi(0)=0$.

\begin{proposition}\label{gena}
Assume that $0\in \mathcal{A}$. 
Let $\varphi\colon \mathcal{A}\to \oH$ be a mapping with $\varphi(0)=0$ and $\varphi(\mathcal{A}\setminus \{0\})\subset\mathcal{C}_I$. 
For $P\in \cP$ we denote $\mathcal{A}_P :=\varphi^{-1}(\ell_{I, P}\setminus \{I\})$.
Then $\varphi$ is a coherency preserver if and only if the following conditions hold. 
\begin{itemize}
\item $A\not\sim B$ holds for any $P, Q\in \mathcal{P}$ with $P \not=Q$ and any $A\in \mathcal{A}_P$, $B\in \mathcal{A}_Q$. 
\item For each $P\in \cP$ there is $Q\in \cP$ such that $\varphi(\mathcal{A}\cap \ell_{0, P}\setminus\{0\})\subset \{Q\}$.
\end{itemize}
\end{proposition}
\begin{proof}
If $\varphi$ is a coherency preserver and $P\in \cP$, then $\varphi(\mathcal{A}\cap \ell_{0, P}\setminus\{0\})$ is a coherent subset of $\cP$, so there is $Q\in \cP$ such that $\varphi(\mathcal{A}\cap \ell_{0, P}\setminus\{0\})\subset \{Q\}$. 
The rest of the proof is similar to the proof of Proposition \ref{genC}.
\end{proof}

In what follows, we observe that $\varphi$ needs to be of a very special form in the case  $\mathcal{A}=H_2$ or $\mathcal{A}=\oH$.  
Recall that $H_2^2= H_2\setminus \cC_0$ is the collection of all invertible matrices in $H_2$.

\begin{proposition}\label{jebjo}
A mapping $\varphi\colon H_2\to \oH$ with $\varphi(0)=0$ and $\varphi(H_2\setminus \{0\})\subset\mathcal{C}_I$ is a coherency preserver if and only if one of the following conditions is satisfied. 
\begin{enumerate}
\item $\varphi(H_2^2)= \{I\}$, and for each $P\in \mathcal{P}$ there is $Q_P \in \cP$ such that $\varphi(H_2\cap\ell_{0, P}\setminus\{0\})= \{Q_P \}$.
\item There is a unique element $R\in \cP$ such that $\mathcal{B}= H_2^2\cap \varphi^{-1}(\ell_{I, R}\setminus\{I\})$ is nonempty, and for each $P\in \mathcal{P}$ there is $Q_P\in \cP$ such that $\varphi(H_2\cap \ell_{0, P}\setminus\{0\})= \{Q_P\}$. Moreover, for every $P\in \cP$ with the property that some point of $H_2\cap \ell_{0, P}$ is coherent to some point of $\mathcal{B}$, we have $Q_P=R$. 
\end{enumerate}
\end{proposition}
\begin{proof}
It is left to the reader to show that a mapping satisfying (1) or (2) preserves coherency. 

Assume that $\varphi\colon H_2\to \oH$ with $\varphi(0)=0$ and $\varphi(H_2\setminus \{0\})\subset\mathcal{C}_I$ is a coherency preserver. 
For each $P\in \mathcal{P}$,  $\varphi(H_2\cap\ell_{0, P}\setminus\{0\})$ is a coherent subset of $\cP$, so there is $Q_P\in \cP$ such that $\varphi(H_2\cap\ell_{0, P}\setminus\{0\})= \{Q_P\}$.

Assume in addition (1) does not hold. 
Then we may find $X\in H_2^2$ satisfying $\varphi(X)=R + cR^\perp$ with $R\in \cP$ and $c\in {\oR}\setminus\{1\}$. 
If $X>0$ or $X<0$, then for every $P\in \cP$ there is $d\in\mathbb{R}\setminus\{0\}$ such that $X\sim dP$. 
(Indeed, we know that $\mathcal{S}_{0,X}\subset H_2$ by Lemma \ref{A<B}, and that every line passing through $0$ intersects $\mathcal{S}_{0,X}$.)
This implies $\varphi(dP)=R$, hence we obtain $Q_P=R$ for every $P\in \cP$. 
Assume that $X\not>0$ and $X\not<0$. 
Then we may find $E\in \cP$ and $c_1, c_2>0$ such that $X=c_1E-c_2E^\perp$. 
By the lemma below, we have $Q_P=R$ for every $P\in \cP$ with $\mathrm{tr}\, (PE)\neq c_1/(c_1+c_2)$. 
In both cases, we have $Q_P=R$ for all points $P$ in the Bloch sphere possibly with the exception of points in one circle.
Thus we see that (2) holds in this case.
\end{proof}

\begin{lemma}
Let $E\in \cP$ and $c_1, c_2>0$. 
For each $P\in \cP$ with $\mathrm{tr}\, (PE)\neq c_1/(c_1+c_2)$, there is $d\in\mathbb{R}\setminus\{0\}$ such that $c_1E-c_2E^\perp\sim dP$. 
\end{lemma}
\begin{proof}
There is no loss of generality in assuming that $E=E_{11}$. 
We may write 
\[
P=\left[ \begin{matrix} c & e^{it}\sqrt{c-c^2}  \cr e^{-it}\sqrt{c-c^2} & 1-c \cr \end{matrix} \right]
\]
for some $c\in [0,1]\setminus\{c_1/(c_1+c_2)\}$ and some $t\in [0,2\pi)$. 
Then we have 
\[
c_1E-c_2E^\perp- dP = \left[ \begin{matrix} c_1-dc & -de^{it}\sqrt{c-c^2}  \cr -de^{-it}\sqrt{c-c^2} & -c_2-d(1-c) \cr \end{matrix} \right]
\]
and thus 
\[
\det(c_1E-c_2E^\perp- dP) = (c_1-dc)(-c_2-d(1-c)) - d^2(c-c^2) = -c_1c_2+d((c_1+c_2)c-c_1),
\]
which takes the value $0$ when $d=c_1c_2/((c_1+c_2)c-c_1)\neq 0$.
\end{proof}

\begin{proposition}\label{oHoH}
A mapping $\varphi\colon \oH\to \oH$ with $\varphi(0)=0$ and $\varphi(\oH\setminus \{0\})\subset\mathcal{C}_I$ is a coherency preserver if and only if one of the following conditions holds. 
\begin{itemize}
\item There is $Q\in \cP$ satisfying $\varphi(\mathcal{C}_0\setminus\{0\})= \{Q\}$ and $\varphi(\oH\setminus \mathcal{C}_0)\subset \ell_{Q,I}$.
(In this case, we have $\varphi(\oH)\subset \cC_Q$.)
\item $\varphi(\oH\setminus\mathcal{C}_0)=\{I\}$, and for each $P\in \mathcal{P}$ there is $Q_P \in \cP$ such that $\varphi(\ell_{0, P}\setminus\{0\})= \{Q_P\}$.
\end{itemize}
\end{proposition}
\begin{proof}
If $\varphi\colon \oH\to \oH$ with $\varphi(0)=0$ and $\varphi(\oH\setminus \{0\})\subset\mathcal{C}_I$ is a coherency preserver, then $\varphi(\ell_{0, P}\setminus\{0\})$ is a coherent subset of $\cP$ for each $P\in \mathcal{P}$. Thus there is $Q_P\in \cP$ such that $\varphi(\ell_{0, P}\setminus\{0\})= \{Q_P\}$.

Assume that $Q_P=Q\in \cP$ for every $P\in \cP$. 
Since every element of $\oH\setminus \mathcal{C}_0$ is coherent to some point of $\cC_0\setminus\{0\}$, we see that $\varphi(\oH\setminus \mathcal{C}_0)\subset \mathcal{C}_Q$. This together with the assumption implies $\varphi(\oH\setminus \mathcal{C}_0)\subset \mathcal{C}_Q\cap \mathcal{C}_I =\ell_{Q,I}$.

Assume that $Q_{P_1}\neq Q_{P_2}$ for some pair $P_1, P_2\in \cP$. 
Observe that every element of $\oH\setminus \cC_0$ is coherent to some element of $\ell_{0,P_1}\setminus\{0\}$ and some element of $\ell_{0,P_2}\setminus\{0\}$. 
This together with 
\[
\varphi(\ell_{0, P_1}\setminus\{0\})= \{Q_{P_1}\}\neq \{Q_{P_2}\}=\varphi(\ell_{0, P_1}\setminus\{0\})
\]
and $\varphi(\oH\setminus \{0\})\subset\mathcal{C}_I$ shows that $\varphi(\oH\setminus\mathcal{C}_0)=\{I\}$. 

Thus we obtain the direct implication. The converse is easy to show.
\end{proof}

Observe that Proposition \ref{jebjo} generalizes Example \ref{varphi1}
and Proposition \ref{oHoH} generalizes Example \ref{exoh}.


\section{Coherency preservers on a closed interval}\label{ci}
The purpose of the current section is to prove Theorem B. 
Before we begin, we give some easy lemmas about intervals. 
\begin{lemma}\label{dupc}
Let $A, B\in H_2$ satisfy $A<B$, and $P\in \mathcal{S}_{A,B}$. 
Then $[A, B]\cap \cC_P =[A, P]\cup [P, B]$.
\end{lemma}
\begin{proof}
By Lemma \ref{Porder}, we may assume $A=0$ and $B=I$ without loss of generality. 
Then $P\in \mathcal{S}_{0,I} =\mathcal{P}$. 
It suffices to think of the case $P=E_{11}$. 

We need to show that $[0, I]\cap \cC_{E_{11}} =[0, E_{11}]\cup [E_{11}, I]$.
Note that $[0, E_{11}] =\{tE_{11}\,:\, t\in [0,1]\}$ and that $[E_{11}, I] =\{E_{11}+tE_{22}\,:\, t\in [0,1]\}$. 
It is clear that $[0, I]\cap \cC_{E_{11}} \supset[0, E_{11}]\cup [E_{11}, I]$.
Let $X\in [0, I]\cap \cC_{E_{11}}$. 
We show that $X\in [0, E_{11}]\cup [E_{11}, I]$.
If $X=E_{11}$, then there is nothing to prove.
So, assume that $X = E_{11} + tQ$ for some rank one projection $Q$ and some nonzero real number $t$. It suffices to show that $Q\in \{E_{11}, E_{22}\}$. If $Q \not = E_{11}$, then $t >0$. Indeed, if $t < 0$, then $Q \not = E_{11}$ implies that $(Qe_2 , e_2) > 0$ and consequently, $(Xe_2 , e_2) <0$, a contradiction. We have $1 \ge (Xe_1 , e_1) = 1 + t(Qe_1, e_1)$, and consequently $Q=E_{22}$. This completes the proof.
\end{proof}

\begin{lemma}\label{0XI}
Let $A, B\in H_2$ satisfy $A<B$, and $P\in \mathcal{S}_{A,B}$. 
Then 
\begin{equation}\label{mmnn}
(A, B]\setminus \ell_{P, B}\subset \bigcup_{X\in [A, P]\setminus\{A, P\}} \cC_X
\end{equation}
 and 
\[
[A, B)\setminus \ell_{A, P}\subset \bigcup_{X\in [P, B]\setminus\{P, B\}} \cC_X
\]
hold.
\end{lemma}
\begin{proof}
With no loss of generality, we may again assume that $A= 0$, $B=I$, and $P=E_{11}$. 
To verify the first inclusion, let $X\in (0, I]\setminus \ell_{E_{11}, I}$.
We need to find $t\in (0,1)$ satisfying $X\sim tE_{11}$.
Since $X>0$, we have $\det X >0$.
On the other hand, we have $((X-E_{11})e_1, e_1)= (Xe_1,e_1)-1<0$ and $((X-E_{11})e_2, e_2)=(Xe_2, e_2)>0$, thus $X-E_{11}\in H_2^{+-}$. It follows that $\det (X-E_{11}) <0$. 
By the intermediate value theorem, there is $t\in (0,1)$ satisfying $\det (X- tE_{11})=0$, or equivalently, $X\sim tE_{11}$.
The proof of the second inclusion is similar.
\end{proof}
The above two lemmas clearly imply:
\begin{corollary}\label{JeDr}
Let $A, B\in H_2$ satisfy $A<B$, and let $P\in \mathcal{S}_{A, B}$. 
Then every point of $[A, B]$ is coherent to some point of $[A, P]$ and some point of $[P, B]$. 
\end{corollary}

\begin{remark}
When applying the above results, we will often use some straightforward consequences. 
For example, it is easy to see that (\ref{mmnn}) yields the following. 
Every point of $(A, B]$ is coherent to some point of $[A, P]\setminus\{ A\}$, and every point of $[A, B]\setminus \ell_{P, B}$ is coherent to some point of $[A, P]\setminus\{P \}$.
\end{remark}

Now, let us begin the proof of Theorem B.
Therefore, suppose that $A, B\in H_2$ satisfy $A<B$, and that $\varphi\colon [A, B]\to \oH$ is a coherency preserver that is not standard.
We need to show that $\varphi$ is either of type $(\mathcal{C})$ or $(\ell)$.
We will distinguish three possibilities: $d(\varphi (A), \varphi (B)) = 0$, or $1$, or $2$. 
We will first consider the case $d(\varphi (A), \varphi (B)) =2$ in Subsection \ref{dab2} and then apply the result there to attack the case $d(\varphi (A), \varphi (B)) \in \{0, 1\}$ in Subsections \ref{dab0}, \ref{dab1}.

\subsection{The case $d(\varphi (A), \varphi (B)) =2$}\label{dab2}
We first deal with the case $d(\varphi (A), \varphi (B)) =2$. 
We would like to show that $\varphi$ is degenerate. 
In fact, we may give a more precise statement. 
To explain it, we make the following definition. 
\begin{definition}
Let $A, B\in H_2$ satisfy $A<B$. 
Let $\varphi\colon [A, B]\to \oH$ be a coherency preserving map with $d(\varphi(A), \varphi(B))=2$. 
\begin{itemize}
\item We say that $\varphi$ satisfies \emph{Condition (o)} \index{Condition (o)} if the set $\varphi(\mathcal{S}_{A, B})$ is a singleton $\{Q\}$ and $\varphi([A, B])\subset \mathcal{C}_Q$. 
\item We say that $\varphi$ satisfies \emph{Condition (i)} \index{Condition (i)} if $\varphi([A, B))=\{\varphi(A)\}$ and there is no $P\in \mathcal{S}_{A, B}$ such that the image $\varphi(\mathcal{S}_{A, B}\setminus \{P\})$ is a singleton. 
\item We say that $\varphi$ satisfies \emph{Condition (ii)} \index{Condition (ii)} if $\varphi((A, B])=\{\varphi(B)\}$ and there is no $P\in \mathcal{S}_{A, B}$ such that the image $\varphi(\mathcal{S}_{A, B}\setminus \{P\})$ is a singleton. 
\item We say that $\varphi$ satisfies \emph{Condition (iii)} \index{Condition (iii)} if there are $P_\times\in \mathcal{S}_{A, B}$ and distinct points $Q_\circ, Q_\times\in \mathcal{S}_{\varphi(A), \varphi(B)}$ such that $\varphi([A, B]\cap \ell_{A, P_\times})\subset \ell_{\varphi(A), Q_\times}$ and $\varphi([A, B]\setminus \ell_{A, P_\times})\subset \ell_{Q_\circ, \varphi(B)}$. 
\item We say that $\varphi$ satisfies \emph{Condition (iv)} \index{Condition (iv)} if there are $P_\times\in \mathcal{S}_{A, B}$ and distinct points $Q_\circ, Q_\times\in \mathcal{S}_{\varphi(A), \varphi(B)}$ such that $\varphi([A, B]\cap \ell_{P_\times, B})\subset \ell_{Q_\times, \varphi(B)}$ and $\varphi([A, B]\setminus \ell_{P_\times, B})\subset \ell_{\varphi(A), Q_\circ}$. 
\end{itemize} 
\end{definition}

The main purpose of the current subsection is to show the following proposition.
\begin{proposition}\label{newprop}
Let $A, B\in H_2$ satisfy $A<B$. 
Let $\varphi\colon [A, B]\to \oH$ be a coherency preserving map with $d(\varphi(A), \varphi(B))=2$. 
Assume that $\varphi$ is not standard.
Then, $\varphi$ satisfies (at least) one of Conditions (o)--(iv). 
\end{proposition}

Let $\varphi$ be a mapping satisfying the assumption of Proposition \ref{newprop}.
By considering $\psi_2\circ\varphi\circ \psi_1$ instead of $\varphi$ for a pair of suitable automorphisms $\psi_1, \psi_2$ of $\oH$, we may assume that $A=0$, $B=I$, and $\varphi(0)=0$, $\varphi(I)=I$.
In this case, we have $\varphi(\cP)\subset \cP$ by Lemma \ref{veryeasy}.
Therefore, Proposition \ref{newprop} can be obtained from the following.

\begin{proposition}\label{proposition08}
Let $\varphi\colon [0, I]\to \oH$ be a coherency preserving map with $\varphi(0)=0$ and $\varphi(I)=I$. 
Assume that $\varphi$ is not standard. Then (at least) one of the following holds.
\begin{itemize}
\item The mapping $\varphi$ satisfies Condition (o), or equivalently, the set $\varphi(\mathcal{P})$ is a singleton $\{Q\}$, and $\varphi([0, I])\subset \mathcal{C}_Q$. 
\item The mapping $\varphi$ satisfies Condition (i), or equivalently, we have $\varphi([0, I))=\{0\}$, and there is no $P\in \mathcal{P}$ such that the image $\varphi(\mathcal{P}\setminus \{P\})$ is a singleton. In this case, the equality $\varphi(P+cP^\perp)=\varphi(P)$ holds for each $P\in \mathcal{P}$ and each $c\in [0, 1)$.
\item The mapping $\varphi$ satisfies Condition (ii), or equivalently, we have $\varphi((0, I])=\{I\}$, and there is no $P\in \mathcal{P}$ such that the image $\varphi(\mathcal{P}\setminus \{P\})$ is a singleton. In this case, the equality $\varphi(cP)=\varphi(P)$ holds for each $P\in \mathcal{P}$ and each $c\in (0, 1]$.
\item The mapping $\varphi$ satisfies Condition (iii), or equivalently, there are $P_\times, Q_\circ, Q_\times\in \mathcal{P}$ satisfying $Q_\circ\neq Q_\times$, $\varphi([0,I] \cap \ell_{0, P_\times})\subset \ell_{0,Q_\times}$, and $\varphi([0,I] \setminus \ell_{0, P_\times})\subset \ell_{Q_\circ,I}$.
In this case, for every $c\in [0,1]$, the set $\varphi(\mathcal{C}_{cP_\times}\cap [0, I]\setminus \ell_{0, P_\times})$ is a singleton that consists of the unique point in $\ell_{Q_\circ, I}$ that is coherent to $\varphi(cP_\times)$.
\item The mapping $\varphi$ satisfies Condition (iv), or equivalently, there are $P_\times, Q_\circ, Q_\times\in \mathcal{P}$ satisfying $Q_\circ\neq Q_\times$, $\varphi([0,I] \cap \ell_{P_\times,I})\subset \ell_{Q_\times, I}$, and $\varphi([0,I] \setminus \ell_{P_\times, I})\subset \ell_{0,Q_\circ}$.
In this case, for every $c\in [0,1]$, the set $\varphi(\mathcal{C}_{P_\times+cP_\times^\perp}\cap [0, I]\setminus \ell_{ P_\times, I})$ is a singleton that consists of the unique point in $\ell_{0,Q_\circ}$ that is coherent to $\varphi(P_\times+cP_\times^\perp)$.
\end{itemize} 
\end{proposition}

\begin{remark}\label{iiiiv}
When Condition (iii) or (iv) holds, then one may easily see that $\varphi(\mathcal{P}\setminus \{P_\times\})=\{Q_\circ\}$ and $\varphi(P_\times)=Q_\times\neq Q_\circ$ hold.
Therefore, the point $P_\times$ (which is an ``exceptional'' point) is uniquely determined in these cases. 
\end{remark}

Using Proposition \ref{proposition08}, we may see that a mapping satisfying one of Conditions (i)--(iv) is automatically a degenerate coherency preserver. 
More precisely, the following holds.
\begin{lemma}\label{conv}
Let $\varphi\colon [0, I]\to \oH$ be a map with $\varphi(0)=0$ and $\varphi(I)=I$, $\varphi(\cP)\subset \cP$. 
If one of the following two conditions holds, then $\varphi$ is a degenerate coherency preserver of type $(\mathcal{C})$.
\begin{enumerate}
\item[(i)'] $\varphi([0, I))=\{0\}$, and $\varphi(P+cP^\perp)=\varphi(P)$ for every $P\in \mathcal{P}$ and $c\in [0, 1)$.
\item[(ii)'] $\varphi((0, I])=\{I\}$, and $\varphi(cP)=\varphi(P)$ for every $P\in \mathcal{P}$ and $c\in (0, 1]$.
\end{enumerate}
If one of the following two conditions holds, then $\varphi$ is a degenerate coherency preserver of type $(\ell)$.
\begin{enumerate}
\item[(iii)'] There are $P_\times, Q_\circ, Q_\times\in \mathcal{P}$ satisfying $Q_\circ\neq Q_\times$, $\varphi([0,I] \cap \ell_{0, P_\times})\subset \ell_{0,Q_\times}$, and $\varphi([0,I] \setminus \ell_{0, P_\times})\subset \ell_{Q_\circ,I}$. Moreover, for every $c\in [0,1]$, the set $\varphi(\mathcal{C}_{cP_\times}\cap [0, I]\setminus \ell_{0, P_\times})$ is a singleton that consists of the unique point in $\ell_{Q_\circ, I}$ that is coherent to $\varphi(cP_\times)$.
\item[(iv)'] There are $P_\times, Q_\circ, Q_\times\in \mathcal{P}$ satisfying $Q_\circ\neq Q_\times$, $\varphi([0,I] \cap \ell_{P_\times,I})\subset \ell_{Q_\times, I}$, and $\varphi([0,I] \setminus \ell_{P_\times, I})\subset \ell_{0,Q_\circ}$. Moreover, for every $c\in [0, 1]$, the set $\varphi(\mathcal{C}_{P_\times+cP_\times^\perp}\cap [0, I]\setminus \ell_{P_\times, I})$ is a singleton that consists of the unique point in $\ell_{0,Q_\circ}$ that is coherent to $\varphi(P_\times+cP_\times^\perp)$.
\end{enumerate} 
\end{lemma}
\begin{proof}
If (ii)' holds, then $\varphi$ is the restriction of a mapping described in (1) of Proposition \ref{jebjo}. 
Thus $\varphi$ is a degenerate coherency preserver of type $(\mathcal{C})$. 
If (i)' holds, then one may check that the mapping $X\mapsto I-\varphi(I-X)$ is of the form (ii)', so we again see that $\varphi$ is a degenerate coherency preserver of type $(\mathcal{C})$.

Assume that (iii)' holds 
and that $A, B\in [0,I]$ satisfy $\varphi(A)\not\sim\varphi(B)$. 
We see that one of $A, B$ lies in $[0,I]\cap \ell_{0,P_\times}$ and the other is in $[0,I]\setminus \ell_{0,P_\times}$, so (iii)' clearly implies that $A\not\sim B$. 
Thus $\varphi$ is a coherency preserver. 
It is apparent that $\varphi$ is of type $(\ell)$. 
Similarly, we see that $\varphi$ is a degenerate coherency preserver of type $(\ell)$ if (iv)' holds.
\end{proof}

Proposition \ref{proposition08} together with Lemma \ref{conv} implies the following.
\begin{corollary}\label{d=2}
Let $A, B\in H_2$ satisfy $A<B$. 
Let $\varphi\colon [A, B]\to \oH$ be a coherency preserving map with $d(\varphi(A), \varphi(B))=2$. 
Assume that $\varphi$ is not standard.
Then $\varphi$ is degenerate.
\end{corollary}

\begin{remark}\label{added}
In Proposition \ref{proposition08}, if $j,k\in \{o,i,ii,iii,iv\}$ with $j\neq k$ and $\{j,k\}\neq \{iii,iv\}$, then it is clear that Conditions (j) and (k) cannot be fulfilled simultaneously. 
However, it can happen that (iii) and (iv) are fulfilled simultaneously, in which case $\varphi([0,I])$ consists of exactly four points, and $\varphi$ is of both types $(\mathcal{C})$ and $(\ell)$. To see this, we may assume with no loss of generality that $P_\times = Q_\times = E_{11}$. (Note that the points $P_\times$, $Q_{\circ}$, $Q_\times$ in (iii) are the same as those in (iv), see Remark \ref{iiiiv}.) If $\varphi$ satisfies both (iii) and (iv), then all elements of the set $\mathcal{A} = [0,I] \setminus (\ell_{0, E_{11}} \cup \ell _{E_{11}, I})$ are mapped into the intersection of lines $\ell_{0, Q_{\circ}}$ and $\ell_{Q_{\circ}, I}$. Thus, we have
\begin{equation}\label{iiiiv1}
\varphi (0) = 0, \ \ \ \varphi (I) = I, \ \ \ \varphi (E_{11}) = E_{11}, \ \ \text{and}\ \ 
\varphi (X) = Q_{\circ},\ \ X \in \mathcal{A}.
\end{equation}
It remains to consider the $\varphi$-images of elements of the set 
\[
\{ tE_{11}\, : \, 0 < t < 1 \} \cup \{ E_{11} + tE_{22}\, : \, 0 < t < 1 \}.
\]
For every $t \in (0,1)$, the matrix $tE_{11}$ is coherent to $0$, $E_{11}$, and some $X \in \mathcal{A}$. It follows that 
\begin{equation}\label{iiiiv2}
\varphi (t E_{11}) = 0,\ \  0 < t < 1.
\end{equation}
Similarly, for every $t \in (0,1)$, the matrix $E_{11} + tE_{22}$ is coherent to $E_{11}$, $I$, and some $X \in \mathcal{A}$. It follows that 
\begin{equation}\label{iiiiv3}
\varphi (E_{11} + t E_{22}) = I,\ \  0 < t < 1.
\end{equation}
Hence, the image of $\varphi$ consists of points $0,E_{11}, Q_{\circ}, I$. 
Because $\varphi ( [0,I] \setminus \{ E_{11} \} ) = \{ 0,I, Q_{\circ} \} \subset \mathcal{C}_{Q_{\circ}}$, the map $\varphi$ is also of type $(\mathcal{C})$. 
Conversely, it is easy to see that $\varphi\colon [0,I]\to \oH$ defined by \eqref{iiiiv1}, \eqref{iiiiv2}, and \eqref{iiiiv3} satisfies (iii) and (iv) simultaneously.
\end{remark}

Since the proof of Proposition \ref{proposition08} is rather long, we separate it into claims. 
In what follows, we assume that $\varphi\colon [0, I]\to \oH$ is a coherency preserving map with $\varphi(0)=0$, $\varphi(I)=I$. We also assume that $\varphi$ is not standard. 

\begin{claim}
We have
\begin{equation}\label{hashi}
\varphi([0, I]\setminus (0, I))\subset \bigcup_{P\in \cP}(\ell_{0, \varphi(P)}\cup\ell_{\varphi(P),I}).
\end{equation}
\end{claim}
\begin{proof}
Observe that every $A\in [0, I]\setminus (0, I)$ is coherent to some $P\in \cP$ and one of $0, I$, so that we have either $0=\varphi(0)\sim \varphi(A)\sim\varphi(P)$ or $I=\varphi(I)\sim \varphi(A)\sim\varphi(P)$.
Since $\varphi(\cP)\subset \cP$, we obtain \eqref{hashi}.
\end{proof}

Let $A\in (0, I)$. We show 
\[
\varphi(A)\in \bigcup_{P\in \cP}(\ell_{0, \varphi(P)}\cup\ell_{\varphi(P), I}) \cup \bigcap_{P\in \cP} \cC_{\varphi(P)}.
\] 
By Lemma \ref{order}, there is an automorphism $\psi$ of $\oH$ such that 
\[
\psi([0, I])=[0, I],\ \ \psi(0)=0,\ \ \psi(I)=I,\ \  \text{and}\ \ \psi((1/2)I)=A.
\] 
By considering $\varphi\circ\psi$ instead of $\varphi$, we may assume $A=(1/2)I$ without loss of generality.  
By Corollary \ref{pqyl}, the triple $\varphi(0)=0, \varphi((1/2)I), \varphi(I)=I$ is not in timelike position. 
Thus we see that one of the following holds: 
\begin{itemize}
\item Either $\varphi((1/2)I)\in \cC_0\cup \cC_I$, or
\item the triple $0, \varphi((1/2)I), I$ is in spacelike position. 
\end{itemize}

\begin{claim}
If $\varphi((1/2)I)\in \cC_0\cup \cC_I$, then we have
\[
\varphi((1/2)I)\in \bigcup_{P \in \mathcal{P}} \left( \ell_{0, \varphi (P)} \cup \ell_{\varphi (P), I}  \right).
\]
\end{claim}
\begin{proof}
Assume that  $\varphi((1/2)I)\in \cC_0\setminus \{0\}$. 
Then there are $a\in {\oR}\setminus \{0\}$ and $P_0\in \cP$ such that $\varphi((1/2)I)=aP_0$. 
Fix any $Q\in \cP$. 
We have $0\sim (1/2)Q\sim (1/2)I$, thus 
\[
0=\varphi(0)\sim \varphi((1/2)Q)\sim \varphi((1/2)I)=aP_0.
\] 
It follows that $\varphi((1/2)Q)=bP_0$ for some $b\in {\oR}$. 
If $b\neq 0$, then $\varphi(Q)=P_0$ because 
\[
bP_0=\varphi((1/2)Q)\sim \varphi(Q)\in \cP.
\] 
If $b= 0$, then $\varphi(Q^\perp+(1/2)Q)=P_0$ because 
\[
aP_0=\varphi((1/2)I)\sim \varphi(Q^\perp+(1/2)Q)\sim \varphi((1/2)Q)=0
\] 
and $\varphi(Q^\perp+(1/2)Q)\sim \varphi(I)=I$. 
This in turn shows $\varphi(Q^\perp)=P_0$ because 
\[
P_0=\varphi(Q^\perp+(1/2)Q)\sim \varphi(Q^\perp)\in \cP.
\] 
In both cases, we see that $\varphi((1/2)I)=aP_0\in  \bigcup_{P\in \cP}\ell_{0, \varphi(P)}$.
Similarly, we obtain  $\varphi((1/2)I)\in  \bigcup_{P\in \cP}\ell_{\varphi(P), I}$ when $\varphi((1/2)I)\in \cC_I\setminus\{I\}$.
\end{proof}

As before, we use the symbol $J = \left[ \begin{matrix} 1 & 0 \cr 0 & -1 \cr \end{matrix} \right]$.

\begin{claim}
If the triple $0, \varphi((1/2)I), I$ is in spacelike position, then
\begin{equation}\label{cccc}
\varphi((1/2)I)\in \bigcap_{P\in \cP} \cC_{\varphi(P)}.
\end{equation}
\end{claim}
\begin{proof}
Assume that the triple $0, \varphi((1/2)I), I$ is in spacelike position.
One may find an automorphism $\psi_3$ of $\oH$ such that the coherency preserver $\varphi':=\psi_3\circ \varphi\colon [0, I]\to \oH$ satisfies 
\[
\varphi'(0)=0,\ \  \varphi'((1/2)I)=(1/2)J,\ \ \text{and} \ \ \varphi'(I)=J
\]
(see Example \ref{IJ}). 

Let $P\in \cP$. 
Then we have $\varphi'(P)\in {\mathcal{S}}_{\varphi'(0), \varphi'(I)} = {\mathcal{S}}_{0, J}$.
Assume that $\varphi'(P)\in H_2$. 
We aim to obtain a contradiction.
Observe that $\varphi'((1/2)P)$ is the unique point (see Remark \ref{ceruk}) that is coherent to the three points 
\[
\varphi'(0)=0,\ \ \varphi'(P),\ \ \text{and}\ \  \varphi'((1/2)I)=(1/2)J.
\] 
Since $(1/2)\varphi'(P)\sim 0, \varphi'(P), (1/2)J$, we obtain $\varphi'((1/2)P)=(1/2)\varphi'(P)$. 
Similarly, $\varphi'(P+(1/2)P^\perp)$ is the unique point that is coherent to the three points 
\[
\varphi'(P),\ \ \varphi'(I)=J,\ \ \text{ and }\varphi'((1/2)I)=(1/2)J,
\]
which leads to $\varphi'(P+(1/2)P^\perp)=(1/2)\varphi'(P)+(1/2)J$. 
Then, the fact that $\varphi'((1/2)P^\perp)$ is the unique point that is coherent to the three points 
\[
\varphi'(0)=0,\ \ \varphi'((1/2)I)=(1/2)J,\ \ \text{and}\ \ \varphi'(P+(1/2)P^\perp)=(1/2)\varphi'(P)+(1/2)J
\]
implies that $\varphi'((1/2)P^\perp)= (1/2)J-(1/2)\varphi'(P)$. 

By Lemma \ref{nnjk}, the set $\cC_I\cap \cC_{(1/2)P}\cap \cC_{(1/2)P^\perp}\subset [0, I]$ is nonempty. 
It follows that 
\[
\begin{split}
\emptyset &\neq \varphi'(\cC_I\cap \cC_{(1/2)P}\cap \cC_{(1/2)P^\perp})\\
&\subset \cC_{\varphi'(I)}\cap \cC_{\varphi'((1/2)P)}\cap \cC_{\varphi'((1/2)P^\perp)}\\
&=\cC_J\cap \cC_{(1/2)\varphi'(P)}\cap \cC_{(1/2)J-(1/2)\varphi'(P)}.
\end{split}
\]
However, Corollary \ref{SJS} together with Lemma \ref{noauto} shows that 
\[
\cC_J\cap\cC_{(1/2)\varphi'(P)}\cap \cC_{(1/2)J-(1/2)\varphi'(P)}=\emptyset,
\] so we get to a contradiction.

It follows that $\varphi'(P)\notin H_2$. 
Thus we obtain $\varphi'(P)=\infty Q$ for some $Q\in \mathcal{P}$ because $0=\varphi'(0)\sim \varphi'(P)$. 
Since $J=\varphi'(I)\in \cC_{\varphi'(P)}=\cC_{\infty Q}$, we have $Q^\perp JQ^\perp=0$. Thus $Q^\perp (1/2)JQ^\perp=0$, which in turn implies $(1/2)J\in \cC_{\infty Q}=\cC_{\varphi'(P)}$. 
Since $P\in \cP$ is arbitrary, we arrive at the conclusion $\varphi'((1/2)I) =(1/2)J\in \bigcap_{P\in \cP} \cC_{\varphi'(P)}$. 
This clearly implies \eqref{cccc}.
\end{proof}

By the above claims, we get 

\begin{equation}\label{lemma07}
\varphi ([0,I]) \subset \bigcup_{P \in \mathcal{P}} \left( \ell_{0, \varphi (P)} \cup \ell_{\varphi (P), I}  \right) \ \cup \ \bigcap_{P \in \mathcal{P}} \mathcal{C}_{\varphi (P)}.
\end{equation}

\begin{claim}
If $\varphi(\cP)$ is a singleton $\{Q\}$, then $\varphi([0, I])\subset \mathcal{C}_{Q}$. 
\end{claim}
\begin{proof}
By \eqref{lemma07}, we get
\[
\varphi([0, I])\subset \bigcup_{P\in \cP}(\ell_{0, \varphi(P)}\cup \ell_{\varphi(P), I}) \cup \bigcap_{P\in \cP} \cC_{\varphi(P)}
= \ell_{0, Q}\cup \ell_{Q, I} \cup \cC_Q
= \cC_Q.
\]
\end{proof}

From now on, we assume that $\varphi(\cP)$ has at least two points.
Consider the automorphism $\psi$ of $\oH$ defined by $X\mapsto (I-X)^{-1} -I$, $X\in \oH$. 
By the definition of the inversion, we have 
\[
\psi([0, I])=\oHp:=\{aP+bP^\perp\,:\, a, b\in [0,\infty],\,P\in \cP\}. \index{$Hp$@$\oHp$}
\]
In the rest of the proof, to simplify the discussion, we work with the coherency preserver $\Phi:=\psi\circ\varphi\circ \psi^{-1}\colon \oHp\to \oH$. 
Note that $\Phi(0)=0$, $\Phi({\oi})={\oi}$ and that $\Phi(\infty\cP)\subset \infty\cP:=\{\infty P\,:\, P\in \cP\}$\index{$Pinfty2$@$\infty\cP$}.
Therefore, for each $P\in \cP$, there is a unique $P'\in \cP$ with $\Phi(\infty P)=\infty P'$.

For distinct $P, Q\in \cP$,
we define $\pp:= \oHp\cap \Pi_{P,Q}$\index{$PiPQ2$@$\pp$}, where 
\[
\Pi_{P,Q} = \{aP+bQ\,:\, a, b\in \bR\}\cup \ell_{\infty P, \oi}\cup \ell_{\infty Q, \oi}
\] 
is the surface as in \eqref{PI}.

\begin{lemma}\label{nasunek}
Let $P, Q\in \cP$ be distinct elements. 
Assume that $T$ is an invertible $2\times 2$ matrix such that $TPT^*=E_{11}$ and $TQT^*=E_{22}$. Let $\psi$ denote the automorphism $X\mapsto TXT^*$ of $\oH$. Then
\begin{itemize}
\item $\psi(\Pi_{P,Q})$ is the diagonal surface $\{ a E_{11} + b E_{22} \, : \, a,b \in {\oR} \}$ as in \eqref{ds}, and
\item $\psi(\oHp) = \oHp$.
\end{itemize}
\end{lemma}

\begin{proof}
Clearly, $\psi(\Pi_{P,Q})$ is a surface that contains $0, \infty E_{11}, \infty E_{22}, \oi$. By Lemma \ref{inftysurface}, the diagonal surface is the unique surface that contains these points. Consequently,  $\psi(\Pi_{P,Q})$ is the diagonal surface. 

Let $S$ be any invertible $2\times 2$ matrix. For $A \in H_2$ satisfying $A \ge 0$ we have $SAS^\ast \ge 0$. If $P$ is any rank one projection and $b \in [0, \infty]$ then $S(\infty P +bP^\perp )S^\ast$ is equal to $\infty Q + c Q^\perp$ for some $Q \in \mathcal{P}$ and some $c \in [0, \infty]$. It follows that 
$\psi(\oHp) \subset \oHp$ and $\psi^{-1}(\oHp) \subset \oHp$, and consequently, $\psi(\oHp) = \oHp$.
\end{proof}

\begin{lemma}\label{inftysquare}
Let $P, Q\in \cP$ be distinct elements. 
Then there is an automorphism of $\oH$ that maps $\pp$ onto ${\square}_{E_{11}}=\{ a E_{11} + b E_{22} \, : \, a,b \in [0,1] \}$.
\end{lemma}
\begin{proof}
We may take an invertible $2\times 2$ matrix $T$ such that $TPT^*=E_{11}$ and $TQT^*=E_{22}$. 
Let $\psi_1$ denote the affine automorphism $X\mapsto TXT^*$ of $\oH$. 
By Lemma \ref{nasunek} we have $\psi_1(\pp)=\{aE_{11}+bE_{22}\,:\, a, b\in [0, \infty]\}$. 
Let $\psi_2$ denote the automorphism $X\ \mapsto I - (I+X)^{-1}$. 
Then we obtain $\psi_2(\{aE_{11}+bE_{22}\,:\, a, b\in [0, \infty]\}) ={\square}_{E_{11}}$ by definition.
Thus the automorphism $\psi_2\circ \psi_1$ maps $\pp$ onto ${\square}_{E_{11}}$. 
\end{proof}

We study a pair $P_1, P_2\in \cP$ satisfying $P_1'\neq P_2'$. Note that the condition  $P_1'\neq P_2'$ is equivalent to $ \mathrm{tr}\,(P_1'{P_2'}^\perp) \not= 0$. 
Observe that $0, \infty P_1', \infty P_2', \oi$ are contained in $\Phi(\ppp)$.
Since $P_1'\neq P_2'$, Lemmas \ref{inftysquare}, \ref{where}, and \ref{inftysurface} imply that 
\begin{equation}\label{phippp}
\Phi(\ppp)\subset\Pi_{P_1',P_2'}.
\end{equation} 

\begin{claim}
We have
\begin{equation}\label{njix}
\Phi(\ppp)\subset \ell_{0, \infty P_1'}\cup \ell_{\infty P_1', \oi}\cup \ell_{0, \infty P_2'}\cup \ell_{\infty P_2', \oi}.
\end{equation}
\end{claim}
\begin{proof}
Observe that \eqref{lemma07} assures that 
\[
\Phi(\oHp)\subset \bigcup_{P\in \cP}(\ell_{0, \infty P'}\cup \ell_{\infty P', \oi}) \cup \bigcap_{P\in \cP} \cC_{\infty P'}\subset \cC_0\cup \cC_\oi \cup \cC_{\infty P_1'}.
\]
On the other hand, it is clear that 
\[
\cC_0 \cap \Pi_{P_1',P_2'} = \ell_{0, \infty P_1'} \cup \ell_{0, \infty P_2'}, \ \ \cC_{\oi} \cap \Pi_{P_1',P_2'} = \ell_{ \infty P_1', \oi} \cup \ell_{\infty P_2', \oi},
\]
and
\[
\mathcal{C}_{\infty P_1'} \cap \Pi_{P_1',P_2'} = \ell_{0, \infty P_1'} \cup \ell_{ \infty P_1', \oi}
\]
hold.
Therefore, \eqref{phippp} leads to \eqref{njix}.
\end{proof}

For each $P\in \cP$, we have $\Phi(\ell_{0, \infty P}\cap \oHp)\subset\ell_{0, \infty P'}$, so we may define a function $f_P\colon [0,\infty]\to {\oR}$ \index{$f_P$} by 
\[
\Phi(tP)=f_P(t)P',\ \ t\in [0,\infty].
\] 
Similarly, we may define a function $g_P\colon [0,\infty]\to {\oR}$ \index{$g_P$} by 
\[
\Phi(\infty P+tP^\perp)=\infty P'+ g_P(t){P'}^\perp, \ \ t\in [0,\infty].
\]
Note that $f_P(0)=g_P(0)=0$ and $f_P(\infty)=g_P(\infty)=\infty$ hold. 

In what follows, we apply the rule $c\cdot \infty = \infty\cdot c=\infty$ when $c\in (0, \infty)$.
Let $t\in [0, \infty]$.
Since $tP_1\sim \infty P_2 + t\mathrm{tr}\,(P_1P_2^\perp) P_2^\perp$, we obtain
\[
f_{P_1}(t)P_1' = \Phi(tP_1)\sim \Phi(\infty P_2 + t\mathrm{tr}\,(P_1P_2^\perp) P_2^\perp) = \infty P_2' + g_{P_2}(t\mathrm{tr}\,(P_1P_2^\perp)){P_2'}^\perp,
\]
which in turn implies 
\begin{equation}\label{trace}
f_{P_1}(t) \mathrm{tr}\,(P_1'{P_2'}^\perp) = g_{P_2}(t\mathrm{tr}\,(P_1P_2^\perp)).
\end{equation}
Note that equality \eqref{trace} is valid in the case $f_{P_1}(t)=\infty$ or $g_{P_2}(t\mathrm{tr}\,(P_1P_2^\perp))=\infty$, too. 
Actually, in the case when $f_{P_1}(t) \in \{ 0, \infty \}$ or $g_{P_2}(t \mathrm{tr}\,(P_1P_2^\perp)) \in \{ 0, \infty \}$ we have
\begin{equation}\label{trace2}
f_{P_1}(t)  = g_{P_2}(t\mathrm{tr}\,(P_1P_2^\perp)).
\end{equation}

\begin{claim}\label{twolines}
Either 
\begin{equation}\label{op1}
f_{P_1}([0, \infty])\subset \{0, \infty\} \ \ \text{and} \ \ \Phi(\ppp)\subset\ell_{0, \infty P_2'}\cup\ell_{\infty P_1', \oi},
\end{equation}
 or 
\begin{equation}\label{op2}
g_{P_1}([0, \infty])\subset \{0, \infty\} \ \ \text{and} \ \ \Phi(\ppp)\subset\ell_{0, \infty P_1'}\cup\ell_{\infty P_2', \oi}
\end{equation}
holds.
\end{claim}
\begin{proof}
Assume that there is a real number $a\in (0, \infty)$ such that $f_{P_1}(a)\notin \{0, \infty\}$. 
Put $A_1:=aP_1$ and $A_2:=\infty P_2 + a\mathrm{tr}\,(P_1P_2^\perp) P_2^\perp$.
Set $B_1:=\Phi(A_1)=f_{P_1}(a)P_1'$ and $B_2:=\Phi(A_2) = \infty P_2' + g_{P_2}(a\mathrm{tr}\,(P_1P_2^\perp)){P_2'}^\perp$. 
Then \eqref{trace} implies $B_2= \infty P_2' +  f_{P_1}(a)\mathrm{tr}\,(P_1'{P_2'}^\perp){P_2'}^\perp$.
It follows that 
\[
\ell_{B_1, B_2} = \{f_{P_1}(a) P_1' + bP_2'\,:\, b\in \mathbb{R}\}\cup \{B_2\}.
\]
Since $f_{P_1}(a)\notin \{0, \infty\}$, \eqref{njix} implies  
\begin{equation}\label{njiy}
\Phi(\ell_{A_1, A_2}\cap \oHp)\subset\ell_{B_1, B_2}\cap (\ell_{0, \infty P_1'}\cup\ell_{\infty P_1', \oi}\cup\ell_{0, \infty P_2'}\cup\ell_{\infty P_2', \oi}) = \{B_1, B_2\}.
\end{equation}
Since $A_1, A_2 \in \ppp$ and $A_1\sim A_2$,
Lemma \ref{inftysquare} combined with Lemma \ref{segment} implies that each element in $\ppp$ is coherent to some element of $\ell_{A_1, A_2}\cap \ppp$. 
This together with \eqref{njix} and \eqref{njiy} yields
\[
\Phi(\ppp) \subset (\cC_{B_1}\cup \cC_{B_2})\cap (\ell_{0, \infty P_1'}\cup\ell_{\infty P_1', \oi}\cup\ell_{0, \infty P_2'}\cup\ell_{\infty P_2', \oi}) =\ell_{0, \infty P_1'}\cup\ell_{\infty P_2', \oi}. 
\]
In particular, for every $t \in [0, \infty]$ we have 
\[
\Phi ( \infty P_1 + tP_{1}^\perp) = \infty P_1' +g_{P_1} (t) {P_1'}^\perp \in \ell_{0, \infty P_1'}\cup\ell_{\infty P_2', \oi}.
\]
Consequently, $g_{P_1}([0, \infty])\subset \{0, \infty\}$ holds.
Thus \eqref{op2} is established in this case.

Assume that $f_{P_1}((0, \infty))\subset \{0,\infty\}$. 
Let $a\in [0,\infty]$. 
By $aP_1\sim \infty P_2 + a\mathrm{tr}\,(P_1P_2^\perp) P_2^\perp$, we have
\[
\Phi(\infty P_2 + a\mathrm{tr}\,(P_1P_2^\perp) P_2^\perp)\sim \Phi(aP_1)\in \{0, \infty P_1'\}.
\]
This together with 
\[
\Phi(\infty P_2 + a\mathrm{tr}\,(P_1P_2^\perp) P_2^\perp)\in \ell_{\Phi(\infty P_2), \Phi(\oi)}=\ell_{\infty P_2', \oi}
\]
implies that $\Phi(\infty P_2 + a\mathrm{tr}\,(P_1P_2^\perp) P_2^\perp)=\infty P_2'$ when $\Phi(aP_1)=0$ and $\Phi(\infty P_2 + a\mathrm{tr}\,(P_1P_2^\perp) P_2^\perp)=\oi$ when $\Phi(aP_1)=\infty P_1'$. 
It follows that 
\[
\Phi(\ell_{aP_1, \infty P_2 + a\mathrm{tr}\,(P_1P_2^\perp) P_2^\perp}\cap \oHp)\subset \ell_{0, \infty P_2'}\cup \ell_{\infty P_1', \oi}.
\] 
Every point of $\ppp$ lies in $\ell_{aP_1, \infty P_2 + a\mathrm{tr}\,(P_1P_2^\perp) P_2^\perp}$ for some $a\in [0, \infty]$. (To see this, one can use the automorphism $\psi_1$ of $\oH$ defined as in the proof of Lemma 6.12 with $P_1 ,P_2$ instead of $P,Q$, which maps $\ppp$ onto $\{aE_{11}+bE_{22}\,:\, a, b\in [0, \infty]\}$ and the line 
$\ell_{aP_1, \infty P_2 + a\mathrm{tr}\,(P_1P_2^\perp) P_2^\perp}$ onto the line $\ell_{aE_{11}, \infty E_{22} + a E_{11}}$.)
Consequently, $\Phi(\ppp)\subset \ell_{0, \infty P_2'}\cup \ell_{\infty P_1', \oi}$ in this case. 
Thus \eqref{op1} is established.
\end{proof}

\begin{claim}\label{zeroinfty}
If the set $\Phi(\infty \cP)$ (or equivalently, $\varphi(\cP)$) has at least three elements, then $f_{P}([0, \infty])\subset \{0, \infty\}$ and $g_{P}([0, \infty])\subset \{0, \infty\}$ hold for every $P\in \cP$.
\end{claim}
\begin{proof}
For $P\in \cP$, we have $f_{P}([0, \infty])\subset \{0, \infty\}$ or $g_{P}([0, \infty])\subset \{0, \infty\}$ by Claim \ref{twolines}. 
In the latter case, take $P_1, P_2\in \cP$ such that $P'\neq P_1'\neq P_2'\neq  P'$. 
Then equation \eqref{trace2} implies $f_{P_1}([0, \infty])\subset \{0, \infty\}$, and similarly, we obtain $g_{P_2}([0, \infty])\subset \{0, \infty\}$, and then $f_{P}([0, \infty])\subset \{0, \infty\}$. 
Therefore, we see that $f_{P}([0, \infty])\subset \{0, \infty\}$ holds for every $P\in \cP$. 
Similarly, we have $g_{P}([0, \infty])\subset \{0, \infty\}$. 
\end{proof}

The same argument shows the following.
\begin{claim}\label{zeroorinfty}
Assume that the set $\Phi(\infty \cP)$ (or equivalently, $\varphi(\cP)$) has at least three elements. 
\begin{itemize}
\item If $f_{P}((0, \infty))= \{0\}$ or $g_{P}((0, \infty))= \{0\}$ for some $P\in \cP$, then $f_{P}((0, \infty))=g_P((0,\infty))= \{0\}$ holds for all $P\in \cP$. 
\item If $f_{P}((0, \infty))= \{\infty\}$ or $g_{P}((0, \infty))= \{\infty\}$ for some $P\in \cP$, then $f_{P}((0, \infty))=g_P((0,\infty))= \{\infty\}$ holds for all $P\in \cP$. 
\end{itemize}
\end{claim}

\begin{claim}\label{qc}
Assume that $P\in \cP$ satisfies $f_P([0, \infty])=\{0, \infty\}$.
Assume in addition that for any $c\in (0, 1)$ there exists $Q_c\in \cP$ satisfying $\mathrm{tr}\,(PQ_c)=c$ and $P'\neq Q_c'\neq (P^\perp)'$.
Then $f_P((0, \infty))=\{0\}$ or $f_P((0, \infty))=\{\infty\}$ holds.
\end{claim}
\begin{proof}
We have $\mathrm{tr}\, (PQ_c^\perp) =1-c$ and  $\mathrm{tr}\, (P^\perp Q_c^\perp) =c$.
Therefore, \eqref{trace2} implies that 
\[
\{0, \infty\} \ni f_{P}(t)=g_{Q_c}((1-c)t)=f_{P^\perp}\left(\frac{1-c}{c}t\right)
\]
for every $t\in [0, \infty]$ and $c\in (0,1)$. 
Let $t \in (0, \infty)$. Then there is $c \in (0,1)$ such that $((1-c)/c)t = 1$. 
Thus $f_P(t)=f_{P^\perp} (1)\in \{0, \infty\}$ for every $t\in (0,\infty)$.
\end{proof}

In the same way, we get the following.

\begin{claim}\label{qc1}
Assume that $R\in \cP$ satisfies $g_R([0, \infty])=\{0, \infty\}$.
Assume in addition that for any $c\in (0, 1)$ there exists $Q_c\in \cP$ such that $\mathrm{tr}\,(RQ_c)=c$ and $R'\neq Q_c'\neq (R^\perp)'$.
Then $g_R((0, \infty))=\{0\}$ or $g_R((0, \infty))=\{\infty\}$ holds.
\end{claim}

\begin{claim}\label{iorii}
Assume that there is no $P\in \mathcal{P}$ such that the image $\varphi(\mathcal{P}\setminus \{P\})$ is a singleton.
\begin{itemize}
\item If $f_P(t)=0$ for every $P\in \cP$ and every $t\in (0,\infty)$, then (i) holds. 
\item If $f_P(t)=\infty$ for every $P\in \cP$ and every $t\in (0,\infty)$, then (ii) holds. 
\end{itemize}
\end{claim}
\begin{proof}
Assume $f_P(t)=0$ for every $P\in \cP$ and every $t\in (0,\infty)$. 
Since $\Phi(\cP)$ is not a singleton,  an application of \eqref{trace2} implies that $g_P(t)=0$ for every $P\in \cP$ and every $t\in (0,\infty)$. 
These equations mean that $\varphi(cP)=0$ and $\varphi(P+cP^\perp)=\varphi(P)\in \cP$ for every $P\in \cP$ and every $c\in (0,1)$. 
Let $X\in (0, I)$. 
For each $P\in \cP$, we may find $c, d\in (0, 1)$ such that $cP\sim X\sim P+dP^\perp$ (Lemma \ref{0XI}). 
It follows that 
\[
0=\varphi(cP)\sim \varphi(X)\sim \varphi(P+dP^\perp)=\varphi(P). 
\]
This implies $\varphi(X)=0$ because $\varphi(\cP)\subset\cP$ has at least two elements.
Thus we have shown that (i) holds in this case. 
Similarly, we see that (ii) holds if $f_P(t)=\infty$ for every $P\in \cP$ and every $t\in (0,\infty)$. 
\end{proof}

Let us continue with a very simple lemma.

\begin{lemma}\label{qapl}
Let $m \colon ( 0, \infty) \to \{ 0 ,\infty \}$ and $g\colon  ( 0, \infty) \to {\oR}$ be two functions and $a,b$ real numbers with $0 < a < b$. Assume that for every $t \in (0, \infty)$ and every $s \in (a,b)$ we have
\[
m(t) = g(ts).
\]
Then either $g(t) = m(t) = 0$ for every $t \in (0, \infty)$, or $g(t) = m(t) =\infty$ for every $t \in (0, \infty)$.
\end{lemma}

\begin{proof}
It is easy to see that both sets $\{ t \in ( 0, \infty) \, : \, m(t) = 0 \}$ and  $\{ t \in ( 0, \infty) \, : \, m(t) = \infty \}$ are open. Since $(0, \infty)$ is connected, one of these two sets is equal to $(0, \infty)$.
\end{proof}

\begin{claim}
If the set $\Phi(\infty \cP)$ has at least three elements, then (i) or (ii) holds.
\end{claim}
\begin{proof}
By Claim \ref{zeroinfty}, $f_{P}([0, \infty])\subset \{0, \infty\}$ and $g_{P}([0, \infty])\subset \{0, \infty\}$ hold for every $P\in \cP$.
Assume that $\Phi$ restricted to $\infty \cP$ is injective. 
Then Claims \ref{qc}, \ref{zeroorinfty}, and \ref{iorii} imply that (i) or (ii) holds.

In the rest of this proof, we assume that there is a pair of distinct elements $P_1, P_2\in \cP$ such that $P_1'=P_2'$. 
By considering $\Phi\circ \psi_1$ for an appropriate affine automorphism $\psi_1$ of $\oH$ of the form $X\mapsto SXS^*$ for some invertible $2\times 2$ complex matrix $S$,
we may assume that $P_2=P_1^\perp$ without loss of generality. 
Assume that for each $c\in (0, 1)$ there is $Q_c\in \cP$ with $\mathrm{tr}\,(P_1Q_c)=c$ and $Q_c'\neq P_1'=(P_1^\perp)'$. 
Applying Claims \ref{qc}, 
\ref{zeroorinfty}, and \ref{iorii} again,
we see that either (i) or (ii) holds.

From now on, we assume that there is $c\in (0, 1)$ such that every projection $Q\in \cP$ with $\mathrm{tr}\,(P_1Q)=c$ satisfies $Q'=P_1'$. 
Assume that there is a pair of distinct points $P_3, P_4\in \cP$ such that $P_3'\neq P_1'\neq P_4'$ and the set 
\begin{equation}\label{p3p4}
\left\{ \frac{\mathrm{tr}\, (QP_4^\perp)}{\mathrm{tr}\,(QP_3^\perp)}\,:\, Q\in \cP,\, \mathrm{tr}\,(P_1Q)=c\right\}
\end{equation}
has nonempty interior.  
Then, $P_3'\neq Q'\neq P_4'$, so \eqref{trace2} implies that 
\begin{equation}\label{circle}
\{0, \infty\}\ni g_{P_3}(t)= f_Q\left(\frac{t}{\mathrm{tr}\,(QP_3^\perp)}\right) = g_{P_4}\left(t\frac{\mathrm{tr}\, (QP_4^\perp)}{\mathrm{tr}\,(QP_3^\perp)}\right)
\end{equation}
for every $t\in [0, \infty]$ and $Q\in \cP$ with $\mathrm{tr}\,(P_1Q)=c$. 
Using Lemma \ref{qapl}, we see that $g_{P_3}$ is constantly $0$ or constantly $\infty$ on $(0, \infty)$. 
It follows from Claims \ref{zeroorinfty} and \ref{iorii} that we may obtain the same conclusion as in the preceding paragraphs. 

If there is no pair $P_3, P_4$ with the above property, then the technical lemma below (Lemma \ref{134}) implies that $P'=P_1'$ holds for all but exactly two points in $\cP$. (Here we used the assumption that $\Phi(\infty\cP)$ has at least three points.)
In this case, let $P_5$ and $P_6$ be the exceptional points. 
Considering $\Phi\circ\psi_2$ for an appropriate affine automorphism $\psi_2$ of $\oH$ 
instead of $\Phi$, we may assume that $P_6=P_5^\perp$ without loss of generality.  We have $P_5'\neq P'\neq (P_5^\perp)'$ for every $P\in \cP\setminus\{P_5, P_5^\perp\}$. 
Let $d$ be any real number, $ 0 < d < 1$. For every $Q \in \mathcal{P}$ satisfying $\mathrm{tr}\,(P_5 Q)=d$, we have $P_5\neq Q\neq P_5^\perp$ and therefore $P_5'\neq Q'\neq (P_5^\perp)'$.
 Thus we may once again apply Claims \ref{qc}, \ref{zeroorinfty}, and \ref{iorii} to obtain the same conclusion.
\end{proof}

\begin{lemma}\label{134}
Let $P_1\in \cP$ and $0<c<1$. 
Set $\mathcal{R}:=\{Q\in \cP\,:\, \mathrm{tr}\,(P_1Q)=c\}$. 
For any $P_3\in \cP\setminus (\{P_1, P_1^\perp\}\cup\mathcal{R})$, there is at most one point $P_4\in \cP\setminus (\{P_1, P_1^\perp, P_3\}\cup\mathcal{R})$ with the property that the set
\begin{equation}\label{nonempty}
\left\{ \frac{\mathrm{tr}\, (QP_4^\perp)}{\mathrm{tr}\,(QP_3^\perp)}\,:\, Q\in \mathcal{R}\right\}
\end{equation}
has empty interior.
\end{lemma}
\begin{proof}
We may assume $P_1=E_{11}$ without loss of generality.
In this case, we have 
\[
\mathcal{R}=\left\{\left[ \begin{matrix} c & e^{it}\sqrt{c-c^2}  \cr e^{-it}\sqrt{c-c^2} & 1-c \cr \end{matrix} \right]\,:\, t\in \mathbb{R}\right\}.
\]
For a pair of elements $P_3, P_4\in  \cP\setminus(\{P_1, P_1^\perp\}\cup\mathcal{R})$, 
there are  $c_j \in (0, 1)$ with $c_j\neq c$ and $t_j \in [0,2\pi)$ such that 
\[
P_j=\left[ \begin{matrix} c_j & e^{it_j}\sqrt{c_j-c_j^2}  \cr e^{-it_j}\sqrt{c_j-c_j^2} & 1-c_j \cr \end{matrix} \right], \ \ j=3,4.
\]
Let $t\in \mathbb{R}$ and
\[
Q=\left[ \begin{matrix} c & e^{it}\sqrt{c-c^2}  \cr e^{-it}\sqrt{c-c^2} & 1-c \cr \end{matrix} \right] \in \mathcal{R}.
\]
Then
\[
\mathrm{tr}\, (QP_j^\perp) = c+c_j-2cc_j-2\cos(t-t_j)\sqrt{(c-c^2)(c_j-c_j^2)}.
\]
We have $c+c_j-2cc_j>0$. 
Note that $P_j\notin \mathcal{R}$ implies $\mathrm{tr}\, (QP_j^\perp) \not=0$, $j=3,4$, for all $Q \in \mathcal{R}$. We see that the set \eqref{nonempty} is equal to the range of the continuous function $f\colon \mathbb{R} \to \mathbb{R}$ given by
\[
f(t) = { a_4 - b_4 \cos ( t-t_4) \over a_3 - b_3 \cos ( t-t_3) } = { g_4(t) \over g_3(t) },
\]
where
\[
a_j = c+ c_j - 2cc_j,\ \ \  b_j = 2 \sqrt{(c-c^2)(c_j-c_j^2)}, \ \ \text{and}\ \  g_j(t)=a_j - b_j \cos ( t-t_j)
\]
for $j=3,4$, $t\in \bR$.
Therefore, the set \eqref{nonempty} has empty interior if and only if $f$ is a constant function, that is, $g_3$ and $g_4$ are linearly dependent. 
Assume that these equivalent conditions hold. Then, using the fact that $a_3, a_4, b_3, b_4$ are positive together with the assumption $0\leq t_3 , t_4 <2\pi$, and considering the behavior of the functions $g_3, g_4$, it is easily seen that $t_3=t_4$ and $a_4/a_3=b_4/b_3>0$.
It follows that
\[
\frac{c+c_4-2cc_4}{c+c_3-2cc_3}=\frac{\sqrt{c_4-c_4^2}}{\sqrt{c_3-c_3^2}},
\]
hence
\[
\left(\frac{c+c_4-2cc_4}{c+c_3-2cc_3}\right)^2=\frac{c_4-c_4^2}{c_3-c_3^2}.
\]
Observe that for a fixed $c_3$ this may be viewed as a quadratic equation with respect to $c_4$ one of whose solution is $c_3$. 
Thus the desired conclusion is obtained.
\end{proof}

In what follows, we assume that the set $\Phi(\infty \cP)$ (or equivalently,  $\varphi(\cP)$) has exactly two elements. 
Let $\Phi(\infty \cP)$ consist of exactly two points $\infty Q_1, \infty Q_2$ with $Q_1, Q_2\in \cP$. 
Set $\mathcal{Q}:=\{P\in \cP\,:\, P'=Q_1\}$.

\begin{claim}
Assume that both $\mathcal{Q}$ and $\cP\setminus \mathcal{Q}$ have at least two points.
Then (i) or (ii) holds. 
\end{claim} 
\begin{proof}
Apply Claim \ref{twolines} together with \eqref{trace2} to see the following.
By exchanging the roles of $\mathcal{Q}$ and $\cP\setminus \mathcal{Q}$ if necessary, we may assume $f_P([0, \infty])\subset \{0,\infty\}$ for every $P\in \mathcal{Q}$
and $g_P([0, \infty])\subset \{0,\infty\}$ for every $P\in \cP\setminus \mathcal{Q}$.

From here, we go along more or less the same lines as in the proof of the preceding claim for a while.
Take distinct elements $P_1, P_2\in \mathcal{Q}$.  
Considering $\Phi\circ \psi_1$ for an appropriate affine automorphism $\psi_1$ of $\oH$ instead of $\Phi$, we may assume that $P_2=P_1^\perp$ without loss of generality.

Assume that for each $c\in (0, 1)$ there is $Q_c\in \cP\setminus \mathcal{Q}$ with $\mathrm{tr}\,(P_1Q_c)=c$. 
Using Claim \ref{qc}, we see that $f_{P_1}(t)$ is constantly $0$ on $(0, \infty)$ or constantly $\infty$ on $(0,\infty)$.

In the current and the next paragraph, we assume that there is $c\in (0, 1)$ such that every projection $Q\in \cP$ with $\mathrm{tr}\,(P_1Q)=c$ is in $\mathcal{Q}$. 
Assume additionally that there is a pair of distinct points $P_3, P_4\in \cP\setminus \mathcal{Q}$ such that the set \eqref{p3p4} has nonempty interior.  
Then, \eqref{trace2} implies \eqref{circle}
for every $t\in [0, \infty]$ and $Q\in \mathcal{Q}$ with $\mathrm{tr}\,(P_1Q)=c$. 
Using Lemma \ref{qapl}, we see that $g_{P_3}((0, \infty))=\{0\}$ or $g_{P_3}((0, \infty))=\{\infty\}$ holds.

Assume that there is no pair $P_3, P_4$ with the above property.
Lemma \ref{134} implies that $P\in \mathcal{Q}$ holds for all but exactly two points in $\cP$. 
Let $\mathcal{P} \setminus \mathcal{Q} = \{ R , R_1 \}$ with $R \not= R_1$. 
As in the second paragraph of the proof, there is no loss of generality in assuming that $R_1 = R^\perp$. Clearly, for every $c \in (0,1)$ there exists $Q_c \in \mathcal{Q}$ such that $\mathrm{tr}\, (RQ_c) = c$. Moreover, $R' \not= Q_{c}' \not= (R^\perp)'$.  
Using Claim \ref{qc1}, we see that $g_{P}((0, \infty))=\{0\}$ or $g_{P}((0, \infty))=\{\infty\}$ holds for every $P\in \cP\setminus\mathcal{Q}$.

Therefore, in all cases, by using \eqref{trace2}, we see that one of the following holds. 
\begin{itemize}
\item $f_Q(t)=g_P(t)=0$ for every $Q\in \mathcal{Q}$, $P\in \cP\setminus \mathcal{Q}$, and $t\in (0, \infty)$, or
\item $f_Q(t)=g_P(t)=\infty$ for every $Q\in \mathcal{Q}$, $P\in \cP\setminus \mathcal{Q}$, and $t\in (0, \infty)$.
\end{itemize}

Let us assume that the first condition holds. 
Let $P\in \cP$.
Then we may find $Q\in \mathcal{Q}$ that is different from $P$. 
For each $t\in (0,\infty)$, we have $\infty P +t \mathrm{tr}\,(Q{P}^\perp) P^\perp \sim t Q$ and $\Phi(tQ)=0$, thus we get $\Phi(\infty P +t \mathrm{tr}\,(Q{P}^\perp) P^\perp)\sim 0$. 
It follows that $g_P((0,\infty))=\{0\}$ for every $P\in \cP$, hence \eqref{trace2} implies $f_P((0,\infty))=\{0\}$ for every $P\in \cP$.
Thus Claim \ref{iorii} implies that (i) holds.
Similarly, we obtain (ii) in the latter case.
\end{proof}

We consider the remaining case. 
Therefore, we assume that either $\mathcal{Q}$ or $\cP\setminus \mathcal{Q}$ has only one point. 
By exchanging the roles of $\mathcal{Q}$ and $\cP\setminus \mathcal{Q}$ if necessary, we may assume that $\mathcal{Q}$ has only one point $P_\times$. 
Put $Q_\times:= (P_\times)'$ and $Q_\circ := P'$ for any $P\in \cP\setminus \{P_\times\}$.

\begin{claim}\label{rojdanb}
Either $f_{P_\times}([0, \infty))=\{0\}$, or $g_{P_\times}((0, \infty])=\{\infty\}$ holds.
\end{claim}
\begin{proof}
Fix a point $P\in \cP\setminus\{P_\times\}$ and a real number $t\in (0, \infty)$. 
Let $R\in \cP\setminus \{P_\times\}$. 
We have
\[
\Phi(tP)\sim \infty Q_\circ + g_{R}(t\mathrm{tr}\,(PR^\perp))Q_\circ^\perp 
\]
because $tP\sim \infty R + t\mathrm{tr}\,(PR^\perp) R^\perp$.
Moreover, by \eqref{trace}, we obtain  
\begin{equation}\label{kkkjk}
f_{P_\times}(s) \mathrm{tr}\,(Q_\times Q_\circ^\perp) = g_{R}(s\mathrm{tr}\,(P_\times R^\perp))
\end{equation}
for every $s\in [0, \infty]$.
It follows that
\[
\Phi(tP)\sim \infty Q_\circ + f_{P_\times}\left(t\frac{\mathrm{tr}\,(PR^\perp)}{\mathrm{tr}\,(P_\times R^\perp)}\right) \mathrm{tr}\,(Q_\times Q_\circ^\perp) Q_\circ^\perp. 
\]
Note that the number ${\mathrm{tr}\,(PR^\perp)}/{\mathrm{tr}\,(P_\times R^\perp)}$ can take all values of $[0, \infty)$ as $R$ runs over $\cP\setminus \{P_\times\}$. 

Therefore, if $f_{P_\times}$ restricted to $[0, \infty)$ is not a constant function, then $\Phi(tP)$ needs to lie in $\ell_{\infty Q_\circ, \oi}$, which together with $0\sim tP$ implies $\Phi(tP)=\infty Q_\circ$ for every $P\in \cP\setminus \{P_\times\}$ and $t\in (0, \infty)$. In this case, \eqref{trace2} implies $g_{P_\times}(t)=\infty$ for every $t\in (0, \infty)$.
\end{proof}

The following claim completes the proof of Proposition \ref{proposition08}.
\begin{claim}\label{saigo}
If $f_{P_\times}([0, \infty))=\{0\}$, then (iv) holds.
If $g_{P_\times}((0, \infty])=\{\infty\}$, then (iii) holds.
\end{claim}
\begin{proof}
We only confirm the latter statement because the other can be verified in a parallel manner.
Thus we assume $g_{P_\times}((0, \infty])=\{\infty\}$.
In this case, we have  $\Phi(\infty P_\times+ tP_\times^\perp)=\oi$ for every $t\in (0, \infty]$. 
It is clear from the definition of the coherency relation that every element of $\oHp \setminus  \ell_{0, \infty P_\times}$ is coherent to $\infty P_\times+ tP_\times^\perp$ for some $t\in (0,\infty]$. 
Thus, we get $\Phi(\oHp \setminus  \ell_{0, \infty P_\times})\subset \cC_{\oi}$.

By \eqref{trace2}, we have $\Phi(tP)=\infty Q_\circ$ for every $P\in \cP\setminus \{P_\times\}$ and $t\in (0, \infty]$.
On the other hand, we see that every $A\in\oHp \setminus  \ell_{0, \infty P_\times}$ is coherent to $tP$ for some $P\in \cP\setminus \{P_\times\}$ and $t\in (0, \infty]$. 
Indeed, this is clear if $A\sim 0$. 
If $A\not\sim 0$, then $A$ is coherent to $tP_\times^\perp$ for some $t\in (0,\infty]$. 
Therefore, we get $\Phi(\oHp \setminus  \ell_{0, \infty P_\times})\subset \cC_{\infty Q_\circ}$.

For any $c \in [0, \infty]$, we have $\Phi (cP_\times) \in \ell_{0, \infty Q_\times}$, so $\Phi (cP_\times) \not\in  \ell_{\infty Q_{\circ}, {\oi}}$. Hence, any $A \in \oHp \setminus  \ell_{0, \infty P_\times}$ that is coherent to $cP_\times$ is mapped by $\Phi$ to the unique point on the line $ \ell_{\infty Q_{\circ}, {\oi}}$ that is coherent to $\Phi (cP_\times)$.
It is now straightforward to deduce that (iii) holds.  
\end{proof}

\begin{remark}\label{saigo2}
In the language of $\varphi$, the condition $g_{P_\times}((0, \infty])=\{\infty\}$ in Claim \ref{saigo} means that $\varphi(P_\times +tP_\times^\perp) =\varphi(I)=I$ for every $t\in (0,1)$.
\end{remark}

Before we close this subsection, we give two lemmas which will be used in the subsequent arguments.
\begin{lemma}\label{ABC}
Let $j\in \{o, i, ii, iii, iv\}$.
Let $A, B, C\in H_2$ satisfy $A<C\leq B$. 
Let $\varphi\colon [A, B]\to \oH$ be a coherency preserving map with $d(\varphi(A), \varphi(B))=d(\varphi(A), \varphi(C))=2$.
Assume that $\varphi\colon [A, B]\to \oH$ satisfies Condition (j). 
Then $\varphi$ restricted to $[A, C]$ also satisfies Condition (j). 
\end{lemma}

\begin{proof}
By Corollary \ref{kiakia}, we see that $\varphi$ restricted to $[A, C]$ is degenerate. Thus, this restriction satisfies (at least) one of the five Conditions (o), (i), (ii), (iii), and (iv). For each $P \in \mathcal{S}_{A,C}$, there is a unique $P' \in \mathcal{S}_{A,B}$ such that $P \sim P'$. Clearly, $P'$ is the unique point on $\ell_{A,P}$ that is coherent to $B$. The correspondence $P \mapsto P'$ is a bijection of  
$ \mathcal{S}_{A,C}$ onto $ \mathcal{S}_{A,B}$. 
Indeed, for $R\in \mathcal{S}_{A,B}$, we take the unique point $Q$ on the line $\ell_{A,R}$ that is coherent to $C$. Then $Q' = R$.

The point $\varphi (P')$ is the unique point on $\ell_{\varphi (A), \varphi (P)}$ that is coherent to $\varphi (B)$. It follows that for $P,Q \in \mathcal{S}_{A,C}$ we have $\varphi (P) = \varphi (Q)$ if and only if 
$\varphi (P') = \varphi (Q')$. Hence, if $\varphi$ satisfies (o), then the restriction of $\varphi$ to $[A,C]$ satisfies (o), and if $\varphi$ satisfies (i) or (ii), then the restriction of $\varphi$ to $[A,C]$ satisfies (i) or (ii), and
if $\varphi$ satisfies (iii) or (iv), then the restriction of $\varphi$ to $[A,C]$ satisfies (iii) or (iv). It is trivial to see that if $\varphi$ satisfies (j), $j \in \{ i, ii \}$, then $\varphi$ restricted to $[A,C]$ satisfies (j).

Assume next that $\varphi$ satisfies (iii). 
Then we have $\varphi([A, B]\cap\ell_{A, P_\times})\subset \ell_{\varphi(A), Q_\times}$ and $\varphi([A, B]\setminus\ell_{A, P_\times})\subset \ell_{Q_\circ, \varphi(B)}$.
It follows that $\varphi([A, C]\cap\ell_{A, P_\times})\subset \ell_{\varphi(A), Q_\times}$ and $\varphi([A, C]\setminus\ell_{A, P_\times})\subset \ell_{Q_\circ, \varphi(B)}$.
From this, it is straightforward to see that $\varphi$ restricted to $[A,C]$ also satisfies (iii).

Finally, assume that $\varphi$ satisfies (iv). 
Then we have $\varphi([A, B]\cap\ell_{P_\times, B})\subset \ell_{Q_\times, \varphi(B)}$ and $\varphi([A, B]\setminus\ell_{P_\times, B})\subset \ell_{\varphi(A), Q_\circ}$.
In this case, the assumption $d(\varphi(A), \varphi(C))=2$ implies that $C\in \ell_{P_\times, B}\setminus\{P_\times\}$. 
Thus we have $\ell_{P_\times, C}=\ell_{P_\times, B}$ and $\ell_{Q_\times, \varphi(C)}=\ell_{Q_\times, \varphi(B)}$, and it is not hard to see that $\varphi$ restricted to $[A,C]$ satisfies (iv).
\end{proof}

In the same way, we get the following lemma.

\begin{lemma}\label{ABC2}
Let $j\in \{o, i, ii, iii, iv\}$.
Let $A, B, C\in H_2$ satisfy $A\le C < B$. 
Let $\varphi\colon [A, B]\to \oH$ be a coherency preserving map with $d(\varphi(A), \varphi(B))=d(\varphi(C), \varphi(B))=2$.
Assume that $\varphi\colon [A, B]\to \oH$ satisfies Condition (j). 
Then $\varphi$ restricted to $[C, B]$ also satisfies Condition (j). 
\end{lemma}

\subsection{The case $d(\varphi(A), \varphi(B))=0$}\label{dab0}
We consider the case $d(\varphi(A), \varphi(B))=0$, or equivalently, $\varphi(A)=\varphi(B)$. 
In this case, we have either $\varphi([A, B])\subset \mathcal{C}_{\varphi(A)}=\mathcal{C}_{\varphi(B)}$, or there is $C\in (A, B)$ such that $d(\varphi(C), \varphi(A))=2$. 
In the former case, $\varphi$ is clearly of type $(\mathcal{C})$. 

We study the latter case.
By considering $\psi_2\circ\varphi\circ \psi_1$ for a pair of suitable automorphisms $\psi_1, \psi_2$ of $\oH$, we may assume that $A=-I$, $B=I$, and $\varphi(0)=0$, $\varphi(-I)=\varphi(I)=I$ (use Corollary \ref{cororder}).

\begin{proposition}\label{00II}
Let $\varphi\colon [-I, I]\to \oH$ be a coherency preserver with $\varphi(0)=0$ and $\varphi(-I)=\varphi(I)=I$. 
Then $\varphi$ is degenerate.
\end{proposition}

To prove this proposition, let us consider a coherency preserver $\varphi\colon [-I, I]\to \oH$ satisfying $\varphi(0)=0$ and $\varphi(-I)=\varphi(I)=I$. 
Observe that $\varphi$ is not standard because $\varphi(-I)=\varphi(I)$. 
By Corollary \ref{kiakia}  we see that $\varphi$ is not standard on every nonempty open subset of $[-I, I]$.
We have $\varphi(\cP)\subset\cP$ and $\varphi(-\cP)\subset \cP$ by the assumption. 
For each $P\in \cP$, we have $P\sim -P$ and thus $\varphi(P)\sim \varphi(-P)$. 
Since $\varphi(P),\varphi(-P)\in \cP$, we obtain $\varphi(P)=\varphi(-P)$.

The restriction $\varphi_1$ of $\varphi$ to $[0, I]$, and the mapping $\varphi_2\colon [0, I]\to \oH$ defined by $\varphi_2(X)=\varphi(-X)$, $X\in [0, I]$, satisfy the assumption of Proposition \ref{proposition08}. 
Moreover, we know that $\varphi(P)=\varphi_1(P)=\varphi_2(P)$ for every $P\in \cP$.
By considering the mapping $X\mapsto \varphi(-X)$ instead of $\varphi$, if necessary, we see that we only need to consider the following seven possibilities: 
\begin{enumerate}[(I)]
\item  Both $\varphi_1$ and $\varphi_2$ satisfy Condition (o). 
\item  Both $\varphi_1$ and $\varphi_2$ satisfy Condition (i). 
\item  The maps $\varphi_1$ and $\varphi_2$ satisfy Condition (i) and Condition (ii), respectively. 
\item  Both $\varphi_1$ and $\varphi_2$ satisfy Condition (ii). 
\item  Both $\varphi_1$ and $\varphi_2$ satisfy Condition (iii). 
\item  The maps $\varphi_1$ and $\varphi_2$ satisfy Condition (iii) and Condition (iv), respectively.
\item  Both $\varphi_1$ and $\varphi_2$ satisfy Condition (iv). 
\end{enumerate}

Let us study each case.

\begin{claim}
If (I) holds, then $\varphi ([-I,I])$ is contained in one cone and $\varphi$ is of type $(\mathcal{C})$.
\end{claim}
\begin{proof}
Assume that both $\varphi_1$ and $\varphi_2$ satisfy Condition (o). 
Then we have $\varphi([-I,0]\cup [0,I])\subset \cC_Q$ and $\varphi(\cP)=\{Q\}$. 
For every $X\in [-I,I]\setminus([-I,0]\cup [0,I])$ we have $X \not\in (0,I)$, $X \le I$, and $X\not< 0$, and so Lemma \ref{c(s)} implies that $X$ is coherent to some point of $\cP$.
Thus $\varphi([-I,I])\subset \cC_Q$. 
\end{proof}

\begin{claim}\label{(II)}
Condition (II) never holds.
\end{claim}
\begin{proof}
If both $\varphi_1$ and $\varphi_2$ satisfy (i), then we have $\varphi([0, I)\cup (-I, 0])=\{0\}$. 
For each real number $c\in (0, 2)$ and $P\in \cP$, we see that $-I+cP^\perp$ is coherent to some point of $[0, I)\cup (-I, 0]$. 
By $-P\sim -I+cP^\perp \sim -I$ and 
\[
\varphi([0, I)\cup (-I, 0])=\{0\},\ \  \varphi(-P)=\varphi(P)\in \cP, \ \  \varphi(-I)=I,
\] we have $\varphi(-I+cP^\perp)=\varphi(P)\in \cP$. 
Similarly, we have $\varphi(I-cP^\perp)=\varphi(P)\in \cP$ for every $P\in \cP$ and $c\in (0, 2)$. 
By Lemma \ref{0XI}, for every pair $P, Q\in \cP$, $P \not= Q^\perp$ and  $c\in (0, 2)$, there is $d\in (0,2)$ such that $-I+cP^\perp\sim I-dQ^\perp$. 
Thus $\varphi(P)\in \cP$ is coherent to $\varphi(Q)\in \cP$, hence $\varphi(P)=\varphi(Q)$ whenever $P \not= Q^\perp$.
This yields that $\varphi(\cP)$ is a singleton, which contradicts our assumption.
\end{proof}

\begin{claim}\label{(III)}
Condition (III) never holds.
\end{claim}
\begin{proof}
If $\varphi_1$ satisfies (i) and $\varphi_2$ satisfies (ii), then we see that $\varphi ([0,I)) = \{ 0 \}$ and $\varphi ([-I,0)) = \{ I \}$. Set $\mathcal{N}:=(-I,I)\cap H_2^{+-}$. Each $A \in \mathcal{N}$ is coherent to some scalar matrix in $(0,I)$ and some scalar matrix in $(-I,0)$ and therefore $\varphi (\mathcal{N}) 
\subset {\mathcal{S}}_{0, I}=\cP$. 
Recall that any two distinct elements in $\cP$ are not coherent.
On the other hand, $\mathcal{N}$ is open and connected, so Lemma \ref{path} implies that $\varphi(\mathcal{N})$ is a singleton $\{ Q \}$ in $\mathcal{P}$. For every $P \in \cP$, we have 
$\varphi (P) = \varphi (P + (1/2)P^\perp)$ because $\varphi_1$ satisfies (i), while
\[
\varphi (P + (1/2)P^\perp) \sim 
\varphi (-(1/2) P + (1/2)P^\perp) \in \varphi(\mathcal{N})=\{Q\}.
\]
This further implies that $\varphi(\cP)=\{Q\}$ and we again get a contradiction.
\end{proof}

\begin{claim}
If (IV) holds, then $\varphi([-I, I]\setminus\{0\})\subset \mathcal{C}_I$, thus $\varphi$ is of type $(\mathcal{C})$.
\end{claim}
\begin{proof}
Assume that both $\varphi_1$ and $\varphi_2$ satisfy (ii). 
Then we have $\varphi([-I, 0)\cup (0, I])=\{I\}$. 
Since every point of $[-I, I]\setminus \{0\}$ is coherent to some point in $[-I, 0)\cup (0, I]$, we have $\varphi([-I, I]\setminus\{0\})\subset \mathcal{C}_I$.
\end{proof}

In what follows, we assume that $\varphi_1$ satisfies (iii) or (iv). 
Then there is $P_\times\in \mathcal{P}$ such that the image $\varphi(\mathcal{P}\setminus \{P_\times\})$ is a singleton $\{Q_\circ\}$ and $Q_\times:=\varphi(P_\times)\neq Q_\circ$.
Since $\varphi(P)=\varphi(-P)$ for every $P\in \cP$, we obtain $\varphi(-\mathcal{P}\setminus \{-P_\times\})=\{Q_\circ\}$ and $\varphi(-P_\times)=Q_\times$.

\begin{lemma}\label{hodai}
If $X \in[-I, I]\setminus \ell_{0, P_{\times}}$, then $X$ is coherent to some point of $\cC_0\cap [-I, I]\setminus \ell_{0, P_{\times}}$.
\end{lemma}
\begin{proof}
We may write $X=aP+bP^\perp$ for some $P\in \cP\setminus\{P_\times\}$ and $a,b\in [-1,1]$, $a\neq 0$. Thus $X$ is coherent to $aP\in \cC_0\cap [-I, I]\setminus \ell_{0, P_{\times}}$.
\end{proof}

\begin{lemma}\label{mukpuk}
If $X \in [-I,  I] \setminus \ell_{0, P_\times}$, then $X$ is coherent to some point of 
\[
\{P_\times + cP_\times^\perp\,:\, c\in [-1,1]\setminus\{0\}\}\cup \{-P_\times + cP_\times^\perp\,:\, c\in [-1,1]\setminus\{0\}\}.
\] 
\end{lemma}
\begin{proof}
By Corollary \ref{JeDr}, we know that every point of $[-I, I]\setminus \cC_{P_\times}$ is coherent to some point of $\{P_\times + cP_\times^\perp\,:\, c\in [-1,1]\setminus\{0\}\}$, and every point of $[-I, I]\setminus \cC_{-P_\times}$ is coherent to some point of $\{-P_\times + cP_\times^\perp\,:\, c\in [-1,1]\setminus\{0\}\}$. 
Since $\cC_{P_\times}\cap \cC_{-P_\times}=\ell_{0, P_\times}$, we get to the desired conclusion.
\end{proof}

\begin{claim}\label{vtovii}
If (V), (VI), or (VII) holds, then  $\varphi$ is of type $(\ell)$.
\end{claim}
\begin{proof}
Let $b\in (0, 1]$. 
Then $-P_\times+bP_\times^\perp$ is coherent to $-I$ and $-P_\times$. 
Moreover, since $-P_\times+bP_\times^\perp\not\sim P_\times$, Lemma \ref{c(s)}  implies that $-P_\times+bP_\times^\perp$ is also coherent to some point of ${\mathcal{S}}_{0,I}\setminus \{P_\times\}=\cP\setminus \{P_\times\}$. 
From
\[
\varphi(-I)=I,\ \ \varphi(-P_\times)=Q_\times,\ \ \text{and}\ \ \varphi(\cP\setminus \{P_\times\})=\{Q_\circ\},
\]
we obtain 
\begin{equation}\label{b1}
\varphi(-P_\times+bP_\times^\perp)=I.
\end{equation} 
Observe that $bP_\times^\perp$ is coherent to the three points $0,  P_\times^\perp, -P_\times+bP_\times^\perp$. 
Since 
\[
\varphi(0)=0,\ \ \varphi(P_\times^\perp)=Q_\circ, \ \ \text{and}\ \ \varphi(-P_\times+bP_\times^\perp)=I, 
\]
we get $\varphi(bP_\times^\perp)=Q_\circ$.
Observe that $P_\times +bP_\times^\perp$ is coherent to the three points $P_\times, I, bP_\times^\perp$. 
Since
\[
\varphi(P_\times)=Q_\times,\ \ \varphi(I)=I, \ \ \text{and}\ \ \varphi(bP_\times^\perp)=Q_\circ,
\]
we get 
\begin{equation}\label{b2}
\varphi(P_\times +bP_\times^\perp)=I.
\end{equation} 
A similar discussion shows 
\begin{equation}\label{b3}
\varphi(P_\times-bP_\times^\perp)=I
\end{equation}  
and 
\begin{equation}\label{b4}
\varphi(-P_\times -bP_\times^\perp)=I.
\end{equation} 
By Lemma \ref{mukpuk}, every point in $[-I, I]\setminus \ell_{0, P_{\times}}$ is coherent to some point of 
\[
\{P_{\times}+cP_{\times}^\perp\,:\, c\in [-1,1]\setminus\{0\}\} \cup \{-P_{\times}+cP_{\times}^\perp\,:\, c\in [-1,1]\setminus\{0\}\}.
\] 
By \eqref{b1}, \eqref{b2}, \eqref{b3}, \eqref{b4}, $\varphi$ sends this set to $\{I\}$.

On the other hand, by  Claim \ref{saigo} (see also Remark \ref{saigo2}) and \eqref{b2} (resp.\ \eqref{b4}), we see that $\varphi_1$ (resp.\ $\varphi_2$) satisfies Condition (iii).
It follows that $\varphi(\cC_0\cap [-I, I]\setminus \ell_{0, P_{\times}})=\{Q_{\circ}\}$. 
By Lemma \ref{hodai}, every point in $[-I, I]\setminus \ell_{0, P_{\times}}$ is coherent to some point of $\cC_0\cap [-I, I]\setminus \ell_{0, P_{\times}}$.
Therefore, we get $\varphi(X)\in \ell_{Q_{\circ}, I}$ for every $X\in [-I, I]\setminus \ell_{0, P_{\times}}$. 
Thus $\varphi$ is of type $(\ell)$.
\end{proof}

Thus we have completed the proof of Proposition \ref{00II}.

\begin{corollary}\label{d=0}
Let $A, B\in H_2$ satisfy $A<B$. 
Let $\varphi\colon [A, B]\to \oH$ be a coherency preserving map with $\varphi(A)=\varphi(B)$. 
Then $\varphi$ is degenerate.
\end{corollary}

\subsection{The case $d(\varphi(A), \varphi(B))=1$}\label{dab1}

Now let us consider a coherency preserver $\varphi\colon [A, B]\to \oH$ with $d(\varphi(A), \varphi(B))=1$.  
By considering $\psi_2\circ\varphi\circ \psi_1$ for a pair of suitable automorphisms $\psi_1, \psi_2$ of $\oH$, we may assume that $A=0$, $B=I$, and $\varphi(0)=0$, $\varphi(I)=E_{11}$.

\begin{proposition}\label{E11}
Let $\varphi\colon [0, I]\to \oH$ be a coherency preserver with $\varphi(0)=0$, $\varphi(I)=E_{11}$. 
Then $\varphi$ is degenerate.
\end{proposition}

The proof of this proposition is involved, so we will separate it into claims, as usual.
We assume that $\varphi\colon [0, I]\to \oH$ is a coherency preserver with $\varphi(0)=0$, $\varphi(I)=E_{11}$. 
Observe that $\varphi$ is not standard because $0\not\sim I$ and $\varphi(0)\sim \varphi(I)$. 
By Corollary \ref{kiakia}, we see that $\varphi$ is not standard on every nonempty open subset of $[0, I]$.
Set $\ell:=\ell_{0, E_{11}}$. 
Note that $\varphi(\cP)\subset \ell_{\varphi(0), \varphi(I)}=\ell$ holds.

\begin{claim}
If $\varphi([0, I]\setminus(0,I))\subset \ell$, then $\varphi([0, I])$ is contained in a cone and $\varphi$ is of type $(\mathcal{C})$.
\end{claim}
\begin{proof}
Let $\varphi$ satisfy $\varphi([0, I]\setminus(0,I))\subset \ell$.
Assume in addition that there is a point $A\in (0, I)$ such that  $\varphi(A)\notin\mathcal{C}_0\cup \mathcal{C}_{E_{11}}$. 
Let $A'$ be the unique point in $\ell$ such that $\varphi(A)\sim A'$.
Note that $\mathcal{S}_{0, A}$ and $\mathcal{S}_{A, I}$ are both contained in $[0,I]\setminus(0,I)$.
Since $\varphi([0, I]\setminus (0, I))\subset  \ell$, we have $\varphi(\mathcal{S}_{0, A})=\varphi(\mathcal{S}_{A, I}) = \{A'\}$. 
Thus, the restriction of $\varphi$ to $[0, A]$ and that to $[A, I]$ both satisfy (o). 
It follows that $\varphi([0,A]\cup [A, I])\subset \cC_{A'}$.
Moreover, we see from Lemma \ref{c(s)} that every element of $[0,I]\setminus([0, A)\cup (A, I])$ is coherent to some point of ${\mathcal{S}}_{0, A}$.
Since $\varphi({\mathcal{S}}_{0, A})=\{A'\}$, we get $\varphi([0,I]\setminus([0, A)\cup (A, I]))\subset\cC_{A'}$. 
Thus we have shown that $\varphi([0,I])\subset \cC_{A'}$ in this case.

It remains to consider the case where
\[
\varphi([0, I])\subset\mathcal{C}_{\varphi(0)}\cup \mathcal{C}_{\varphi(I)}=\mathcal{C}_0 \cup \mathcal{C}_{E_{11}}.
\] 
If every $X\in (0, I)$ satisfies $\varphi(X)\in \mathcal{C}_{E_{11}}$, then the assumption $\varphi([0, I]\setminus (0, I))\subset \ell$ implies $\varphi([0, I])\subset \cC_{E_{11}}$, as desired. 
Therefore, we assume that there is an element $X\in (0, I)$ satisfying $\varphi(X)\in \mathcal{C}_0\setminus \mathcal{C}_{E_{11}}$.
For any $W \in \mathcal{S}_{X,I}$, we have $\varphi (W) \in \varphi([0, I]\setminus (0, I))\subset \ell$. 
Since $0$ is the unique point on the line $\ell$ that is coherent to $\varphi(X)$, we conclude that
$\varphi(\mathcal{S}_{X, I}) =\{0\}$. 
Thus Proposition \ref{newprop} implies that $\varphi$ satisfies Condition (o), and we see that $\varphi ([X,I]) \subset  \mathcal{C}_0$.
We claim that 
\begin{equation}\label{metka}
\varphi([0, I]\setminus [0, X))\subset \mathcal{C}_0.
\end{equation}
Indeed, we already know that $\varphi (Z) \in  \mathcal{C}_0$ whenever $Z \in [X,I]$. For $Z \in [0, I]\setminus ([0, X) \cup [X,I])$, we can apply Lemma \ref{c(s)} together with $\varphi(\mathcal{S}_{X, I}) =\{0\}$ to get  $\varphi (Z) \in  \mathcal{C}_0$.

We show that every $Y\in [0,X)$ satisfies $\varphi(Y)\in\mathcal{C}_0$. 
Assume contrarily that $\varphi(Y)\notin \mathcal{C}_0$.
Then $\varphi(Y)\in \mathcal{C}_{E_{11}}\setminus \ell$. 
Observe that $0\leq Y<X<I$ and thus Corollary \ref{A<B} implies $\mathcal{S}_{Y,X}\subset [0,I]$.
By \eqref{metka}, we have 
\begin{equation}\label{eq1}
\varphi(\mathcal{S}_{Y,X})\subset \mathcal{C}_0\cap \mathcal{C}_{\varphi(Y)}=\mathcal{S}_{0, \varphi(Y)}.
\end{equation}
On the other hand, $\mathcal{S}_{Y, I}\subset [0, I]\setminus (0, I)$ implies
\begin{equation}\label{eq2}
\varphi(\mathcal{S}_{Y, I})\in \ell \cap \mathcal{C}_{\varphi(Y)} = \{E_{11}\}.
\end{equation}
Since every point of $\mathcal{S}_{Y, X}$ is coherent to some point of $\mathcal{S}_{Y, I}$,   \eqref{eq1} and \eqref{eq2} lead to 
\[
\varphi(\mathcal{S}_{Y, X})\subset \mathcal{S}_{0, \varphi(Y)}\cap \cC_{E_{11}} = \{E_{11}\}.
\]
This contradicts $\varphi(X)\not\sim E_{11}$.
Thus we get $\varphi([0,I])\subset \cC_0$.
\end{proof}

Therefore, in what follows we assume that $\varphi([0, I]\setminus(0,I))\not\subset \ell$.
By considering the mapping $X\mapsto E_{11}-\varphi(I-X)$ instead of $\varphi$, if necessary, we may and do assume that there is $A\in \cC_I\cap [0, I]$ such that $\varphi(A)\notin \ell$.
Since $\varphi(\cP)\subset \ell$, we get $0<A$.
We will apply Proposition \ref{newprop} to the restriction $\varphi_1$ of $\varphi$ to $[0, A]$.

\begin{claim}
If $\varphi_1$ satisfies (o), then $\varphi([0,I])$ is contained in one cone and $\varphi$ is of type $(\mathcal{C})$.
\end{claim}
\begin{proof}
Suppose that $\varphi_1$ satisfies (o).
There is a singleton $\{ Q \}$ such that $\varphi (\mathcal{S}_{0,A}) = \{ Q \}$ and $\varphi ([0,A]) \subset \mathcal{C}_{Q}$. 
By Lemma \ref{c(s)}, every element of $[0,I] \setminus [0, A]$ is coherent to some element of $\mathcal{S}_{0,A}$. 
It follows from $\varphi (\mathcal{S}_{0,A}) = \{ Q \}$ that $\varphi([0,I] \setminus [0, A]) \subset \mathcal{C}_{Q}$.
\end{proof}

\begin{claim}\label{oc}
If $\varphi(\mathcal{S}_{0, A})\subset \ell$, then $\varphi_1$ satisfies (o), so $\varphi$ is of type $(\mathcal{C})$.
\end{claim}
\begin{proof}
Suppose that $\varphi(\mathcal{S}_{0, A})\subset \ell$. 
Let $C \in \mathcal{S}_{0, A} \,(\subset [0,I])$. 
Because $\varphi (C) \sim \varphi (A)$ and $\varphi (C) \in \ell$, we have $\varphi (C) = E_{11}$. 
Thus $\varphi_1$ satisfies (o).
\end{proof}

In what follows, we assume that $\varphi(\mathcal{S}_{0, A})\not\subset \ell$.
Then there is $B\in \mathcal{S}_{0, A}$ such that $\varphi(B)\notin \ell$. 
Observe that $E_{11}$ (resp.\ $0$) is the unique point on $\ell$ that is coherent to $\varphi(A)$ (resp.\ $\varphi(B)$).
Let $E\in\mathcal{P}$ be the unique projection satisfying $A\sim E$.
Then we have $E\sim A, 0, I$, which implies $\varphi(E)\in \cC_{\varphi(A)} \cap \ell=\{E_{11}\}$. 
Similarly, for the unique projection $F\in\mathcal{P}$ satisfying $B\sim F$, we have $F\sim B, 0, I$, and thus $\varphi(F)=0$. 

\begin{lemma}\label{kicmabol}
Let $E,F \in \cP$, $E \not= F$. 
Then there exists an automorphism $\psi$ of $\oH$ satisfying $\psi([0,I])=[0,I]$,  $\psi (0) = 0$, $\psi (I) = I$, $\psi (E) = E$, and $\psi (F) = E^\perp$.
\end{lemma}
\begin{proof}
Let $\psi_0$ be the automorphism $X\mapsto I-(I+X)^{-1}$.
Take an invertible $2 \times 2$ complex matrix $T$ such that $TET^\ast = E$ and $TFT^\ast = E^\perp$. 
Then the automorphism $X\mapsto \psi_0(T\psi_0^{-1}(X)T^\ast)$ satisfies the desired properties.
(See also Lemmas \ref{nasunek}, \ref{inftysquare} and their proofs.)
\end{proof}

Take an automorphism $\psi$ as in this lemma. 
By considering $\varphi\circ \psi^{-1}$ instead of $\varphi$, we may and do assume that $F=E^\perp$ without loss of generality. 
In this case, we have $A=E+cE^\perp$ for some $c\in (0,1)$,  which implies $B=cE^\perp$.

\begin{claim}\label{1i}
If $\varphi_1$ satisfies (i), then $\varphi([0, I]\setminus\{A\})\subset \cC_0$, and thus $\varphi$ is of type $(\mathcal{C})$.
\end{claim}
\begin{proof}
Recall that $A=E+cE^\perp$.
Suppose that $\varphi_1$ satisfies (i).
Then we have $\varphi([0, A)) =\{\varphi(0)\}=\{0\}$. 
In particular, we have $\varphi(tE)=0$ for every $t\in [0,1)$. 
Therefore, Lemma \ref{0XI} implies that $\varphi([0,I]\setminus\ell_{E, I})\subset \cC_0$.
It remains to consider an element of the form $E+tE^\perp$, $t\in [0,1]\setminus\{c\}$.
If $t\in [0,c)$, then $E+tE^\perp\sim tE^\perp\in [0,A)$ and hence we get $\varphi(E+tE^\perp)\sim \varphi(tE^\perp)=0$.
Let $t>c$
and assume that $\varphi(E+tE^\perp)\not\sim 0$.
Then Lemma \ref{ABC} implies that $\varphi$ restricted to $[0,E+tE^\perp]$ also satisfies (i). 
(Note that a coherency preserver satisfying (i) cannot satisfy (j) for $j\in \{o, ii, iii, iv\}$.)
It follows that $\varphi([0,E+tE^\perp]\setminus\{E+tE^\perp\})\subset \mathcal{C}_0$, which contradicts the facts $A\in [0,E+tE^\perp]\setminus\{E+tE^\perp\}$ and $\varphi(A)\not\sim 0$. 
Therefore, we get $\varphi(E+tE^\perp)\sim 0$.
Thus we have shown that $\varphi([0, I]\setminus\{A\})\subset \cC_0$.
\end{proof}

\begin{claim}
The mapping  $\varphi_1$ never satisfies (ii). 
\end{claim}
\begin{proof}
Recall that $A=E+cE^\perp$ and $B=cE^\perp$.
Suppose that $\varphi_1$ satisfies (ii). We are going to get a contradiction.
We have $\varphi((0,A])=\{\varphi(A)\}$. In particular, we get $\varphi(tE+cE^\perp)=\varphi(A)$ for every $t\in (0,1]$.
If $t\in (0,1]$, then $tE+E^\perp$ is coherent to $tE+cE^\perp$.
Consequently,
\[
\varphi(tE+E^\perp)\in \ell_{\varphi(E^\perp), \varphi(I)}\cap \cC_{\varphi(A)} =\ell\cap \cC_{\varphi(A)}=\{E_{11}\}.
\]  
Hence 
\begin{equation}\label{apor}
\varphi(tE+E^\perp)=E_{11}.
\end{equation}
By Lemma \ref{0XI}, every point of $[0, I]\setminus \ell_{0, E^\perp}$ is coherent to $tE+E^\perp$ for some $t\in (0,1]$. 
Since $\mathcal{S}_{0, A}\setminus\{B\}\subset [0,I]\setminus \ell_{0, E^\perp}$, \eqref{apor} implies 
\[
\varphi(\mathcal{S}_{0, A}\setminus\{B\})\subset \mathcal{S}_{\varphi(0), \varphi(A)}\cap \cC_{E_{11}}
=\{ E_{11} \}.
\]
Thus we obtain $\varphi(\mathcal{S}_{0, A}\setminus\{B\})=\{E_{11}\}$, contradicting our assumption that $\varphi_1$ satisfies (ii). 
\end{proof}

From now on, let us also consider the restriction $\varphi_2$ of $\varphi$ to $[B, I]$. 
By imitating the above argument, one may complete the proof of Proposition \ref{E11} whenever we assume that $\varphi_2$ satisfies (o),  (i), or (ii). 
Therefore, let us consider the case where $\varphi_1$ satisfies (iii) or (iv) and $\varphi_2$ also satisfies (iii) or (iv).  
Since $E, B\in \mathcal{S}_{0, A}$ and $E_{11}=\varphi(E)\neq \varphi(B)$, we have two possibilities: 
Either $\varphi(\mathcal{S}_{0, A}\setminus\{E\})=\{\varphi(B)\}$, or $\varphi(\mathcal{S}_{0, A}\setminus\{B\})=\{\varphi(E)\}=\{E_{11}\}$.
Similarly, either  $\varphi(\mathcal{S}_{B, I}\setminus\{{E^\perp}\})=\{\varphi(A)\}$, or $\varphi(\mathcal{S}_{B, I}\setminus\{A\})=\{\varphi(E^\perp)\}=\{0\}$ holds.

\begin{claim}
If $\varphi(\mathcal{S}_{0, A}\setminus\{E\})=\{\varphi(B)\}$, then $\varphi([0, I]\setminus\ell_{E, I})\subset \ell_{0, \varphi(B)}$, hence $\varphi$ is of type $(\ell)$.
\end{claim} 
\begin{proof}
Suppose that $\varphi(\mathcal{S}_{0, A}\setminus\{E\})=\{\varphi(B)\}$.
Let $t\in [0,1)$. 
By Lemma \ref{c(s)} and the relation $tE+E^\perp\not\sim E$, we see that $tE+E^\perp$ is coherent to some point in $\mathcal{S}_{0, A}\setminus\{E\}$. 
This together with $E^\perp\sim tE+E^\perp\sim I$ and 
\[
\varphi(E^\perp)=0,\ \ \varphi(I)=E_{11},\ \  \varphi(\mathcal{S}_{0, A}\setminus\{E\})=\{\varphi(B)\}
\]
shows that $\varphi(tE+E^\perp) \in \ell_{0,E_{11}}\cap \ell_{0, \varphi(B)} =\{0\}$
and hence
$\varphi(tE+E^\perp)= 0$.

Observe that $tE+cE^\perp$ is coherent to the three points $A, B$, and $tE+E^\perp$.
It follows that 
\[
\varphi (tE+cE^\perp) \in \cC_{\varphi(A)}\cap \cC_{\varphi(B)}\cap\cC_{\varphi(tE+cE^\perp)} =\ell_{\varphi(A), \varphi(B)}\cap \cC_0=\{\varphi(B)\}.
\]
Since $tE$ is coherent to the three points $0, E$, and $tE+cE^\perp$, we get
\[
\varphi(tE) \in \cC_{\varphi(0)}\cap\cC_{\varphi(E)}\cap\cC_{\varphi(tE+E^\perp)} = \ell\cap \cC_{\varphi(B)} =\{0\}.
\]
By Lemma \ref{0XI}, every point of $[0, I]\setminus\ell_{E, I}$ is coherent to $tE$ for some $t\in [0,1)$. 
Thus we get $\varphi([0, I]\setminus\ell_{E, I})\subset \cC_0$.

To complete the proof, we show that $\varphi([0, I]\setminus\ell_{E, I})\subset \cC_{\varphi(B)}$. 
It suffices to verify that every point of $[0, I]\setminus\ell_{E, I}$ is coherent to some point of 
\[
(\mathcal{S}_{0, A}\setminus\{E\})\cup \{tE+cE^\perp\,:\, t\in [0,1)\}.
\]
Lemma \ref{0XI} implies that every element of $[0,A)=[0, E + cE^\perp)$ is coherent to $tE+cE^\perp$ for some $t\in [0,1)$.
So, let $X \in [0, I]\setminus(\ell_{E, I}\cup [0, A))$. 
Lemma \ref{dupc} implies that $X\not\sim E$. 
Because $X\not> A$ and $X \ge 0$,  Lemma \ref{c(s)} yields that $X$ is coherent to some point in  $\mathcal{S}_{0, A}\setminus\{E\}$. 
Thus we get the desired conclusion. 
\end{proof}

Essentially the same argument shows that $\varphi$ is of type $(\ell)$ whenever $\varphi(\mathcal{S}_{B, I}\setminus\{{E^\perp}\})=\{\varphi(A)\}$. 
In what follows, we consider the case where 
\[
\varphi(\mathcal{S}_{0, A}\setminus\{B\})=\{\varphi(E)\}=\{E_{11}\}\ \ \text{and}\ \ \varphi(\mathcal{S}_{B, I}\setminus\{A\})=\{\varphi(E^\perp)\}=\{0\}
\]
hold.

\begin{claim}
If $\varphi_1$ satisfies (iii), then $\varphi_2$ also satisfies (iii), and 
$\varphi([0, I]\setminus\{B\})\subset \cC_{E_{11}}$.
Thus $\varphi$ is of type $(\mathcal{C})$.
\end{claim}
\begin{proof}
Suppose that $\varphi_1$ satisfies (iii).
Then we have 
\[
\varphi([0,A]\setminus\ell_{0,B})\subset \ell_{\varphi(E), \varphi(A)}=\ell_{E_{11}, \varphi(A)}.
\]
Let $t\in (0,1]$.
Since the set $\varphi(\mathcal{C}_B\cap [0, A]\setminus \ell_{0, B})$ is a singleton, 
we get $\varphi(tE+cE^\perp)=\varphi(E+cE^\perp)=\varphi(A)$.
Observe that $tE+E^\perp$ is coherent to three points $tE+cE^\perp, E^\perp, I$. 
Since 
\[
\varphi(tE+cE^\perp) =\varphi(A), \ \ \varphi(E^\perp)=0, \ \ \text{and} \ \ \varphi(I)=E_{11}, 
\]
we get $\varphi(tE+E^\perp)\in\ell_{E_{11}, \varphi (A)}\cap\cC_0=\{E_{11}\}$.

We now prove that $\varphi_2$ does not satisfy (iv). 
Assume on the contrary that $\varphi_2$ satisfies (iv). 
Then we have 
\[
\varphi([B,I]\setminus \ell_{A, I}) \subset \ell_{\varphi(E^\perp), \varphi(B)} =\ell_{0, \varphi(B)}.
\]
However, we already know that $\varphi((1/2)E+cE^\perp)=\varphi(A)\notin \ell_{0, \varphi(B)}$ although $(1/2)E+cE^\perp\in [B,I]\setminus \ell_{A, I}$, so we obtain a contradiction. 
Thus $\varphi_2$ satisfies (iii). 
It follows that for every $t\in [0,1]$, the set $\varphi (\mathcal{C}_{tE+cE^\perp} \cap [B,I] \setminus \ell_{B,A})$ is a singleton.
This together with Lemma \ref{0XI} and the fact that $\varphi(tE+E^\perp)=E_{11}$ for every $t\in (0,1]$ yields $\varphi((B, I])=\{E_{11}\}$. 

We show that $\varphi([0, I]\setminus\{B\})\subset \cC_{E_{11}}$. 
By Lemma \ref{0XI}, we see that every element of $[0,I]\setminus\ell_{0,E^\perp}$ is coherent to $tE+E^\perp$ for some $t\in (0,1]$. Thus $\varphi([0,I]\setminus\ell_{0,E^\perp}) \subset \cC_{E_{11}}$.
If $t\in (c,1]$, then $tE^\perp \sim E+tE^\perp \in (B, I]$, thus $\varphi(tE^\perp)\in \cC_{E_{11}}$. 
It remains to show that $\varphi(tE^\perp)\sim E_{11}$ for every $t\in [0,c)$.
Assume towards a contradiction that $\varphi(tE^\perp)\not\sim E_{11}$ for some $t\in [0,c)$. 
Since $\varphi$ restricted to $[B,I]$ satisfies (iii) and does not satisfy (iv), Lemma \ref{ABC2} implies that $\varphi$ restricted to $[tE^\perp, I]$ needs to satisfy (iii). 
In particular, we see that $\varphi((tE^\perp, I])$ is necessarily contained in a line. 
However, it is easily seen that $\{I\}\cup({\mathcal{S}}_{B, I}\setminus\{E^\perp\}) \subset (tE^\perp, I]$ and
that $\varphi(\{I\}\cup({\mathcal{S}}_{B, I}\setminus\{E^\perp\}))=\{E_{11},\varphi(A), 0\}$ is not contained in a line. 
Thus we get to a contradiction.
\end{proof}

Similarly,  if $\varphi_2$ satisfies (iv), then so does $\varphi_1$, and $\varphi([0, I]\setminus\{A\})\subset \cC_0$.
Let us finish the proof of Proposition \ref{E11} by considering the remaining case.

\begin{claim}
If $\varphi_1$ satisfies (iv) and $\varphi_2$ satisfies (iii), then $\varphi([0, I]\setminus\ell_{A, B})\subset\ell$, so $\varphi$ is of type $(\ell)$.
\end{claim}
\begin{proof}
From the assumption, we obtain $\varphi ([0, A] \setminus \ell_{A,B}) \subset \ell$ and $\varphi ([B, I] \setminus \ell_{A,B}) \subset \ell$.
We also see that $\varphi(tE^\perp)=\{0\}$ for every $t\in [0,c)$ and $\varphi(E+tE^\perp)=\{E_{11}\}$ for every $t\in (c, 1]$.

We see that each element of $[0, I]\setminus[0, A]$ is coherent to $E+tE^\perp$ for some $t\in (c, 1]$.
Indeed, if $X\in \cC_A\cap[0, I]\setminus[0, A]$, then $X=A+R$ for some $R$ of rank at most one, and $X \not\le A$ implies $R \ge 0$, which yields $A\leq X\leq I$ and hence $X\sim I$.
If $X\in [0, I]\setminus([0, A]\cup \cC_A)$, then $\det (X-A) \not= 0$. Because neither $X <A$ nor $X > A$, we have $\det (X-A) < 0$. This together with $\det (X-I)\geq 0$ and 
the intermediate value theorem shows the existence of $t\in (c, 1]$ with $X\sim E+tE^\perp$.
Similarly, each element of $[0, I]\setminus[B, I]$ is coherent to $tE^\perp$ for some $t\in [0,c)$.
Thus we obtain $\varphi([0, I]\setminus([0, A]\cup[B, I]))\subset \ell$. 
Therefore, we get $\varphi([0, I]\setminus\ell_{A, B})\subset\ell$.
\end{proof}

\begin{corollary}\label{d=1}
Let $A, B\in H_2$ satisfy $A<B$. 
Let $\varphi\colon [A, B]\to \oH$ be a coherency preserving map with $d(\varphi(A),\varphi(B))=1$. 
Then $\varphi$ is degenerate.
\end{corollary}

Theorem B is the union of Corollaries \ref{d=2}, \ref{d=0}, \ref{d=1}.

\medskip

\noindent\textbf{Acknowledgements.} \quad 
The authors are grateful to the anonymous referee for the helpful comments.



\bibliographystyle{amsalpha}

\printindex

\end{document}